\documentclass[a4paper,12pt]{article}

\usepackage{amsmath}
\usepackage{amssymb}
\usepackage{amsthm}
\usepackage{color}
\usepackage{bbm}
\usepackage{graphicx}

\newtheorem{thm}{Theorem}
\newtheorem{lma}{Lemma}
\newtheorem{claim}{Claim}
\newtheorem{fact}{Fact}
\newtheorem{definition}{Definition}
\newtheorem{example}{Example}
\newtheorem{corollary}{Corollary}

\renewcommand{\d}{\boldsymbol{d}}
\renewcommand{\S}{\mathcal{S}}

\newcommand{\N}{\mathbb{N}}
\newcommand{\K}{\mathcal{K}}
\newcommand{\f}{\boldsymbol{f}}
\newcommand{\1}{\mathbbm{1}}

\newcommand{\Wu}[1]{#1}
\newcommand{\Wuu}[1]{#1}

\title{Selfishness \Wu{need} not be bad: a general proof}

\author{Zijun Wu\\
	\texttt{zijunwu1984a@163.com},\\
	Center of Traffic and Big Data (CTB),\\
	Sino-German Institute for Applied Optimization (SG-IAO),\\
	Department of Computer Science and Technology (DCST),\\
	Hefei University (HU)\\[1ex]
	 Rolf H. M{\"o}hring\\
	 \texttt{rolf.moehring@me.com},\\
	 Beijing Institute for Scientific and Engineering Computing (BISEC),\\
	 Beijing University of Technology (BJUT)\\[1ex]
	 Dachuan Xu\\
	 \texttt{xudc@bjut.edu.cn},\\
	 Beijing Institute for Scientific and Engineering Computing (BISEC),\\
	 Beijing University of Technology (BJUT)
	 }

\begin{document}
	\maketitle
	
	\begin{abstract}
		This article studies the user behavior in non-atomic congestion games. We consider non-atomic congestion games with \Wu{continuous and non-decreasing} functions \Wuu{and investigate} the limit of the price of anarchy \Wu{when} the total user
		volume approaches infinity. 
		\Wu{We deepen the knowledge on 
			 \Wu{{\em asymptotically well designed
			games}} \cite{Wu2017Selfishness}, {\em limit games} \cite{Wu2017Selfishness}, {\em scalability} \cite{Wu2017Selfishness} and {\em gaugeability} \cite{Colini2017b}
			that were recently} \Wuu{used} \Wu{in 
			the limit analyses \Wuu{of} the price of anarchy for 
			non-atomic congestion games.}
		\Wu{We develop} a unified
			\Wu{framework and derive new techniques that allow a general
			limit analysis \Wuu{of} the price of anarchy.} 
		\Wu{With these new techniques,} \Wu{we are able to
			prove a global convergence on the
			price of anarchy for non-atomic congestion
			games with arbitrary polynomial price functions
			and arbitrary user volume vector sequences,
			see Theorem~\ref{theo:SelfishMainTh}.
			This means that
			non-atomic congestion games with polynomial price
			functions are asymptotically well designed.}  
		\Wu{Moreover, we show that
			these new techniques are very flexible
			and robust \Wuu{and} apply also to 
			non-atomic congestion games with 
			price functions of other types. 
			In particular, we prove that non-atomic congestion
			games with regularly varying
			price functions are also asymptotically
			well designed, provided that
			the price functions are slightly
			restricted, see Theorem \ref{theo:MainTheoremRegularVariation}
			and Theorem \ref{theo:MainTheoremGaugeableExtension}.
			Our proofs are direct and very elementary
			without using \Wuu{any} heavy machinery.
			They only use basic properties
			of Nash equilibrium and system optimum 
			profiles, simple facts about \Wuu{the} asymptotic
			notation $O(\cdot),\Omega(\cdot),$
			etc,
			and  induction.}
		\Wu{Our results greatly} generalize
		recent results from 
		\cite{Colini2016}, \cite{Colini2017a}, \cite{Colini2017b} and
		\cite{Wu2017Selfishness}. 
		\Wu{In particular, our results
			further support the view of \cite{Wu2017Selfishness}
			with a general proof that
			selfishness need not be bad for
			non-atomic congestion games.} 
	\end{abstract}

\section{Introduction}

\subsection{Motivation}
Nowadays, traffic congestion has almost become a daily annoyance to every
citizen in large cities of China. According to the newest data from 
AMap \cite{amap2017} in 2017, more than 26\% of cities in China \Wu{experienced}
traffic congestion in rush hours, 55\% of cities \Wu{experienced} low speed, and only
19\% of cities \Wu{did} not \Wu{suffer from} traffic congestion. 

Traffic congestion does
not only considerably \Wu{enlarge} travel latency, but also causes serious economic
loss. We take the capital city of China, Beijing, as an example. The average economic loss 
caused by congestion in 2017 was about
4,013 RMB per person, see \cite{baidu2017}, 
which accounts for 3.1\% of the annual GDP of Beijing in that year. Note that the annual GDP growth of Beijing
was only about 6.8\% in 2017. This means that
traffic congestion \Wu{has} almost \Wu{destroyed} one third of the potential
economic growth of Beijing.

To alleviate problems caused by traffic congestion, 
\Wu{the government of China} has actively implemented a series of traffic management
measures in \Wu{some large} cities in recent years, including \Wu{the} even and odd license 
plate number
rule, license plate \Wu{lotteries}, encouraging public transportation and others. 
These measures definitely prevent further deterioration of traffic, but
not yet completely cure congestion. 

\Wu{Road traffic conditions are} a direct result 
from simultaneous travel of citizens 
in a particular area. Given road \Wu{conditions}, the routing behavior of travelers  
almost determines
how the traffic \Wu{develops}. Thus, to comprehensively cure congestion, a preliminary
step is to well understand the routing behavior of travelers. In particular, we
need to \Wu{find} out the extent to which the \Wu{autonomous} routing behavior of travelers contributes to congestion.
This motivates the present article. 

\subsection{\Wu{The} static model}

To that end, one needs to model road traffic
appropriately.
A popular static model for road traffic is the so-called {\em non-atomic congestion
	game} (NCG), see \cite{Rosenthal1973A}
or \cite{Nisan2007}. NCGs are non-cooperative games of perfect information. In an NCG, users (players) are collected
into $K$ different groups according to some measurement on their similarities, for a \Wu{fixed} integer $K\in \N_+$. Associated
with each group $k\in \K:=\{1,\ldots,K\}$ \Wu{is} a finite non-empty \Wu{set $\S_k$
containing all} strategies {\em only} available to users from \Wuu{group $k$}. 
Every user engaged in the game \Wu{chooses} a strategy $s\in \S:=\bigcup_{k\in \K}\S_k$ \Wu{that he will follow}, and 
every \Wu{chosen} strategy $s\in \S$  \Wu{
 consumes}
$r(a,s)$ units of resource $a$ for each $a\in A.$ \Wu{Here,} $A$ is a 
finite non-empty set \Wu{containing} all available resources, and $r(a,s)$
is a \Wu{fixed} non-negative constant
denoting
the consumed (or demanded) volume
of \Wuu{resource $a$ by strategy
$s$ for each $a\in A$ and
each $s\in \S.$} The eventual price of a resource $a\in A$ 
depends {\em only} on its consumed volume. Given a vector $d=(d_k)_{k\in \K}$
of user volumes, a {\em feasible strategy profile} $f$ is  an assignment \Wu{that} assigns to
each of the $d_k$ users from group $k\in \K$ a feasible strategy $s\in \S_k$ for each group $k\in \K.$ 
See \cite{Wu2017Selfishness} or Section \ref{sec:Model} for details. 

Obviously, NCGs model road traffic \Wu{on} an macroscopic level. \Wuu{Then}, resources $a\in A$ will be arcs (streets) \Wu{of} the
underlying road network, a group $k\in \K$ will be a travel origin-destination
(OD) pair, and a strategy $s\in \S_k$ will be \Wu{a path}
from the $k$-th origin to the $k$-th destination, for each
$k\in \K$. The constant 
$r(a,s)$ is just an indicator
function of the \Wu{membership} relation ``$a\in s$" for each arc $a\in A$ and each path $s\in \S.$ \Wuu{A feasible strategy profile
$f$ is then} a feasible {\em traffic (path) assignment} \cite{Dafermos1969} for all the $T(d):=\sum_{k\in \K}d_k$ travelers.
\Wu{As our results hold on a more general level}, we will not stick to these terminologies of
road traffic in the sequel, \Wu{but still be able} to understand the 
routing behavior of travelers.

In an NCG, the price of a resource $a\in A$ is often expressed as 
an {\em non-negative, non-decreasing} and {\em continuous} function $\tau_a(\cdot)$ of 
its demanded volume, see, e.g.,
\cite{Roughgarden2000How}, \cite{Roughgarden2001Designing}, \cite{Roughgarden2005Selfish}, \cite{Nisan2007}, \cite{Roughgarden2015Intrinsic}.
Popular price functions \Wu{are} {\em polynomials}. For instance, \Wu{latency functions}
$\tau_a(\cdot)$ in road traffic are conventionally assumed to be 
{\em Bureau of Public Road} (BPR) functions \cite{BPR}, which
are polynomials of degree $4.$ In our study, we will follow this fashion, and
\Wu{emphasize 
	on polynomial price functions $\tau_a(\cdot)$ 
and others that are related to polynomials, \Wuu{e.g.,} 
regularly varying functions \cite{Bingham1987Regular}.}
However, \Wu{the polynomials we will consider are general, i.e., they are allowed to
	\Wuu{have} different degrees.}
\Wu{We} do not consider 
	strategies \Wuu{that}
are {\em completely free,} i.e., 
$\sum_{a\in A} r(a,s)\cdot \tau_a(x)\equiv 0$ for all
$x\ge 0,$ for some group
$k\in \K$ and some strategy $s\in \S_k.$ 
This is rather reasonable \Wuu{in} congestion game, since
users with choices of free strategies are actually
outside the underlying game!

\subsection{\Wu{Selfish user} behavior}

\Wu{NCGs are non-cooperative, and so} users  are \Wu{considerd to be} selfish. They would like to use
strategies minimizing their own \Wu{cost}.
For instance, travelers would like to
follow \Wu{a quickest} \Wuu{path,} so as
to reduce their travel latency. In general,
the cost of a user is just the cost
of the strategy adopted by that user.
Given a 
vector $d=(d_k)_{k\in \K}$ of user volumes
and a feasible profile $f,$
the cost of a strategy $s\in \S$ equals
$\sum_{a\in A}r(a,s)\cdot\tau_a(f_a),$
where $f_a$ denotes the total 
consumed volume of resource $a$
w.r.t. profile $f$,
for each $a\in A.$ Obviously,
for road traffic, 
the cost of a strategy $s\in \S$ 
is just the 
total travel time \Wu{(latency) along} path $s.$

The selfish behavior of users will eventually lead the underlying game into a 
so-called {\em Wardrop equalibrium}
(WE) \cite{Wardrop1952ROAD}, in which
every user follows a cheapest strategy
\Wuu{he} could follow \Wu{given the choices of
	all other users}, see Section \ref{sec:Model} for details. 
Under our assumption
of \Wu{continuous and non-decreasing} price functions
$\tau_a(\cdot),$ a WE is actually
a {\em pure Nash equilibrium} (NE) \cite{Rosenthal1973A}, in which all users
will loyally adhere to the choices they have done, since no unilateral change
in strategy can introduce any extra profit. Therefore, the game will
enter a steady state if no external force interferes. An NE
(equivalently, a WE)
is thus \Wu{a} macro model
of user (selfish) behavior.
The average cost of users \Wu{in} an
NE profile therefore reflects 
\Wu{the cost} users need to pay in practice.
Note that under our setting of continuous and non-decreasing
price functions
$\tau_a(\cdot),$ all NE profiles will have the same
\Wu{cost}, see, e.g., \cite{Smith1979The}.

\Wu{A question of great 
interest in NCGs is} whether user
(selfish) behavior is harmful. This
actually concerns whether the selfish behavior of users will
demage {\em social welfare}, i.e., \Wu{increases} the average cost of users
engaged in the game. If this is the case, then we may need to
employ
some external
\Wu{forces or measures} to break up the equilibrium induced by
the selfish behavior, so as to \Wu{get closer to} social welfare.
For road traffic, possible external \Wu{measures} could be
some road guidance policies like congestion pricing,
see, e.g.,
\cite{Cole2003Pricing}, \cite{Fleischer2004Tolls}, \cite{Cole2006How}, \cite{Phang2004Road} and \cite{Harks2015Computing}.

The
price of anarchy (PoA), a concept \Wu{stemming} from \cite{Papadimitriou2001Algorithms}, is
a popular measure for
the ``\Wuu{defficiency}" of \Wu{selfish user}
behavior.
Given a vector
$d=(d_k)_{k\in \K}$ of user
volumes, a feasible strategy 
profile $f^*$
is said to be at {\em system optimum}
(SO) if it minimizes
the average cost of 
users engaged in the game. 
\Wu{The value of the PoA for \Wuu{non-decreasing}
	price functions $\tau_a(\cdot)$}
equals the {\em ratio of the average cost of users in an NE profile over that in an SO
profile}. Obviously, the larger the
value of PoA, the \Wu{more} \Wu{selfish user} behavior \Wuu{demages}
the social welfare.

The value of the PoA does not only reflect the 
extent to which
the \Wu{selfish user} behavior ruins social welfare, but also
the potential benifit we could get if all users
were appropriately guided. \Wuu{For road traffic,} the value of the PoA \Wuu{thus} shows the
extent to which the average latency \Wuu{can} be reduced in principle, if
all travelers use the ``right" paths.
Therefore, \Wu{for} our purpose, we need a close inspection \Wu{of} the value of the PoA. 

\subsection{The state of the art}
Traditionally, \Wu{selfish user} behavior
is considered to be harmful, see, e.g.,
the studies in
\cite{Roughgarden2000How},
\cite{Roughgarden2001Designing},
\cite{Roughgarden2002The}, \cite{Roughgarden2004Bounding},
\cite{Roughgarden2005Selfish},
\cite{Correa2005On}, \cite{Roughgarden2007Introduction},
and \cite{Roughgarden2015Intrinsic}.
These studies investigated \Wu{an} upper bound of the PoA 
for several classes of price functions.
They demonstrated that the worst-case upper bound \Wu{can} be very large.
A famous example \Wu{motivating} these studies is
Pigou's game, see, e.g., \cite{Nisan2007} or 
Figure \ref{fig:pigou}.
\begin{figure}[!htb]
	\centering
	\includegraphics[scale=0.3]{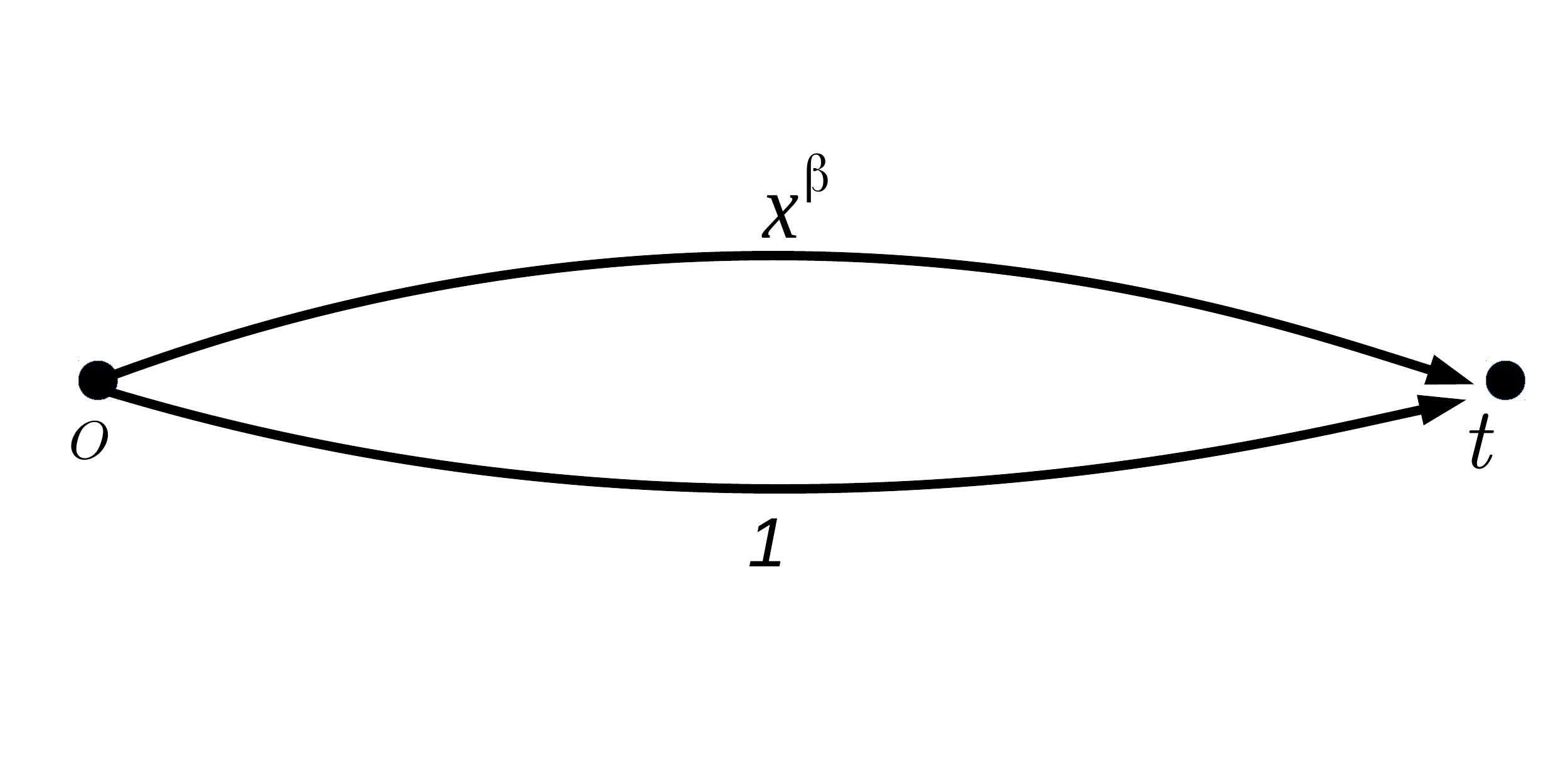}
	\caption{Pigou's game}
	\label{fig:pigou}
\end{figure}
In Pigou's game, there is only one
user group with two strategies of
price functions
$x^{\beta}$ and $1,$ respectively,
for some constant $\beta\ge 0.$ The PoA
of this game equals $T/\big(T-(\beta+1)^{-1/\beta}+(\beta +1)^{-1}\big),$ where
$T\ge 1$ is the volume of users engaged in the game. Obviously,
considering all possible $\beta,$ the worst-case
upper bound of the PoA is infinity, since the PoA
tends to $\infty$ as
$\beta\to \infty,$ if $T=1.$ 

\Wu{Worst-case upper bounds} of the PoA \Wu{are} actually
not a valid evidence to show that \Wu{selfish user} behavior 
is bad, especially \Wu{not} for a large
total volume of users. For instance, the PoA of Pigou's game
actually tends to $1,$ as $T$ approaches infinity, for
each \Wu{fixed} $\beta\ge 0.$
This means that \Wu{selfish user} behavior \Wu{may} well guarantee social
welfare if \Wu{the} total user volume is large. 
Generally, the total travel demand in rush hours
is usually very large. Therefore, to comprehensively understand
selfish (routing) behavior,
a closer inspection \Wu{of} the value of the PoA
is still needed, especially for the case of heavy traffic, i.e., 
the case \Wuu{when} \Wu{the} total travel demand $T(d)=\sum_{k\in \K}d_k$ is large.

Recently, \Wu{several} studies have been done \Wu{in this direction,} 
see \cite{Colini2016}, \cite{Colini2017a},
\cite{Colini2017b}, and \cite{Wu2017Selfishness}. 
Colini et al. \cite{Colini2016} considered two special cases: (i) \Wu{games with a} single
OD pair, and (ii) \Wu{games with a} single OD pair with parallel feasible paths.
For case (i), they proved that if one of the feasible paths has a 
latency function that is bounded by a constant
from above, then the PoA will converge to 1 as the
total demand $T(d)$ tends to infinity. For case (ii), they proved that if the latency functions
$\tau_a(\cdot)$ are regularly varying \cite{Bingham1987Regular},
then the PoA will converge to 1 as \Wuu{the} total demand $T(d)$ tends
to infinity. 

Colini et al. \cite{Colini2017b} continued the study of 
\cite{Colini2016}. They investigated more general cases, i.e., \Wu{games with} multiple OD pairs. \Wu{Using the notion of regular
	variation \cite{Bingham1987Regular},} they proposed the concept of \Wu{gaugeable games that
	are \Wuu{defined} only for particular user volume vector
	sequences $\{d^{(n)}\}_{n\in \N}$
	fulfilling \Wuu{the} condition \Wuu{that} $
	\varliminf_{n\to\infty}\sum_{k\in \K_{tight}}\frac{d_k^{(n)}}{T(d^{(n)})},
	$ 
	see \cite{Colini2017b} or Subsection~\ref{sec:KnownResults}
	for details.
	}
 \Wu{With the technique of
 	regular variation \cite{Bingham1987Regular}, they proved} that the PoA of
gaugeable games converges to 1 as the total
demand $T(d)=\sum_{k\in \K} d_k$ tends to infinity.
However, \Wu{this convergence result only
	holds for particular user volume vector sequences, due
	to the sequence-specific nature of gaugeable games.
	See Subsection \ref{sec:KnownResults}
	for details \Wuu{about} the convergence results of gaugeable games.}


Colini
et al. \cite{Colini2017a} further continued
the study of \cite{Colini2017b}.
They \Wu{applied the technique of gaugeability
	\cite{Colini2017b} to NCGs with
	polynomial price functions.
Due to the sequence-specific nature of
 gaugeability \cite{Colini2017b},
they assumed a user volume vector sequence $\{d^{(n)}\}_{n\in \N}$ \Wuu{with} $\varliminf_{n\to \infty}\tfrac{d_k^{(n)}}{T(d^{(n)})}>0$
for each group $k\in \K.$}
 They proved that
if all $\tau_a(\cdot)$ are
polynomials, then \Wu{$\lim_{n\to\infty}\text{PoA}(d^{(n)})=1,$
provided that $\lim_{n\to\infty}T(d^{(n)})=\infty$.}
\Wu{Moreover},
Colini et al. \cite{Colini2017a} further extended the study
of \cite{Colini2017b} to light traffic, i.e., $T(d)\to 0$. 

All the results of \cite{Colini2016}, \cite{Colini2017a} and
\cite{Colini2017b} are derived \Wu{with} the technique of regular variation
\cite{Bingham1987Regular}. A different study was done by 
Wu et al. \cite{Wu2017Selfishness}. They
assumed \Wu{arbitrary} 
user volume vector sequence \Wu{$\{d^{(n)}\}_{n\in \N}$} \Wuu{and}
aimed to explore properties of {\em asymptotic
	well designed games}, \Wu{i.e., games} in which the PoA \Wu{tends} to
$1,$ as \Wu{the} total volume $T(d)$ approaches infinity.
They proposed the concept of
{\em scalable games}, and proved that 
all scalable games are asymptotically well designed. They also
proved that gaugeable games are special cases of scalable games \Wu{w.r.t. the particular user volume vector sequences
	assumed by gaugeability.}
Moreover, they \Wu{provided} examples that are scalable, but not
gaugeable. 

In addition,
 Wu et al. \cite{Wu2017Selfishness}
\Wu{made a detailed study of} the case that all $\tau_a(\cdot)$
are BPR-functions. They proved for this particular
case that an SO profile is an {\em $\epsilon$-approximate
NE profile} \cite{Roughgarden2000How} for
a small $\epsilon>0$, and 
the PoA equals $1+O\big(T(d)^{-\gamma}\big),$
where $\gamma$ is the common degree of the BPR-functions.
This proves a conjecture proposed by O'Hare et al. \cite{O2016Mechanisms}. \Wu{They} also proved for this particular case
that the cost of both, an NE-profile and
an SO-profile, can be \Wu{asymptotically} approximated by $L(\boldsymbol{d})\cdot
 T(d)^{\gamma},$ where
$L(\d)$ is a computable constant that only depends on 
distribution $\d:=\big(d_k/T(d)\big)_{k\in \K}$ of users
among groups,
when \Wu{the} total volume $T(d)$ is large enough.
However, Wu et al. \cite{Wu2017Selfishness} \Wu{still}
failed to show that NCGs with general polynomials
are asymptotically well designed.

In summary, \cite{Colini2016}, \cite{Colini2017a}, 
\cite{Colini2017b} and \cite{Wu2017Selfishness} 
definitely show that selfish behavior \Wu{need} not be bad
for the case of a large $T(d)$ \Wu{under certain conditions.} 
However, one important question \Wu{has remained} open, namely,
whether NCGs with \Wu{arbitrary} polynomial price functions $\tau_a(\cdot)$
are asymptotically well designed. Although \cite{Colini2017a} and \cite{Colini2017b} have preliminary \Wuu{results} towards this question,
\Wu{their results only partially answer this
	question due to the sequence-specific nature
	\Wuu{of} their study.}

\Wuu{Polynomial functions are so popular because} \Wu{they
usually serve} as a \Wu{first} prototype to understand quantitative 
relations between variables. Price functions $\tau_a(\cdot)$ 
are key components of an NCG \Wu{and model} the quantitative relations between resource
prices and demanded volumes. NCGs with polynomial
price functions $\tau_a(\cdot)$ are thus of great importance in practice.
The open question \Wuu{about the PoA
	for polynomials thus} 
concerns properties of the \Wu{selfish user} behavior in such games
for a large user volume $T(d)$
\Wuu{and is thus} of great \Wuu{interest} to our purpose
\Wu{of understanding heavy traffic}. 

\subsection{Our main result}

We will continue 
the study of \cite{Wu2017Selfishness}.
\Wu{However, to better understand  the state of
	the art, we will first \Wuu{discuss} the 
	notions of scalability \cite{Wu2017Selfishness} and 
	gaugeability \cite{Colini2017b}, and give a detailed
	description \Wuu{of} the known results from
	\cite{Wu2017Selfishness}, \cite{Colini2017a} and
	\cite{Colini2017b}, see Subsection \ref{sec:KnownResults}.
	}
	
	\Wu{We then apply the concepts
		of scalability and limit games stemming from \cite{Wu2017Selfishness}.
		We first prove that if the limit game
		exists, then an NCG is asymptotically well designed
		if and only if \Wuu{it} is scalable,
		see Theorem~\ref{theo:NegativeConditionforAWDG}.
		This deepens the knowledge about scalability and asymptotically well designed games.
		\Wuu{For an even deeper understanding},
		we adapt some algebraic ideas to our analysis
		and consider decompositions of NCGs.
		We prove that the class of asymptotically well designed
		games is actually closed under direct \Wuu{sums},
		see Corollary \ref{theo:DirectSum}.}
	\Wuu{This demonstrates \Wuu{in a certain sense the} extent 
		of the notion of asymptotically designed games.}
	
	\Wu{To obtain a general proof for
		the convergence of the PoA, we
		develop a new technique called {\em asymptotic decomposition.}
		This technique generalizes
		the idea of direct \Wuu{sums}, and is designed for
		\Wuu{handling general NCGs in the limit analysis.} 
		We are able to demonstrate its power
		by \Wuu{applying it to} 
		NCGs with arbitrary polynomial price functions
		and NCGs with regularly varying
		price functions, see Theorem~\ref{theo:SelfishMainTh},
		Theorem~\ref{theo:MainTheoremRegularVariation} and
		Theorem~\ref{theo:MainTheoremGaugeableExtension}.
		}
	
\Wu{With \Wuu{the} asymptotic decomposition,} we are able
to prove that \Wu{NCGs with arbitrary
	polynomial price functions are
	asymptotically well designed}, 
see Theorem \ref{theo:SelfishMainTh}.
\Wu{This \Wuu{
	completely solves the convergence of} the PoA of NCGs
	with polynomial price functions,
	and thus the aforementioned open question.}
\Wu{This result greatly extends the findings of \cite{Wu2017Selfishness},
\cite{Colini2016}, \cite{Colini2017b} and
\cite{Colini2017a} for road traffic, and 
deepens the understanding that selfishness 
need not be bad, and might be the best choice in a \Wuu{bad}
environment. Moreover, this result also indicates that
selfish routing is actually not the main cause of congestion, when
the total travel demand $T(d)$ is large. In particular, if
the total travel demand stays high, then we \Wuu{cannot}
significantly reduce the average travel latency by any road guidance
policies. }


\Wu{Theorem \ref{theo:SelfishMainTh}} also 
brings some insight \Wuu{into} free market economics. In 
\Wu{market} economics, resources correspond
to factors of production, \Wuu{groups} correspond
to \Wuu{sets} of suppliers manufacturing a 
particular type of product, and 
resource prices $\tau_a(\cdot)$
are the purchasing prices of those production
factors.  In a free market system,
the prices $\tau_a(\cdot)$ of production factors are completely determined by the demanded volumes, and are often assumed to be polynomial functions. \Wu{Theorem \ref{theo:SelfishMainTh} then shows}
that given the demand of each \Wuu{product type,} the free market \Wu{will} autonomously minimize
the average manufacturing cost, when the total number
of suppliers is large. 

\Wu{Asymptotic decomposition also applies to
	NCGs with price \Wuu{functions} of other types.
	To \Wuu{demonstrate} this, we also applied this
	technique to NCGs with regularly
	varying price functions.
	The result shows that
	these NCGs are in general also
	asymptotically well designed,
	see Theorem \ref{theo:MainTheoremRegularVariation}.
	In particular,
	with this technique, we
	are able to remove the sequence limitation
	for gaugeable games and generalize
	the main result Theorem 4.4 in \cite{Colini2017b}
	for gaugeability, see Theorem~\ref{theo:MainTheoremGaugeableExtension}.}

\Wu{Theorem~\ref{theo:SelfishMainTh},
	Theorem~\ref{theo:MainTheoremRegularVariation} and
	Theorem~\ref{theo:MainTheoremGaugeableExtension}
	definitely demonstrate
	the power of asymptotic decomposition.
	\Wuu{They assume an arbitrary user volume vector
		sequence, and thus the results \Wuu{hold globally}.}
	In particular, \Wuu{together they constitute} a general proof
	that selfishness need not be bad for NCGs.
	\Wuu{Their proofs} are direct and very elementary
	without using any heavy machinery,
	\Wuu{and} only use some basic properties
	of Nash equilibrium and system optimum 
	profiles, simple facts about asymptotic
	notation $O(\cdot),\Omega(\cdot),$
	etc,
	and a \Wuu{simple induction \Wuu{over} the group set $\K$}.}

\subsection{The structure of the paper}
The \Wuu{remainder} of the article is arranged as \Wu{follows:}
Section \ref{sec:Model} gives \Wu{a} detailed description of
the NCG model \Wu{that} we will study. Section \ref{sec:MainResults} \Wu{gives a detailed description
	on known results and} \Wuu{then}
reports our results.
Section \ref{sec:Conclusion} gives a brief summary of the whole
article. To improve readability, we \Wuu{move} the \Wu{elementary but long}
proofs of \Wuu{our} results \Wuu{to} the Appendix.

\section{The Model}
\label{sec:Model}
In our study, we follow the formulation of NCGs \Wu{in}
\cite{Wu2017Selfishness}. \Wu{This} formulation is slightly
different from the traditional model commonly used in
the literature, see, e.g., \cite{Nisan2007}. Traditionally,
a strategy is assumed to be a subset of resources. 
In our study, we employ a constant $r(a,s)\ge 0$ to 
refelct the relation between a resource $a\in A$
and a strategy $s\in \S.$ This \Wu{slightly generalizes our results.}

\Wu{An NCG is} represented by a tuple
\[
\Gamma=\big(\K, A, \S, (r(a,s))_{a\in A, s\in \S}, (\tau_a)_{a\in A},d\big),
\]
where:
\begin{itemize}
	\item $\K$ is a finite non-empty set of groups. We assume, w.l.o.g.,
	that $\K=\{1,\ldots,K\},$ i.e., there \Wu{are} 
	$K$ groups of users.
	
	\item $A$ is a finite non-empty set of resources that will be demanded by 
	users engaged in the game. 
	\item $\S=\bigcup_{k\in \K}\S_k$ is a finite non-empty set of available
	strategies. \Wu{Herein,} each $\S_k$ \Wu{contains} all strategies available
	to users in group $k$ for each $k\in \K.$ We assume that
	$\S_k\cap \S_{k'}=\emptyset,$ provided that $k\ne k',$
	for each $k,k'\in \K.$ 
	
	\item $r(a,s)\ge 0$ denotes the demanded volume of resource $a$
	by a user adopting strategy $s,$ for each $a\in A$ and 
	each $s\in \S.$
	\item $\tau_a:[0,+\infty)\mapsto [0,+\infty)$ denotes
	the {\em price function} of resource $a$ for each
	$a\in A.$ We assume that each $\tau_a(x)$ is a \Wu{continuous} function that is non-negative and non-decreasing for all $x\ge 0,$ for all $a\in A$.
	\item $d=(d_k)_{k\in \K}$ is a non-negative {\em user volume vector},
	where each component $d_k\ge 0$ represents the volume of users belonging
	to group $k\in \K.$
\end{itemize}
In our study, we assume further that for each group 
$k\in \K$ and each strategy $s\in \S_k,$
\begin{equation}\label{eq:NoFreeStrategies}
	\sum_{a\in \S_k} r(a,s)\cdot\tau_a(x)\not\equiv 0.
\end{equation}
Note that \eqref{eq:NoFreeStrategies} is a reasonable assumption.
Otherwise, there would be a group $k\in \K,$ \Wu{whose users}
can consume resources without paying any \Wu{price}. This
actually conflicts \Wu{with} the spirit of congestion games in practice.  

In an NCG, users usually want to adopt strategies
that \Wu{minimize} their own \Wuu{cost}. However, the cost of
a user does not only depend on his/her
choice, but
also \Wu{on the choices of} other users, i.e., 
the cost of a user is eventually determined by the strategy
profile formed by all users engaged in the game.
Herein,
a {\em feasible strategy profile} $f$
can be represented by a vector $f=(f_s)_{s\in \S},$
where:
\begin{itemize}
	\item[p1)] Each component $f_s\ge 0$ represents the total volume of 
	users adopting strategy $s,$ for each strategy $s\in \S.$ 
	\item[p2)] $\sum_{s\in \S_k}f_s=d_k,$ for each group $k\in \K,$ which indicates that every user must choose a strategy
	to follow.
\end{itemize}

Consider a feasible strategy profile $f=(f_s)_{s\in \S}.$
The {\em demanded volume (or consumed volume)} of each resource $a\in A$ \Wu{w.r.t} profile
$f,$ denoted by $f_a,$ can be computed \Wu{as}
\[
f_a=\sum_{s\in \S} r(a,s)\cdot f_s.
\]
Thus, the price of a resource 
$a\in A$ w.r.t. profile $f$ equals
$\tau_a(f_a).$ Therefore,
the cost of a strategy $s\in \S$ w.r.t. 
profile $f,$ denoted by $\tau_s(f),$ equals
\[
\tau_s(f)=\sum_{a\in A} r(a,s)\cdot \tau_a(f_a).
\]
Then, the {\em average cost} of users w.r.t.
profile $f$ equals
\[
C(f):=\frac{1}{T(d)}\cdot \sum_{s\in \S} f_s\cdot \tau_s(f)=
\frac{1}{T(d)}\cdot \sum_{a\in A} f_a\cdot \tau_a(f_a),
\]
where $T(d)=\sum_{k\in \K}d_k$ denotes the {\em total volume}
of users \Wu{in} the game.

The selfishness of users will autonomously lead their choices
to eventually form a feasible profile $\tilde{f}=\big(\tilde{f}_s\big)_{s\in \S}$ satisfying that
\begin{equation}\label{def:WE}
	\forall k\in \K \ \forall s,s'\in \S_k \Big(
	\tilde{f}_s>0\implies \tau_s(\title{f})\le 
	\tau_{s'}(\tilde{f})\Big),
\end{equation}
i.e., every user \Wu{chooses} a ``cheapest" strategy \Wu{he/she}
could follow w.r.t. profile $\tilde{f}$.
Such profiles \Wu{are called} Wardrop equilibria (WE, \cite{Wardrop1952ROAD}). Under our assumption \Wu{on} the
price functions
$\tau_a(\cdot),$ Wardrop equilibria coincide with
\Wu{the} {\em pure Nash equilibria} (NE). A feasible
strategy profile $f=(f_s)_{s\in \S}$ is said to be
at NE, if
\begin{equation}\label{def:NE}
	\forall k\in \K \ \forall s,s'\in \S_k\Bigg(
	f_s>0\!\implies\! \Big(\forall \epsilon (f_s>\epsilon>0)\!\implies\!\tau_s(f)\le \tau_{s'}(f^{1+\epsilon})\Big)\Bigg),
\end{equation}
where $f^{1+\epsilon}=(f_{s''}^{1+\epsilon})_{s''\in \S}$
is a feasible strategy profile that moves 
$\epsilon$ users from strategy $s$ to strategy $s',$ i.e.,
for each strategy $s''\in \S,$ 
\[
f^{1+\epsilon}_{s''}=
\begin{cases}
f_{s''}&\text{if } s''\notin\{s,s'\},\\
f_{s''}-\epsilon&\text{if } s''=s,\\
f_{s''}+\epsilon &\text{if } s''=s'.
\end{cases}
\]
In the sequel, we shall always put a tilde above a strategy
profile, if the strategy profile \Wu{is}
an NE profile (or, equivalently a Wardrop equilibrium).

An NE profile $\tilde{f}$ is \Wu{a macro} model for the selfish behavior
of users in practice. Under our assumption \Wu{on}
price functions $\tau_a(\cdot),$  all NE profiles $\tilde{f}$
have the same average cost, see e.g. \cite{Smith1979The} for a proof. Obviously, these profiles \Wu{are 
user ``optimal"} \eqref{def:WE},
 and stable \eqref{def:NE} to some extent.

Besides \Wu{NE} profiles, system optimum (SO) profiles are also of great \Wu{interest}, for the sake of achieving social welfare.
\Wu{Formally}, a feasible strategy profile $f^*=(f_s^*)_{s\in\S}$ is \Wu{an} SO profile if
it solves the following program:
\begin{equation}\label{def:SO}
	\begin{split}
	&\min\quad\quad C(f)\\
	&s.t.\\
	&\quad \sum_{s\in \S_k} f_s=d_k,\forall k\in \K,\\
	&\quad f_s\geq 0,\forall s\in \S.
	\end{split}
\end{equation}
In the sequel, we shall always \Wu{use} a star in the \Wu{superscript}
of a feasible profile \Wuu{to indicate that it} is an SO
profile.

In general, an NE profile \Wu{need not} be a solution
to the program \eqref{def:SO}. The PoA is
a popular index to show the extent to which
the \Wuu{selfish user behavior} \Wuu{destroys} 
social welfare in practice. \Wu{It} is a concept \Wu{stemming}
from \cite{Papadimitriou2001Algorithms}, \Wu{and}
can be defined as \Wu{follows}
\begin{equation}\label{eq:PoA}
	\text{PoA}:=\frac{C(\tilde{f})}{C(f^*)}
	=\frac{\sum_{s\in \S} \tilde{f}_s\cdot \tau_s(\tilde{f})}{\sum_{s\in \S} f^*_s\cdot \tau_s(f^*)},
\end{equation}
where $\tilde{f}$ is an NE profile, and
$f^*$ is an SO profile. 

As mentioned, we will investigate the limit
of the PoA \Wu{when} the total volume $T(d)=\sum_{k\in \K}d_k$ approaches infinity. Therefore, to avoid ambiguity, we shall denote
by PoA$(d)$ the corresponding PoA \Wu{for} user volume
vector $d=(d_k)_{k\in \K}$ in the sequel.
 
\section{Selfishness need not be bad: a general discussion}
\label{sec:MainResults}

In this Section, we consider the limit of the PoA under
our assumption of \Wu{continuous, non-decreasing and non-negative}
price functions $\tau_a(\cdot)$. \Wu{In particular, we
	will emphasize on \Wuu{the} polynomial functions and
	regularly varying functions that have been recently
	studied in \cite{Wu2017Selfishness}, \cite{Colini2016},
	\cite{Colini2017a} and \cite{Colini2017b}.}
To better understand our
result, we first introduce
some relevant concepts and results from \cite{Wu2017Selfishness} and \cite{Colini2017b}. 

\subsection{\Wu{Scalability and gaugeability}}
\label{sec:KnownResults}

NCGs are static models for decision-making behavior of selfish users (players)
in systems with limited resources.
\Wu{Designing an NCG such that the selfish choices of users
autonomously optimize social welfare is in general not easy,} see e.g.
\cite{Roughgarden2001Designing}. However, \Wu{such games exist,} see
\cite{Wu2017Selfishness}.
\begin{definition}[See \cite{Wu2017Selfishness}]\label{def:WDG}
	An NCG $\Gamma$ is said to be a {\em well designed game} (WDG), if  \text{PoA}$(d)=1$ for \Wu{each} given user volume vector $d=(d_k)_{k\in \K}$ with total user volume $T(d)>0.$
\end{definition}
Obviously, selfish behavior of users in a WDG
should be strongly favored, \Wu{as it leads} the underlying
system into a steady state with minimum average cost. 
\Wu{Readers may refer to \cite{Wu2017Selfishness} for examples
of WDGs.}

As mentioned, WDGs are often too restrictive \Wu{for designing them} 
in practice. \Wu{Nevertheless, an important goal in
NCGs concerns} how to effectively allocate
limited resources to a large volume of
users. Therefore, an alternative choice \Wu{to designing a WDG} is to design
an NCG that \Wu{will} approximate a WDG when the total user volume
$T(d)$ becomes large. This inspires the 
concept of \Wu{an} {\em asymptotically well designed game} (AWDG) \cite{Wu2017Selfishness}.
\begin{definition}[see \cite{Wu2017Selfishness}]\label{def:AWDG}
	An NCG $\Gamma$ is said to be an AWDG, if 
	the PoA$(d)$ converges to $1$ as 
	$T(d)$ approaches infinity. \Wu{For later use, we also denote the class of all AWDGs as AWDG.}
\end{definition}
Scalable games \Wu{introduced} by
Wu \Wu{et al.} \cite{Wu2017Selfishness} 
are examples of AWDGs. These games
require the existence of a well designed {\em limit} game
for each sequence 
\Wuu{$\{d^{(n)}\}_{n\in \N}$} of user volume vectors
with 
\[
\lim_{n\to \infty}T(d^{(n)})=\sum_{k\in \K}d_k^{(n)} =\infty,
\]
where $d_k^{(n)}$ is the $k$-th component
of $d^{(n)}=\big(d_k^{(n)}\big)_{k\in \K}$ \Wu{and denote} the user volume
in group $k$ for each $k\in \K$ \Wu{and each} $n\in \N.$ 
\begin{definition}\label{def:limitGame}
	Given a sequence \Wuu{$\big\{d^{(n)}=(d_k^{(n)})_{k\in \K}\big\}_{n\in \N}$} of 
	user volume vectors with $\lim_{n\to \infty} T(d^{(n)})=\lim_{n\to \infty} \sum_{k\in \K} d_k^{(n)}=\infty,$
	an NCG 
	\[
	\Gamma^{\infty}=\big(\K^{\infty}, A, \S^{\infty}, (r(a,s))_{a\in A, s\in \S^{\infty}}, (\tau^{\infty}_a)_{a\in A},\d\big)
	\]
	is called a {\em limit} of the NCG 
	\[
\Gamma =	\Big(\K, A, \S, (r(a,s))_{a\in A, s\in \S}, (\tau_a)_{a\in A},d\Big)
	\]
	w.r.t. the user volume vector sequence \Wuu{$\{d^{(n)}\}_{n\in \N}$,} if
	there exists an infinite subsequence 
	$\{n_i\}_{i\in \N}$ such that:
	\begin{itemize}
		\item[L1)]  
		For each $k\in \K,$
		\[
		\lim_{i\to \infty}\frac{d_k^{(n_i)}}{T(d^{(n_i)})}=\d_k,
		\]
		where $\d_k\in [0,1]$ is the limit volume of group $k.$
		\item[L2)] There exists a sequence $\{g_i\}_{i\in \N}$ of positive scaling factors, such that
		\[
		\lim_{i\to \infty} \frac{\tau_a\big(T(d^{(n_i)})x\big)}{g_i}
		=\tau^{\infty}_a(x)
		\] 
		for all $x>0,$ where $\tau_a^{\infty}(\cdot)$ is the \Wu{{\em limit price}}
		of resource $a,$ for each $a\in A.$
		\item[$L3)$] Each limit price function 
		$\tau^{\infty}_a(\cdot)$ is {\em either} a 
		continuous and non-decreasing function,
		{\em or} $\tau^{\infty}_a(x)\equiv \infty$ for
		all $x>0,$
		for each $a\in A.$ And for each group $k\in \K,$
		\begin{itemize}
			\item[$L3.1)$] {\em either} group $k$ is {\em negligible} w.r.t. 
			scaling factors $\{g_i\}_{i\in \N},$ i.e.,
			\[
			\lim_{i\to \infty} \frac{\sum_{s\in \S_k}f^{(n_i)}_s\cdot\tau_s(f^{(n_i)})}{
				T(d^{(n_i)})\cdot g_i}=0,
			\]
			where each $f^{(n_i)}$ is an arbitrary feasible
			strategy profile of $\Gamma$ \Wu{w.r.t.} user volume
			vector $d^{(n_i)}$ \Wu{for each $i\in \N,$}
			\item[$L3.2)$] {\em or}
			there exists a strategy $s\in \S_k$ that is {\em tight} \Wu{w.r.t.}
			scaling factors $\{g_i\}_{i\in \N},$ i.e.,
			$\tau^{\infty}_a(x)\not\equiv\infty$ for $x>0,$ for each resource
			$a\in A$ with $r(a,s)>0.$
		\end{itemize}
		\item[L4)] Put \begin{equation*}
		\begin{split}
		S^{\infty}&:=\big\{s\in \S: s \text{ is tight w.r.t. }\{g_i\}_{i\in \N}\big\},\\
		\K^{\infty}&:=\{k\in \K: k
		\text{ is not negligible, or } \S_k\cap \S^{\infty}\ne \emptyset\}.
		\end{split}
		\end{equation*} 
		The cost of NE profiles of the limit game $\Gamma^{\infty}$ is positive
		\Wu{w.r.t. the} limit user volume vector $(\d_k)_{k\in\K^{\infty}}.$
	\end{itemize}
\end{definition}
\begin{definition}[See also \cite{Wu2017Selfishness}]\label{def:scalablegame}
	An NCG $\Gamma$ is called a {\em scalable game} if, 
	for each user volume sequence $\{d^{(n)}\}_{n\in \N}$ with
	total volume $T(d^{(n)})\to \infty$ as $n\to \infty,$
	there is a well designed game $\Gamma^{\infty}$ that is a limit of $\Gamma$
	w.r.t. the user volume sequence $\{d^{(n)}\}_{n\in \N}.$ 
\end{definition}
\Wu{Conditions $L3)$ and $L4)$} of Definition 
\ref{def:limitGame} are imposed to guarantee that
the cost of NE profiles in the limit game $\Gamma^{\infty}$ will be neither unbounded,
nor \Wu{vanishing w.r.t.} the scaling factors $\{g_i\}_{i\in \N}.$
Therefore, the scaling factors $\{g_{i}\}_{i\in \N}$ should be carefully
chosen with reference to
the sequence $\{T(d^{(n)})\}_{n\in \N}$ of user volume vectors
and price functions $\tau_a(\cdot)$, \Wu{so as to fulfill these conditions.}

\Wu{
Note that limit games only consider tight strategies
$s\in \S^{\infty}$ and ``non-negligible" groups
$k\in \K^{\infty}$.
This is actually reasonable, since users will asymptotically adopt only tight strategies w.r.t. both, NE profiles and SO profiles, see the proofs of Lemma \ref{theo:LimitGame_Wu2017} and Theorem \ref{theo:NegativeConditionforAWDG} in the Appendix.}

\Wuu{Condition $L2)$} of Definition \ref{def:limitGame} can be further relaxed. \Wu{Note that for
each resource $a\in A,$ only those $x>0$ make sense that
are possible consumed volumes of resource $a$ w.r.t. the 
limit $\Gamma^{\infty}.$} 
We
can therefore
require only that \Wu{the} limit price functions $\tau_a^{\infty}(x)$
exist for $x\in I_a\cap (0,\infty),$ where 
$I_a$ is a non-empty set \Wuu{containing} all the possible consumed volumes of resource
$a$ \Wu{w.r.t.} the limit game $\Gamma^{\infty},$ 
for each $a\in A.$ \Wuu{See \cite{Wu2017Selfishness} for details.}

Note that it is possible that \Wu{there are several limit games for a given user volume sequence
$\{d^{(n)}\}_{n\in \N},$ see
the following example.}
\begin{figure}[!htb]
	\centering
	\includegraphics[scale=0.15]{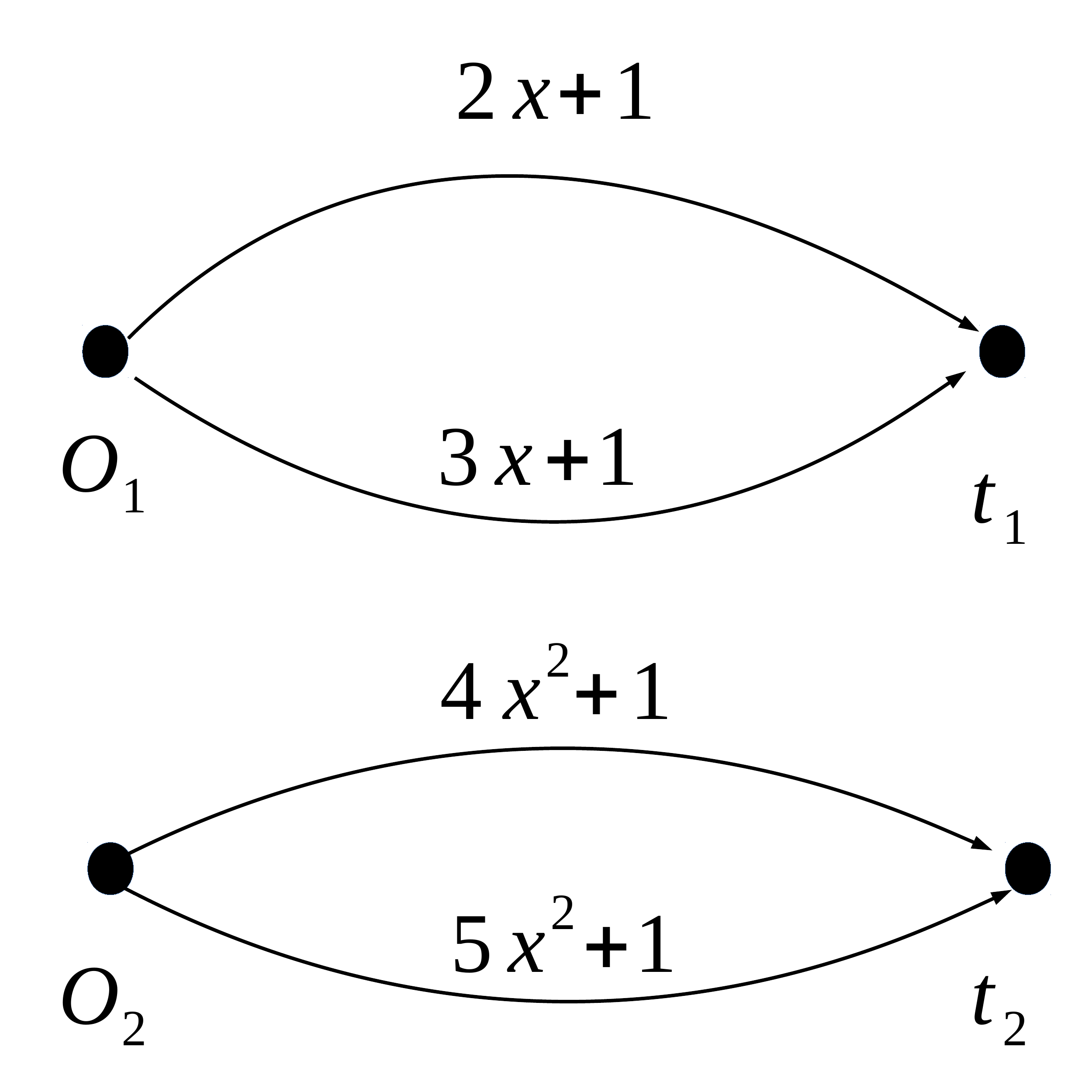}
	\caption{An NCG with double limits}
	\label{fig:doublelimits}
\end{figure}
\begin{example}\label{exam:doublelimts}
	Consider \Wu{the NCG $\Gamma$ shown} in Figure \ref{fig:doublelimits}.
	The game has two groups (OD pairs), each of which has
	two strategies (two parallel paths). The price functions
	are listed above the corresponding paths. Let 
	$\{d^{(n)}=(d_1^{(n)},d_2^{(n)})\}_{n\in \N}$ be a user volume \Wuu{vector} sequence such that
	\begin{equation*}
		d_1^{(n)}=
		\begin{cases}
		0 &\text{if } n\text{ is odd},\\
		n &\text{if } n\text{ is even},
		\end{cases}
		\quad
		d_2^{(n)}=
		\begin{cases}
		n &\text{if } n\text{ is odd},\\
		0 &\text{if } n\text{ is even},
		\end{cases}
	\end{equation*}
	where $d_1^{(n)}$ denotes user volume of the upper group, 
	and $d_2^{(n)}$ denotes user volume of the lower group, for
	each $n\in \N.$ Obviously, \Wu{w.r.t.} subsequence $\{2i\}_{i\in \N}$
	and scaling factor sequence $\{g_{i}=i\}_{i\in \N},$ the
	lower group is negligible since $d^{(2i)}_2\equiv 0$ for
	all $i\in \N.$ Moreover, limit price functions
	\[
	\lim_{i\to\infty}\frac{2\cdot (2i\cdot x)+1}{2i}=2x
	\quad \text{and}\quad \lim_{i\to\infty}\frac{3\cdot (2i\cdot x)+1}{2i}=3x
	\]
	\Wu{exist and are continuous and non-decreasing} for all $x>0.$ 
	
	Furthermore, the NCG game $\Gamma_1^{\infty}$ consisting of these two limit price functions
	and the upper group has positive average cost for NE profiles \Wu{w.r.t.}
	the limit user volume vector $(\d_1),$ where $\d_1=1.$
	\Wu{Thus,} $\Gamma_1^{\infty}$ is a limit game of $\Gamma$
	\Wu{w.r.t.} the given user volume vector sequence $\{d^{(n)}\}_{n\in \N}.$
	
	Similarly, considering the subsequence $\{2i+1\}_{i\in \N}$ and
	scaling factor sequence $\{g_i=(2i+1)^2\}_{i\in \N},$ we can define
	another limit game $\Gamma_2^{\infty}$ 
	\Wu{consisting} of the lower group and two price functions
	$4x, 5x,$ respectively.
	
	Obviously, these two limit games $\Gamma_1^{\infty}$
	and $\Gamma_2^{\infty}$ are different. However, both of them are
	limits of the given NCG $\Gamma$ \Wu{w.r.t.} the user volume
	sequence $\{d^{(n)}\}_{n\in \N}.$
\end{example}

Wu et al. \cite{Wu2017Selfishness} 
proved that for an NCG $\Gamma$ and
a given user volume vector $\{d^{(n)}\}_{n\in \N}$
with total volume $T(d^{(n)})\to\infty$
as $n\to \infty,$ if $\Gamma^{\infty}$ is the
limit of $\Gamma$ \Wu{for} an infinite subsequence 
$\{n_i\}_{i\in \N}$ and a scaling factor sequence
	$\{g_i\}_{i\in \N},$ 
	then the average cost of NE profiles normalized by
the scaling factor sequence $\{g_{i}\}_{i\in \N}$ \Wu{converges} to the total cost of NE profiles of $\Gamma^{\infty}.$
\begin{lma}\label{theo:LimitGame_Wu2017}
	Consider an NCG 
	\[
	\Gamma =	\Big(\K, A, \S, (r(a,s))_{a\in A, s\in \S}, (\tau_a)_{a\in A},d\Big),
	\]
	in which \Wu{all} price functions $\tau_a(\cdot)$ are non-negative, non-decreasing and 
	continuous.  Let
	$\{d^{(n)}\}_{n\in \N}$ be an \Wu{arbitrary} sequence of user volume vectors such that
	the total volume $T(d^{(n)})=\sum_{k\in \K}d_k^{(n)}
	\to \infty$ as $n\to \infty,$ where
	each $d^{(n)}=(d_k^{(n)})_{k\in \K}$
	for each $n\in \N.$ Let $\tilde{f}^{(n)}=\big(f_s^{(n)}\big)_{s\in \S}$ be an NE profile of
	$\Gamma$ \Wu{for the} user volume vector
	$d^{(n)},$ for each $n\in \N.$
	If $\Gamma$ has a limit 
	\[
	\Gamma^{\infty}=\big(\K^{\infty}, A, \S^{\infty}, (r(a,s))_{a\in A, s\in \S^{\infty}}, (\tau^{\infty}_a)_{a\in A},\d\big)
	\]
	\Wu{for} the user volume vector
	sequence $\{d^{(n)}\}_{n\in \N},$
	then there exists an infinite
	subsequence $\{n_i\}_{i\in \N}$
	and a sequence $\{g_{i}\}_{i\in \N}$
	of positive numbers such that
	\[
	\lim_{i\to\infty} \frac{C(\tilde{f}^{(n_i)})}{g_i}
	=\sum_{s\in \S^{\infty}} \tilde{\f}_s^{\infty}\cdot\tau_s^\infty(\tilde{\f}^{\infty}),
	\]
	where $\tilde{\f}^{\infty}=(\tilde{\f}_s)_{s\in \S^{\infty}}$ is some NE profile of
	the limit game $\Gamma^{\infty}.$
\end{lma}
\begin{proof}
	See \Wuu{the proof of Theorem 3.2 in \cite{Wu2017Selfishness}, or the appendix for an alternative proof.}
\end{proof}

\Wu{Using} Lemma \ref{theo:LimitGame_Wu2017}, 
Wu et al. \cite{Wu2017Selfishness} then proved that
all scalable games are asymptotically well designed.
In that proof, the condition that the limit game
is well designed \Wu{plays} a pivotal role, \Wu{which actually implies that the average costs of NE profiles is
asymptotically not larger than the average costs
of SO profiles. We summarize this in Lemma \ref{theo:MainResult_Wu2017}.}
\begin{lma}\label{theo:MainResult_Wu2017}
	Every scalable game is an AWDG.
\end{lma}
\begin{proof}
	See \Wuu{the proof of Theorem 3.2 in \cite{Wu2017Selfishness} for details.}
\end{proof}

Theorem \ref{theo:NegativeConditionforAWDG} \Wu{below} continues
the study of Wu et al. \cite{Wu2017Selfishness}.
\Wu{It states} that scalable games and AWDGs actually coincide \Wu{when they have a limit game.}
\Wu{Moreover, we showed in the proof of Theorem \ref{theo:NegativeConditionforAWDG}
that users will asymptotically 
adopt only tight strategies w.r.t. SO profiles, and
Lemma \ref{theo:LimitGame_Wu2017}
also applies to SO profiles. Therefore, it
is reasonable to consider only tight strategies \Wuu{in limit games.}}
\begin{thm}\label{theo:NegativeConditionforAWDG}
	Consider an NCG $\Gamma$ and a user volume
	vector $\{d^{(n)}\}_{n\in \N}$ with the total volume $T(d^{(n)})\to \infty$
	as $n\to \infty.$ If $\Gamma$ has a limit game $\Gamma^{\infty}$ \Wu{w.r.t.}
	the given user volume vector, and the limit game  $\Gamma^{\infty}$
	is not well designed, then $\Gamma$ is not an AWDG.
\end{thm}
\begin{proof}
	See the Appendix.
\end{proof}

Specific \Wuu{for} polynomial price functions, \Wu{Theorem
3.2}
from Wu et al. \cite{Wu2017Selfishness} \Wu{directly
	yields} that
NCGs with polynomial price functions $\tau_a(\cdot)$
of the same degree are asymptotically well designed.
In \Wuu{that case}, we can take scaling factors $g_n=T(d^{(n)})^{\gamma}$
for each given user volume vector sequence $\{d^{(n)}\}_{n\in \N}$
with $T(d^{(n)})\to \infty$ as $n\to \infty,$ where
$\gamma\ge 0$ is the common degree \Wu{of all polynomials}. Then, the corresponding
limit game is well designed. 

\Wu{To show this for general polynomial price functions}
$\tau_a(\cdot)$ \Wuu{with different degrees}, we \Wu{will use} 
some helpful notations from \cite{Colini2017a}. 
For each resource $a\in A,$ let $\rho_a(\cdot)$ be
the degree of polynomial $\tau_a(\cdot).$ Put
$\rho_s=\max\{\rho_a: r(a,s)>0, a\in A\}$ and
$\rho_k=\min\{\rho_s: s\in \S_k\}$ for each $s\in \S_k$
and $k\in \K.$ \Wu{If $\rho_k=\rho_l$
for all $k,l\in\K,$ then the above argument shows} that 
the underlying NCG is scalable, and therefore asymptotically well designed.
\Wu{We summarize this in Corollary \ref{theo:SameDegreeScalable}.}
\begin{corollary}\label{theo:SameDegreeScalable}
	Consider an NCG $\Gamma$ with polynomial price functions
	$\tau_a(\cdot).$ If $\rho_k=\rho_l$
	for all $k,l\in\K,$ then $\Gamma$ is asymptotically well designed.
\end{corollary}
\begin{proof}
	Let $\{d^{(n)}\}_{n\in \N}$ be \Wu{an arbitrary
	} user volume vector such that $T(d^{(n)})\to \infty$
	as $n\to \infty.$ Let $\rho=\rho_k$ for some
	$k\in \K.$ Put $g_n=T(d^{(n)})^{\rho}$ for each $n\in \N.$
	By assumption \eqref{eq:NoFreeStrategies}, 
	one can then easily show that the limit of $\Gamma$ \Wu{w.r.t.}
	sequence $\{n\}_{n\in \N}$ and \Wu{the} scaling factor sequence $\{g_n\}_{n\in \N}$
	is well designed. 
\end{proof}

\Wu{The result of Wu et al. \cite{Wu2017Selfishness}
\Wuu{does not} directly apply if  $\rho_k\neq \rho_l$ for some
$l,k\in \K$.} The reason is that, in this case, there
\Wu{need} not exist a unified limit game for all \Wu{{\em groups}} for some
user volume vector sequences $\{d^{(n)}\}_{n\in \N}$.
We \Wu{thus need additional arguments in this case and leave the proof of this case to Subsection~\ref{sec:asymptoticdecomposition}.}

\Wu{To better understand
the current state} of the art, we now introduce some relevant results from 
\cite{Colini2017a}
and \cite{Colini2017b}.
They 
employed an alternative technical path to prove
the convergence of \Wu{the} PoA. They \Wu{introduced} the so-called
{\em gaugeable games}, which \Wu{consider only} particular sequences
of user volume vectors. Gaugeability is a concept based on \Wu{the notion of}
{\em regular variation} \cite{Bingham1987Regular}.
A {\em non-negative} function $g(\cdot)$ is said \Wu{to} be {\em regularly varying}, if the limit
\begin{equation}\label{eq:regularvarying}
	\lim_{t\to\infty}
	\frac{g(t\cdot x)}{g(t)}=q(x)\in (0,\infty)
\end{equation}
exists for all $x>0.$ 

\begin{definition}[See also \cite{Colini2017b}]\label{def:GaugeableGame}
	An NCG $\Gamma$ is said to be gaugeable \Wu{for}
	a user volume vector $\{d^{(n)}\},$ if there exists a \Wu{regularly} varying
	function $g(\cdot)$ such that:
	\begin{itemize}
		\item[$G1)$] The limit
		\[
		\lim_{n\to \infty} \frac{\tau_a(x)}{g(x)}=:q_a\in [0,\infty]
		\]
		exists for all resource $a\in A.$
		\item[$G2)$] For each group $k\in \K,$ there exists a strategy 
		$s\in \S_k$ such that
		\[
		q_a<\infty, \quad\text{for all resource } a\in A \text{ with }r(a,s)>0.
		\]
		\item[$G3)$] The lower limit
		\begin{equation}\label{eq:GaugeableCondition}
		 \varliminf_{n\to \infty} \sum_{k\in \K_{tight}}\frac{d_k^{(n)}}{T(d^{(n)})}>0,
		\end{equation}
		where $\K_{tight}$ is the set of all tight groups, and
		a group $k$ is said to \Wu{be} tight \Wu{in such a case} if 
		\[
		0<\min_{s\in \S_k}\max\{q_a: r(a,s)>0, a\in A\}<\infty.
		\]
	\end{itemize}
\end{definition}

With the technique of regular variation \cite{Bingham1987Regular}, Colini et al. \cite{Colini2017b} proved that for  
each NCG $\Gamma,$ if $\Gamma$ is gaugeable \Wu{w.r.t.}
a user volume vector $\{d^{(n)}\}_{n\in \N},$
then PoA$(d^{(n)})\to 1$ as $n\to \infty,$ \Wuu{see Theorem
4.4 in \cite{Colini2017b}. Here,}
we recall that PoA$(d^{(n)})$ denotes the price of
anarchy \Wu{w.r.t.} user volume vector $d^{(n)},$ for each
$n\in \N.$

In fact, if $\Gamma$ is gaugeable \Wu{w.r.t.} a user volume vector
$\{d^{(n)}\}_{n\in \N},$ then $\Gamma$ has a well designed
limit \Wu{w.r.t.} that user volume vector sequence.
Let $g(\cdot)$ be the required \Wu{regularly} varying
function in Definition~\ref{def:GaugeableGame}.
By $G1),$ \Wu{the} limit price functions 
\[
\tau_a^{\infty}\big(x\big)=\lim_{n\to \infty}
\frac{\tau_a\big(T(d^{(n)})x\big)}{g_n}=q_a\cdot x^{\rho}
\]
exist, where the scaling factors 
$g_n=g\big(T(d^{(n)})\big),$
and $\rho\ge 0$ is the \Wu{regular variation index} of $g$ in Karamata's Characterization Theorem \cite{Bingham1987Regular}. Moreover, by $G2)$ and
$G3),$ one can easily check that the limit game
consisting of groups $\K^{\infty}=\K$ and price functions
$\tau_a^{\infty}$ is well designed. See Wu et al. \cite{Wu2017Selfishness}
for details. Therefore, gaugeable games are \Wu{special cases of} scalable games.
To explicitly show this, we define the ``sequential counterpart" of
scalable games in a natural way.
\begin{definition}\label{def:ScalableSequence}
	Consider an NCG $\Gamma$ and \Wu{some} user volume vector $\{d^{(n)}\}_{n\in \N}$
	with $T(d^{(n)})\to \infty$ as $n\to \infty.$ 
	$\Gamma$ is said to be {\em scalable} \Wu{w.r.t. the} sequence
	$\{d^{(n)}\}_{n\in \N},$ if 
	for each infinite subsequence $\{n_i\}_{i\in \N},$
	$\Gamma$ has a well designed limit \Wu{w.r.t.} the 
	subsequence $\{d^{(n_i)}\}_{i\in \N}.$  
\end{definition} 

\Wu{\Wuu{Obviously,} an NCG $\Gamma$ is scalable if and only if
	$\Gamma$ is scalable w.r.t. each user volume sequence
	$\{d^{(n)}\}_{n\in \N}$ with $\lim_{n\to \infty}T(d^{(n)})=\infty.$
	However, if $\Gamma$ is only scalable w.r.t. some user
	volume sequence $\{d^{(n)}\}_{n\in \N},$ then
	$\Gamma$ need not to be \Wuu{globally} scalable.
	Nevertheless,
	Lemma \ref{theo:Gaugeable<=Scalable} below states
	that this restricted \Wuu{notion} of scalability 
	already \Wuu{generalizes} the gaugeability \Wuu{of} \cite{Colini2017b}.}
\begin{lma}\label{theo:Gaugeable<=Scalable}
	Consider an NCG $\Gamma$ and some user volume vector $\{d^{(n)}\}_{n\in \N}$
	with $T(d^{(n)})\to \infty$ as $n\to \infty.$
	\Wu{If} $\Gamma$ is gaugeable \Wu{w.r.t. the sequence} $\{d^{(n)}\}_{n\in \N},$
	then $\Gamma$ is also scalable \Wu{w.r.t.} that user volume vector
	sequence.
\end{lma}
\begin{proof}
	See \Wuu{the proof of Corollary 3.1 in Wu et al. \cite{Wu2017Selfishness} for details.}
\end{proof}

The difference between scalable games and gaugeable games 
are \Wu{therefore} obvious. Scalable games \Wu{consider arbitrary} user volume
vector sequence, \Wu{while} gaugeable games \Wu{consider} particular
user volume vector sequence satisfying \eqref{eq:GaugeableCondition}.
Thus, the convergence result of \Wu{the PoA for} scalable games is 
global, while that \Wu{for gaugeable games holds only locally.}

Actually, \Wu{scalablity is more general than gaugeability even
for a specific user volume vector sequence $\{d^{(n)}\}_{n\in \N}$.} Gaugeability requires that
\Wu{there exists a uniform \Wu{regularly} varying function
$g(\cdot)$ for the whole sequence that fulfills conditions $G1)$-$G3)$. }
\Wu{As shown above, this results} in \Wu{the same} well designed limit game for every subsequence
of $\{d^{(n)}\}_{n\in \N}$. However,
scalability \Wu{allows} different subsequences of $\{d^{(n)}\}_{n\in \N}$ to have different well designed limit games. \Wu{The} NCG $\Gamma$ in Example \ref{exam:doublelimts}
has two limit games \Wu{w.r.t.} the given user volume vector sequence,
and \Wu{both of them are} well designed, see Wu et al. \cite{Wu2017Selfishness}
for details. \Wu{We summarize this in Lemma \ref{theo:ScalableGammes>GaugeableGames}.} 
\begin{lma}\label{theo:ScalableGammes>GaugeableGames}
	There exists an NCG $\Gamma$ and a user volume vector sequence
	$\{d^{(n)}\}_{n\in \N}$ with 
	$T(d^{(n)})\to \infty$ as $n\to\infty,$ such that
	$\Gamma$ is scalable \Wu{w.r.t.} $\{d^{(n)}\}_{n\in \N},$
	but not gaugeable \Wu{w.r.t.} $\{d^{(n)}\}_{n\in \N}.$
\end{lma}
\begin{proof}
	See \Wuu{the proof of Theorem 3.4 in Wu et al. \cite{Wu2017Selfishness}, } \Wu{or Example \ref{exam:doublelimts}.}
\end{proof}

Now, let us \Wu{return} to the discussion of NCGs with \Wu{arbitrary} polynomial
price functions $\tau_a(\cdot).$
By considering \Wu{a} particular user volume vector \Wu{sequence} $\{d^{(n)}\}_{n\in \N},$
Colini et al. \cite{Colini2017a} proved that 
PoA$(d^{(n)})\to 1,$ as $n\to \infty,$ \Wuu{see also Corollary 4.7 in \cite{Colini2017b}.} They assumed that 
for each $k\in \K,$
\begin{equation}\label{eq:ColiniAssumption}
\varliminf_{n\to \infty} \frac{d_k^{(n)}}{T(d^{(n)})} >0.
\end{equation}
Obviously, \eqref{eq:ColiniAssumption} fulfills \eqref{eq:GaugeableCondition}.
Let $g(x)=T(d^{(n)})^{\rho},$ where $\rho=\max\{\rho_k:k\in \K\}.$ They
actually proved that if $\{d^{(n)}\}_{n\in \N}$ satisfies 
\eqref{eq:ColiniAssumption}, then the underlying NCG is gaugeable \Wu{w.r.t.}
$\{d^{(n)}\}_{n\in \N}$ and regularly varying function $g(\cdot).$
\Wu{We summarize this in Lemma \ref{theo:ColiniPoly}.}
\begin{lma}\label{theo:ColiniPoly}
	Consider an NCG $\Gamma$ with polynomial price functions
	$\tau_a(\cdot),$ and a user volume vector sequence
	$\{d^{(n)}\}_{n\in \N}$ such that $T(d^{(n)})\to \infty$
	as $n\to \infty,$ and satisfies \eqref{eq:ColiniAssumption}.
	Then, $\Gamma$ is scalable \Wu{w.r.t.} 
	$\{d^{(n)}\}_{n\in \N},$ i.e., PoA$(d^{(n)})\to 1$ as $n\to \infty.$
\end{lma}
\begin{proof}
	\Wu{The proof follows immediately 
		\Wuu{from Lemma \ref{theo:Gaugeable<=Scalable}.}}
\end{proof}

Although Lemma \ref{theo:ColiniPoly} shows an \Wu{inspiring} result
for general polynomial price functions, we can not
directly conclude that NCGs with polynomial price functions
are asymptotically well designed, due to \Wu{the sequence-specific
nature of Lemma \ref{theo:ColiniPoly}.} \Wu{To} show \Wu{that} NCGs with polynomial price functions
are asymptotically well designed, we still need a more sophisticated
analysis. \Wu{Motivated by} Lemma \ref{theo:Gaugeable<=Scalable},
Lemma \ref{theo:ScalableGammes>GaugeableGames}, and
Example \ref{exam:doublelimts},  we
\Wuu{will} \Wu{use} the idea of scalability \Wu{for a general proof.}

\subsection{AWDG is closed under direct sum}
\label{sec:DirectSum}

\Wu{A direct application of scalability does not lead \Wuu{to} a general proof.} 
\Wu{To see this,} consider again Example \ref{exam:doublelimts}, but
now with user volume vector sequence $d^{(n)}=(d_1^{(n)}=n^2,d_2^{(n)}=n).$ 
In this case, the two groups are \Wu{mutually} non-negligible, and we \Wu{cannot}
find a suitable scaling factor sequence \Wu{that results in} a well designed
limit game. However, \Wu{a closer inspection shows} that
either of the two groups has its own well designed limit game \Wu{w.r.t.}
the given user volume sequence.
This inspires us to consider the two groups separately.
\begin{definition}\label{def:MutuallyDisjointGroups}
	An NCG $\Gamma$ is said to have {\em mutually disjoint groups} (MDGs), if
	\begin{equation}\label{eq:MDG}
	\sum_{k\in \K} \1_{\{x: x>0\}}\big(r(a, s_k)\big)\le 1,
	\end{equation}
	\Wu{for each $a\in A,$ for each K-dimensional strategy vector
		$(s_1,\ldots,s_K)\in \prod_{k\in \K} \S_k,$} where $\1_{\{x: x>0\}}(\cdot)$ is the indicator function of
	set $\{x:x>0\}.$
\end{definition}

\Wu{Condition} \eqref{eq:MDG} in Definition \ref{def:MutuallyDisjointGroups}
implies that users from different groups \Wu{cannot} share resources.
Therefore, users from different groups of an NCG with MDGs
will not affect each other when they determine strategies to follow.
This means that each group in an NCG with MDGs forms an independent
\Wu{subgame.} Let 
\[
\Gamma=	\Big(\K, A, \S, (r(a,s))_{a\in A, s\in \S}, (\tau_a)_{a\in A},d\Big)
\]
be an NCG with MDGs. For each group $k\in \K,$ we denote by
\[
\Gamma_{| k}=\Big(\{k\}, A, \S, (r(a,s))_{a\in A, s\in \S_k}, (\tau_a)_{a\in A},(d_k)\Big)
\]
the {\em $k$-marginal} of $\Gamma$, and denote by $f_{|k}=(f_s)_{s\in \S_k}$
the {\em $k$-marginal} of a feasible strategy $f=(f_s)_{s\in \S}.$
\begin{lma}\label{theo:MarginalProperty}
	Consider an NCG $\Gamma$ with MDGs. Then, $f^*=(f_s^*)_{s\in \S}$
	is an SO profile of $\Gamma$ if and only if the $k$-th marginal profile $f^*_{|k}=\big(f_s^*\big)_{s\in \S_k}$
	is an SO profile of the $k$-marginal game $\Gamma_{|k},$ for each $k\in \K.$
	\Wu{This holds similarly} for NE profiles.
\end{lma}
\begin{proof}
	Trivial.
\end{proof}

With Lemma \ref{theo:MarginalProperty} \Wu{and by applying}
scalability to each of \Wuu{the} mutually independent marginals, we can easily derive that
NCGs with mutually \Wu{disjoint} and scalable marginals
are asymptotically well designed. 
\begin{lma}\label{theo:AWDG_NCG_MDG}
	Consider an NCG $\Gamma$ with \Wu{MDGs. If all the marginals are
		scalable,} then $\Gamma$
	is asymptotically well designed.
\end{lma}
\begin{proof}
	\Wu{This} is easy by observing the trivial facts that
	\[
	\text{PoA}(d)=\frac{C(\tilde{f})}{C(f^*)}=
	\frac{\sum_{k\in \K} \sum_{s\in \S_k} \tilde{f}_s\cdot \tau_s(\tilde{f})}{\sum_{k\in \K} \sum_{s\in \S_k} f^*_s\cdot \tau_s(f^*)},
	\]
	and that \Wu{each marginal is scalable, and thus }the marginal PoA
	\[
	\frac{ \sum_{s\in \S_k} \tilde{f}_s\cdot \tau_s(\tilde{f})}{\sum_{k\in \S_k} f^*_s\cdot \tau_s(f^*)},
	\]
	tends to $1,$ as $d_k\to \infty,$ for
	each $k\in \K.$ Herein we recall  assumption \eqref{eq:NoFreeStrategies} \Wu{that there are} no free strategies, and that
	\[
	\tau_s(\tilde{f})=\tau_s(\tilde{f}_{|k}) \quad\text{and}
	\quad \tau_s(f^*)=\tau_s(f^*_{|k}),
	\]
	if $s\in \S_k,$ for each $k\in \K.$ 
\end{proof}

\Wu{Combining Corollary \ref{theo:SameDegreeScalable} and
	Lemma \ref{theo:AWDG_NCG_MDG}, we obtain immediately
	that NCGs with MDGs and polynomial price functions
	$\tau_a(\cdot)$ are asymptoticall well designed.}

A direct extension of Lemma \ref{theo:AWDG_NCG_MDG} \Wu{considers} the {\em direct sum} of asymptotically \Wuu{well} designed games. Let 
\[
\Gamma_l=	\Big(\K_l, A_l, \S_l, (r(a,s))_{a\in A_l, s\in \S_l}, (\tau_a)_{a\in A_l},d(l)\Big)
\]
be \Wu{an NCG}, for $l=1,\ldots,m,$ such that 
$A_1,\ldots,A_m$ are mutually disjoint.
Then, we call the game
\[
\Big(\bigcup_{l=1}^{m}\K_l, \bigcup_{l=1}^{m}A_l, \bigcup_{l=1}^{m}\S_l, (r(a,s))_{a\in \bigcup_{l=1}^{m}A_l, s\in \bigcup_{l=1}^{m}\S_l}, (\tau_a)_{a\in \bigcup_{l=1}^{m}A_l},\bigcup_{l=1}^{m}d(l)\Big)
\]
the {\em direct sum} of $\Gamma_1,\ldots,\Gamma_m,$
denoted by $\oplus_{l=1}^{m}\Gamma_l,$ where
$\bigcup_{l=1}^{m}d(l)$ means the concatenation of vectors
$d(1),\ldots,d(m).$ Obviously, direct sums of asymptotically \Wuu{well}
designed games are again asymptotically well designed.
\begin{corollary}\label{theo:DirectSum}
	The class AWDG is closed under direct sums.
\end{corollary}
\begin{proof}
	Trivial.
\end{proof}

Corollary \ref{theo:DirectSum} \Wu{suggests} a possible
approach to check whether an NCG is aysmptotically 
well designed. One can try to decompose the underlying NCG
into a direct sum of several independent marginals, and
then check the scalability of each marginal.
\Wuu{Here, we allow compound marginals, i.e.,
	each marginal can contain more than \Wuu{one}
	group. The independence between them then
	means that users from different marginals
	do not affect each other when they determine 
	strategies to follow.}

However, in general, it could be difficult to
find \Wu{such} a decomposition, since \Wuu{different} groups
might \Wuu{compete for} the \Wuu{same} resources.
The above \Wu{discussion} has already shown that if collisions
are heavy or slight, then it \Wu{is} easy to check
whether the underlying NCG is asymptotically well designed.
In fact, if all groups heavily collide on resources, i.e.,
every resource can be used by all groups, then 
the game is not decomposable and Lemma \ref{theo:MainResult_Wu2017}
applies, see e.g. Corollary \ref{theo:SameDegreeScalable}.
\Wu{On the other hand,} if groups \Wuu{do not} collide on resources, i.e.,
\Wu{if} they
can be partitioned into several mutually
disjoint classes \Wu{w.r.t.} the use of resources, then the game is decomposable and 
Corollary \ref{theo:DirectSum} may \Wu{apply}. 

The above two cases are, somehow, regular. \Wu{Below, we}
consider the case of irregular collisions, which might be
\Wu{the general case} in practice. 

\subsection{Asymptotic decomposition: a general proof for polynomial price functions}
\label{sec:asymptoticdecomposition}

In general, it may be difficult to directly apply
the idea of scalability, since there \Wu{need} not exist a unified
well designed game \Wu{for all groups} for some user volume vector sequence.
Moreover, it \Wu{may} also be difficult to directly decompose
the game, due to irregular collisions of groups on resources.
If this is the case, then one may consider an {\em asymptotic
decomposition.} The idea is similar \Wu{to} \Wuu{direct sums.}
However, one needs to additionally deal with \Wuu{the} problems caused by the irregular collisions. 

\Wu{An} asymptotic decomposition \Wu{is based} on a suitable
partition \Wuu{of} the group set $\K.$ Consider an \Wu{arbitrary} sequence $\{d^{(n)}\}_{n\in \N}$ of user volume
vectors such that $T(d^{(n)})\to \infty$ as $n\to\infty.$
The decomposition aims to eventually \Wuu{partition} $\K$ into 
$\K_0,\ldots,\K_t,$ for some integer $t\ge 0,$ such that
$\K=\bigcup_{m=0}^{t}\K_m,$ and
\begin{equation}\label{eq:ObjectiveAsymDecomp}
	\lim_{n\to\infty}
	\frac{\sum_{k\in \bigcup_{u=0}^{m}\K_u} \sum_{s\in \S_k}\tilde{f}_s^{(n)}\cdot\tau_s(\tilde{f}^{(n)})}{
		\sum_{k\in \bigcup_{u=0}^{m}\K_u}\sum_{s\in \S_k}f_s^{*(n)}\cdot\tau_s(f^{*(n)})
		}=1,
\end{equation}
for each $m=0,\ldots,t,$ where $f^{*(n)}, \tilde{f}^{(n)}$
are SO and NE profiles \Wu{w.r.t.} $d^{(n)},$ respectively,
for each $n\in \N.$ 
Obviously, if \Wu{such} partition can indeed be \Wu{constructed}, then
the underlying NCG is well designed, due to the
\Wu{arbitrary choice} of $\{d^{(n)}\}_{n\in \N}.$

\Wu{One can try to construct the partition inductively.}
In the beginning of each \Wu{inductive} step $l=0, \ldots,t,$ we assume that we have
already constructed classes $\K_0, \ldots,\K_{l-1},$
such that $\K_0\subseteq \K, \ldots,\K_{l-1}\subseteq \K,$
and \eqref{eq:ObjectiveAsymDecomp} holds for
step $m=l-1$,
where we employ \Wu{the} convention that 
$\K_{-1}=\emptyset,$ $\frac{0}{0}=1,$
and ``$\K_0,\K_{-1}$" means ``$\emptyset$".
The objective of step $l$ \Wu{then} is to constuct
class $\K_l\subseteq \K\backslash\bigcup_{u=0}^{l-1}\K_u$
such that \eqref{eq:ObjectiveAsymDecomp} holds \Wu{again}
for $m=l.$

To construct $\K_l$, one can inspect the remaining
groups \Wu{more closely}, and \Wu{pick those} groups
$k\in \K\backslash\bigcup_{u=0}^{l-1}\K_u$
\Wu{that have} a non-vanishing limit proportion in the 
remaining total user volume $T_l(d^{(n)}):=T(d^{(n)})-
\sum_{k\in \bigcup_{u=0}^{l-1}\S_k}d_k^{(n)},$
since these groups are most significant
in the limit. To show 
\eqref{eq:ObjectiveAsymDecomp} for
$m=l,$ one needs to argue that these groups
are {\em either} asymptotically independent of the groups
that have been considered before, {\em or}
negligible \Wu{compared to them}.

To that end, one needs to {\em suitably} estimate the costs of users 
\Wu{w.r.t.} NE profiles $\tilde{f}^{(n)}$ and SO profiles
$f^{*(n)},$ respectively. By comparing \Wu{the cost of users  from groups $k\in\K_l$ with those from groups $k'\in \bigcup_{u=0}^{l-1}\K_u$,} one can then
\Wu{learn} whether groups $k\in \K_l$ are negligible.
\Wuu{If they} are negligible, then 
\eqref{eq:ObjectiveAsymDecomp} holds trivially for 
$m=l.$ Otherwise, groups $k\in \K_l$ \Wu{will be} asymptotically independent
of \Wuu{groups} $k'\in \bigcup_{u=0}^{l-1}\K_u,$ since \Wuu{the cost} of users
from groups $k'\in \bigcup_{u=0}^{l-1}\K_u$ will be negligible
compared with the \Wuu{cost} of users from groups $k \in \K_l$.
If this is the case,
one can then check the scalability of groups \Wuu{from} $\K_l$
\Wuu{under the condition} that users from groups $k'\in \bigcup_{u=0}^{l-1}\K_u$
adopt strategies \Wuu{that} they used in NE profiles $\tilde{f}^{(n)}$
and SO profiles $f^{*(n)},$ respectively. Moreover, if 
these groups are scalable, then \eqref{eq:ObjectiveAsymDecomp}
follows immediately for $m=l.$

 This procedure 
can tactically avoid the impact of possible irregular collisions 
\Wu{in} the limit analysis \Wuu{by} comparing
the costs of users from different classes $\K_l$. \Wu{If the above 
	partition can be constructed, then the underlying game 
	\Wuu{decomposes} into several asymptotically independent subgames
	corresponding to the partition  $\K_0,\ldots,\K_t$. Although
	these subgames share resource set $A,$ they are
	asymptotically independent, since the choices of users from one
	subgame will be asymptotically independent of those
	from other subgames.}

\Wu{Asymptotic decomposition
	can be successfully applied to NCGs with arbitrary polynomial price functions,
	which directly implies that NCGs with arbitrary
	polynomial price functions are asymptotically well designed.} 
Theorem
\ref{theo:SelfishMainTh} \Wu{summarizes}
\Wuu{this} result. \Wu{We \Wuu{move} the detailed decomposition
	procedure \Wuu{to} the Appendix.}
\begin{thm}\label{theo:SelfishMainTh}
	Consider \Wu{an NCG $\Gamma$} with 
	polynomial price functions $\tau_a(\cdot)$ such that 
	\Wu{each $\tau_a(x)$ is
	non-negative and non-decreasing
	for all $x\ge 0$ for $a\in A$.} \Wu{Then, $\Gamma$ is
	asymptotically decomposable, and thus asymptotically well designed.}
\end{thm}
\begin{proof}
	See the appendix.
\end{proof}

Different from the direct \Wu{sum}
in Subsection \ref{sec:DirectSum},
\Wu{an SO profile} \Wuu{need} not be locally \Wuu{a} system
optimum \Wu{w.r.t. some ``marginals" corresponding
	to the partition in the asymptotic decomposition}. This introduces \Wu{extra
	difficulties}
\Wu{in the application of scalability.}
The proof of Theorem~\ref{theo:SelfishMainTh} overcomes
\Wu{them} by considering SO profiles as 
NE profiles \Wu{w.r.t.} price functions
$c_a(x):=x\tau_a'(x)+\tau_a(x),$
where each $\tau_a'(\cdot)$ is the
derivative function of
$\tau_a(\cdot).$ Note that \Wu{NE profiles} are still at (pure) Nash equilibrium \Wu{w.r.t.}
each marginal \Wuu{under the condition} that users from other marginals adhere to the strategies they used in corresponding
NE profiles. 

\Wu{The proof of
Theorem \ref{theo:SelfishMainTh} is based on three
basic properties of polynomial functions.} The first is
that \Wu{polynomial functions
	can be asymptotically sorted according
	to their degrees,} which forms a base for the
cost comparison and \Wu{the construction of
	\Wuu{scaling factors} at
	each inductive step}. The second is that the price
functions $c_a(x)=x\tau_a'(x)+\tau_a(x)$ are
comparable with the price functions $\tau_a(x),$
i.e., \Wu{$\lim_{x\to\infty}\frac{c_a(x)}{\tau_a(x)}=q_a$
	for some constant $q_a\in (0,\infty),$ when
	all $\tau_a(\cdot)$ are polynomials. This plays
	a pivotal role when we check scalability
	for marginals in the inductive steps.}
The last \Wuu{properly} is \Wu{the relatively
	clear
	structure of polynomial functions, from which
	we can \Wuu{obtain} suitable scaling factors
	$g_n^{(l)}$ \Wuu{at} each inductive step $l$}.

The proof of Theorem \ref{theo:SelfishMainTh} is
very \Wuu{elementary.} It
does not involve any \Wu{advanced techniques}, but only
 mathematical induction, \Wu{basic calculus,} and \Wu{a very crude ranking \Wuu{of} non-negative
functions.} Therefore, it should be widely readable.

Theorem \ref{theo:SelfishMainTh}  
\Wu{fully settles the convergence of the} PoA for \Wu{arbitrary} 
polynomial price functions. \Wu{This} greatly extends 
\Wu{the partial}
results from \cite{Colini2017a}, \cite{Colini2017b}
and \cite{Wu2017Selfishness} \Wuu{for} polynomial price
functions. Due to the \Wu{popularity} of
polynomial functions,
Theorem \ref{theo:SelfishMainTh} may apply to a more
general context other than road traffic, \Wuu{e.g.} the senario of free market mentioned
in \Wuu{the} Introduction.

\subsection{A further extension: a general proof for regularly
	varying price functions}
\label{sec:Extension}

The idea of asymptotic decomposition may apply also to NCGs with price functions of other types. 
\Wu{Polynomial functions are special cases of regularly
	varying functions. This \Wuu{subsection} aims to apply
	\Wuu{the} asymptotic decomposition to this more general
	notion.}

\Wu{By Karamata's Characterization Theorem
\cite{Bingham1987Regular}, a regularly varying function} $\tau(\cdot)$ 
\Wu{can be written as}
\[
\tau(x)=x^{\rho}\cdot Q(x),
\] 
where $\rho\in \mathbb{R}$ is called the {\em regular
	variation index} of $\tau(\cdot)$ and 
$Q(x)$ is a {\em slowly varying} function, i.e.,
for each $x>0,$
\[
\lim_{t\to\infty}\frac{Q(tx)}{Q(t)}=1.
\]
\Wu{The class of regularly varying functions
	is very extensive \Wuu{and} includes \Wuu{many popular} analytic
	functions, e.g.,
	all affine functions, polynomials, logarithms,
	and others.}

\Wu{Although regularly varying functions are
	more extensive than polynomial functions,
	they actually \Wuu{have similar properties.} Lemma \ref{theo:RV_Properties}
	summarizes these properties.
	}
\begin{lma}\label{theo:RV_Properties}
	Consider a regularly varying function $\tau(\cdot)$ with
	index $\rho \ge 0.$ Then, \Wu{the} following \Wu{statements} hold.
	\begin{itemize}
		\item[$a)$] For each $\epsilon>0,$
		\[
		\lim_{x\to \infty} \frac{\tau(x)}{x^{\rho+\epsilon}}=0,
		\quad \text{and} \quad \lim_{x\to \infty}\frac{\tau(x)}{x^{\rho-\epsilon}}=\infty.
		\]
		\item[$b)$] For each non-negative function $g(\cdot),$ if 
		\[
		\lim_{x\to \infty} \frac{g(x)}{\tau(x)}=q\in (0,\infty)
		\]
		for some constant $q,$ then $g(\cdot)$ is also regularly
		varying with index $\rho$.
		\item[$c)$] For each regularly varying function $g(\cdot)$
		with index $\rho',$ the qoutient
		\[
		\frac{g(x)}{\tau(x)} 
		\]
		is again regularly varying, but with index $\rho'-\rho.$
		Therefore, the qoutient of two slowly varying non-zero functions is again slowly varying.
	\end{itemize}
\end{lma}
\begin{proof}
	See the Appendix, \Wu{or \cite{Bingham1987Regular}.}
\end{proof}

\Wu{Lemma \ref{theo:RV_Properties}~$a)$ and
	Lemma \ref{theo:RV_Properties}~$c)$ indicate
	a partial ordering on the class of regularly
	varying functions, i.e., 
	regularly varying functions of
	different indices can be sorted according to their indices.
	However, two regularly \Wuu{varying} functions
	of the same index are generally incomparable.
	Therefore, we cannot directly reuse the simple
	ordering \Wuu{of} polynomial functions when we
	apply the asymptotic decomposition to 
	NCGs with regularly varying functions.
	Lemma \ref{theo:RV_Properties}~$b)$
	will be implicitly used in our discussion.
	\Wuu{It guarantees} that the
	auxiliary price functions \Wuu{$c_a(x)=x\tau_a'(x)+\tau_a(x)$}
	are again regularly varying and have the same
	indices \Wuu{as} the price functions $\tau_a(x)$.  
	}

\Wu{Due to the generality of
	regularly varying functions,
	we cannot have a uniform argument. 
	To simplify the discussion, we 
	shall consider regularly
	varying functions with particular
	properties in this Subsection, e.g., \Wuu{convex and differentiable} 
	regularly varying functions.}
\begin{lma}\label{theo:RegularlyVaryingConvex}
	Consider a regularly varying function $\tau(\cdot)$
	that \Wuu{is} non-decreasing, non-negative, convex and differentiable on $[0,\infty).$
	Then, the regular variation index of $\rho$ \Wuu{is}
	non-negative, and
	\begin{equation}\label{eq:ImportantCondiAD}
	\lim_{x\to\infty} \frac{x\cdot\tau'(x)}{\tau(x)}=\rho\ge 0.
	\end{equation}
\end{lma}
\begin{proof}
	See the Appendix.
\end{proof}

\Wu{Lemma \ref{theo:RV_Properties}~$b)$ and Lemma \ref{theo:RegularlyVaryingConvex} \Wuu{yield} immediately
	that the auxiliary 
	price functions $c_a(x)=x\tau'_a(x)+\tau_a(x)$
	are again regularly varying, provided
	that the price functions $\tau_a(x)$ are
	regularly varying and convex. Moreover,
	\Wuu{the marginal games in an} asymptotic decomposition
	will be asymptotically well designed in \Wuu{this}
	case, since \eqref{eq:ImportantCondiAD} \Wuu{holds.}
	We summarize \Wuu{this} result in
Theorem~\ref{theo:MainTheoremRegularVariation}.}
\begin{thm}\label{theo:MainTheoremRegularVariation}
	Consider an NCG with price functions $\tau_a(\cdot)$
	satisfying \Wu{all the conditions of} Lemma \ref{theo:RegularlyVaryingConvex}.
	\Wu{If} the price functions
	are mutually comparable, i.e., for each 
	$a,b\in A,$ the limit
	\[
	\lim_{x\to \infty} \frac{\tau_a(x)}{\tau_b(x)}
	=q_{a,b}\in [0,\infty]
	\]
	exists for some constant $q_{a,b},$ then
	$\Gamma$ is 
	asymptotically well designed.
\end{thm}
\begin{proof}
	See the Appendix.
\end{proof}

\Wuu{Theorem \ref{theo:MainTheoremRegularVariation} generalizes
	Corollary 4.8 in \cite{Colini2017b}. Colini et al. \cite{Colini2017b}
	showed that NCGs with regularly varying
	and mutually comparable price functions $\tau_a(\cdot)$
	are gaugeable w.r.t. each user volume vector sequence
	$\{d^{(n)}\}_{n\in \N}$ such that
	$\lim_{n\to\infty}T(d^{(n)})=\infty$
	and $\varliminf_{n\to\infty}
	\tfrac{d_k^{(n)}}{T(d^{(n)})}>0$
	for each $k\in \K$.
	Obviously, this \Wuu{is only a partial result,} due to the
	sequence-specific \Wuu{condition on $\{d^{(n)}\}_{n\in \N}$ in Colini et al. \cite{Colini2017b}}. 
	With Lemma~\ref{theo:Gaugeable<=Scalable},
	such games are obviously scalable w.r.t. these specific
	user volume sequences. With \Wuu{the} asymptotic decomposition,
	Theorem~\ref{theo:MainTheoremRegularVariation}
	now successfully removes the sequence limitation and
	gives a global convergence of the PoA for such games,
	\Wuu{when} the
	price functions \Wuu{$\tau_a(\cdot)$ are convex.}
	}

\Wu{
	The condition that all $\tau_a(\cdot)$
	are mutually comparable implies a \Wuu{suitable} ordering on
	the $\tau_a(\cdot),$
	see the proof of \Wuu{Theorem} \ref{theo:MainTheoremRegularVariation}
	for details on the ordering.
	By Lemma~\ref{theo:RV_Properties}~$b),$
	this ordering also carries over to the auxiliary
	price functions $c_a(\cdot).$
	Hence, the proof of \Wuu{Theorem} \ref{theo:MainTheoremRegularVariation}
	directly defines the ordering on
	the resources $a\in A.$ With this ordering, we can then
	\Wuu{compare the cost} and \Wuu{construct scaling factors} at each inductive step.
	}

\Wu{Convexity} is only \Wuu{needed} to guarantee
\eqref{eq:ImportantCondiAD}. \Wu{The} proof of Theorem~\ref{theo:MainTheoremRegularVariation} \Wu{is
	still valid}
if \Wu{we use}
\eqref{eq:ImportantCondiAD} instead
of convexity. \Wu{Therefore,}
Theorem \ref{theo:MainTheoremRegularVariation} \Wu{can
	be further extended}. For instance, 
if \Wu{we substitute convexity by concavity in
	Theorem~\ref{theo:MainTheoremRegularVariation},}
then the \Wuu{conclusion} still holds.

Moreover, \Wu{we can even 
	weaken the condition that} \Wu{all} 
price functions $\tau_a(\cdot)$ are 
regularly varying and mutually comparable. 
Actually, similar arguments may still apply
\Wuu{when} only some of the
price functions $\tau_a(\cdot)$
are regularly varying and mutually comparable.
If this is
the case, we \Wuu{need} that
for each subset $\K'\subseteq \K,$ there exists a
regularly varying function $g(\cdot)$ such that:
\begin{itemize}
	\item[$G1')$] For each $a\in A,$ the limit
	\[
	\lim_{x\to\infty}\frac{\tau_a(x)}{g(x)}=q_a\in [0,\infty].
	\]
	\item[$G2')$] For each $k\in \K',$ there exists an strategy $s\in \S_k$ such that
	\[
	\max\{q_a: a\in A\quad\text{and}\quad r(a,s)>0\}<\infty.
	\]
	\item[$G3')$] There exists a group $k\in \K'$ such that
	\[
	\min_{s\in \S_k}\max\{q_a: a\in A\quad\text{and}\quad r(a,s)>0\}\in (0,\infty).
	\]
\end{itemize} 
\Wu{By additionally assuming condition \eqref{eq:ImportantCondiAD},
	one can again \Wuu{obtain} a proof for \Wuu{this}
	case by a slight \Wuu{modification of} the proof of Theorem~\ref{theo:MainTheoremRegularVariation}.
	Theorem \ref{theo:MainTheoremGaugeableExtension}
	below summarizes this result.}
\begin{thm}\label{theo:MainTheoremGaugeableExtension}
	Consider an NCG $\Gamma$ with non-decreasing,
	non-negative and differentiable price functions
	$\tau_a(\cdot).$ If \eqref{eq:ImportantCondiAD}
	holds for each $a\in A,$ and \Wuu{if} for each non-empty subset $\K'\subset \K$ there exists
	a regularly varying function $g(\cdot)$
	fulfilling $G1')$-$G3'),$  then
	$\Gamma$ is asymptotically decomposable, and thus
	asymptotically well designed.
\end{thm}
\begin{proof}
	See the appendix.
\end{proof}

\Wu{Conditions $G1')$-$G3')$ \Wuu{correspond} to $G1)$-$G3)$
in Definition \ref{def:GaugeableGame} for gaugeability.
However, they are more flexible and general than gaugeability.
They can apply to arbitrary 
user volume \Wuu{sequences}, while gaugeability
only applies to user volume sequences fulfilling condition \eqref{eq:GaugeableCondition}.
Hence, Theorem \ref{theo:MainTheoremGaugeableExtension} is not only an extension
\Wuu{of} Theorem~\ref{theo:SelfishMainTh}, but also \Wuu{of}
the \Wuu{main result \Wuu{about} gaugeability (Theorem 4.4)  in} Colini et al. \cite{Colini2017b}.
} 

\Wu{The proof of Theorem \ref{theo:MainTheoremGaugeableExtension}
	\Wuu{does} not \Wuu{need an} ordering \Wuu{of} the price functions
	$\tau_a(\cdot).$ Conditions \Wuu{$G1')$-$G3')$}
	already imply the existence of suitable scaling factors at each
	inductive step in the asymptotic decomposition.
	With \Wuu{these} scaling factors, we can then \Wuu{compare the cost} of marginals.}

\Wu{Therefore,
condition \eqref{eq:ImportantCondiAD} and 
the existence of suitable scaling factors
$g_n^{(l)}$ at each inductive step are  pivotal when
we apply \Wuu{the} asymptotic decomposition.
The ordering \Wuu{of}
price functions is only \Wuu{needed} for constructing
scaling factors and comparing the cost of
marginals at each inductive step. If the existence of
suitable scaling factors at each inductive step can be guarenteed
\Wuu{in advance}, 
then we do not need such \Wuu{an} ordering
\Wuu{of} the price functions.
Note that then the cost can be \Wuu{compared by} comparing
the scaling factors.}

\Wu{Actually, condition \eqref{eq:ImportantCondiAD} may not
	be so restrictive for regularly varying functions in practice.
	By Karamata's Representation Theorem
	for slowly varying functions \cite{Bingham1987Regular},
	each regularly varying function $\tau(\cdot)$ can be
	written as
	\[
	\tau(x)=x^{\rho}\cdot \exp\big(\eta(x)+\int_{b}^{x}\frac{\epsilon(t)}{t} dt\big),
	\]
	where $\lim_{x\to\infty}\eta(x)=\kappa\in \mathbb{R},$
	$\lim_{x\to\infty}\epsilon(x)=0$ and $b\ge 0$ is
	a constant dependent of function $\tau(\cdot).$ \Wuu{If} $\tau(x)$ is differentiable and non-decreasing, \Wuu{then} we can assume w.l.o.g. that
	the bounded function $\eta(x)$ is differentiable and
	$\epsilon(x)$ is continuous. Then,
	\[
	\lim_{x\to\infty}\frac{x\tau'(x)}{\tau(x)}
	=\rho+\lim_{x\to\infty} x\eta'(x),
	\]
	which \Wuu{is} a non-negative constant when
	the limit $\lim_{x\to\infty}x\eta'(x)$
	exists. Note that the limit $\lim_{x\to\infty}x\eta'(x)$
	\Wuu{exists} if $\eta(x)$ converges to $\kappa$
	eventually in a relatively \Wuu{steady} way.
	Such $\eta(\cdot)$ may possess some regular properties, e.g.,
	convexity, concavity, monotonicity, and others.
	
	\Wuu{Note that} \Wuu{there are differentiable,
	non-decreasing and regularly varying functions}
	$\tau(x)$ that do not fulfill \eqref{eq:ImportantCondiAD}.
	For instance, consider $\tau(x)=x^3\cdot \exp\big(1-\tfrac{1}{x}\sin (x)\big)$
	for \Wuu{large enough $x$.} \Wuu{Here,} $\eta(x)=1-\frac{\sin x}{x}$ and the limit
	$\lim_{x\to \infty} x\eta'(x)$ does not
	exist, since 
	$
	x\eta'(x)=\tfrac{1}{x}\sin x-\cos x
	$
	diverges as $x\to \infty.$
	Although such price functions are meaningful in theory, they may not be of interest in practice
	because of their irregular properties.
 	}

Theorem~\ref{theo:SelfishMainTh},
Theorem~\ref{theo:MainTheoremRegularVariation}
and Theorem~\ref{theo:MainTheoremGaugeableExtension} 
further \Wuu{demonstrate} the power of
scalability stemming from \cite{Wu2017Selfishness}. 
They together
\Wuu{deepen} our knowledge \Wu{on} user behavior
in NCGs. In particular, they
positively support
the view of \cite{Wu2017Selfishness}
that selfish \Wuu{user behavior} need not be bad.

\section{Conclusion}\label{sec:Conclusion}

\Wu{We \Wuu{have} unified recent results from
	\cite{Colini2016}, \cite{Colini2017a},
	\cite{Colini2017b} and \cite{Wu2017Selfishness}
	on the convergence
	of the PoA for NCGs. We \Wuu{have} reformulated
	the concept of limit games that are
	implicitly used in \cite{Wu2017Selfishness}.
	With the concept of
	limit games, we were able to bring
	new insight \Wuu{into} AWDGs,
	see, e.g., Theorem \ref{theo:NegativeConditionforAWDG}
	and Corollary \ref{theo:DirectSum}.
	We \Wuu{have} deepened
	the knowledge of  scalability introduced 
	by Wu et al. \cite{Wu2017Selfishness}.
	We \Wuu{have} introduced the technique of asymptotic decomposition
	that allows us to \Wuu{analyze} the convergence of the PoA
	for NCGs with general price functions.
	}

\Wu{With \Wuu{this new technique},
we were able to show that NCGs with arbitrary polynomial
price functions $\tau_a(\cdot)$ are asymptotically
well designed, see Theorem \ref{theo:SelfishMainTh}.}
This \Wu{completes} the results from
\cite{Colini2016}, \cite{Colini2017a}, \cite{Colini2017b}
and \cite{Wu2017Selfishness} for NCGs with polynomial price functions. 
\Wu{Moreover, we were able to
	apply \Wuu{the} asymptotic decomposition to NCGs
	with regularly varying price function,
	and \Wuu{prove} that these NCGs 
	are also asymptotically well designed,
	see Theorem \ref{theo:MainTheoremRegularVariation}
	and Theorem \ref{theo:MainTheoremGaugeableExtension}.
	Our results definitely demonstrate
	the power of scalablity, and
	positively support the view of \cite{Wu2017Selfishness}
	on user behavior in NCGs.}

\Wuu{The profit goals of users in an NCG are in general inconsistent with
the profit goal of the underlying central authority.
Both of them want to minimize cost.} \Wuu{But the difference is that} users \Wu{only} \Wuu{locally minimize}
their own \Wuu{cost}, \Wu{while}
the central authority \Wu{wants to} \Wuu{globally minimize} social cost. \Wu{Our results show that
the local minimization will lead to a
glocal minimization when the volume of users
becomes large. \Wu{Thus, selfishness need not be bad in
general.}
}

\Wu{
	\Wuu{Future work in this direction could} consider
	the {\em saturation point} of an NCG, i.e., 
	a threshold value for user volumes, beyond which
	NE profiles will almost be
	SO profiles. This could be very interesting
	in game theory. NCGs are {\em open games},
	i.e., users can freely join such games.
	When the user volume reaches its saturation,
	what the users need to do is to perform as selfish
	as possible, since this is the best choice in
	a bad environment. 
}

\section*{Acknowledgement}

\Wu{The authors would like to thank \Wuu{Dr. Colini-Baldeschi, Dr. Cominetti, Dr. Mertikopoulos, and Dr. Scarsini.}
	Their results on gaugeability \cite{Colini2017b} \Wuu{inspired} us greatly and have led to current article. The \Wuu{authors} would also like
	to thank Dr. Chen from the Department of Urban Transportation
	at Beijing University of Technology, who has offered us
	traffic data from Beijing \Wuu{for computation}. \Wuu{Based on} that data, the first two
	authors have computed the PoA
	and \Wuu{the results}
	motivated them to \Wuu{write} \cite{Wu2017Selfishness} and 
	the current article.}

\section*{Appendix}
\subsection*{Proof of Lemma~\ref{theo:LimitGame_Wu2017}}
\begin{proof}[Proof of Lemma~\ref{theo:LimitGame_Wu2017}]
	Consider an NCG
	\[
	\Gamma =	\Big(\K, A, \S, (r(a,s))_{a\in A, s\in \S}, (\tau_a)_{a\in A},d\Big),
	\]
	and an \Wu{arbitrary} user volume vector sequence
	$\{d^{(n)}\}_{n\in \N}$ such that
	the total volume $T(d^{(n)})\to\infty$
	as $n\to \infty.$ Let $\tilde{f}^{(n)}$ be
	an NE profile w.r.t. user volume vector
	$d^{(n)}=(d_k^{(n)})_{k\in \K},$ for 
	each $n\in \N.$ We assume that $\Gamma$
	has a limit
	\[
	\Gamma^{\infty} =	\Big(\K^{\infty}, A, \S^{\infty}, (r(a,s))_{a\in A, s\in \S^{\infty}}, (\tau_a^{\infty})_{a\in A},\d\Big)
	\]
	w.r.t. $\{d^{(n)}\}_{n\in \N}.$
	
	By Definition \ref{def:limitGame}, we obtain that there is an infinite subsequence $\{n_i\}_{i\in \N}$ and a sequence 
	$\{g_i\}_{i\in \N}$ of positive scaling factors, s.t., conditions \Wu{$L1)$-$L4)$} hold. 
	Note that we can further assume that:
	\begin{itemize}
		\item The limit
		\[
		\lim_{i\to \infty} \frac{\tilde{f}^{(n_i)}_s}{T(d^{(n)})}
		=\tilde{\f}^{\infty}_s \in [0,1]
		\]
		exists, for some constant
		$\tilde{\f}^{\infty}_s,$ for each 
		$s\in \S.$
	\end{itemize}
	Otherwsie, we can take an infinite subsequence
	$\{n_{i_j}\}_{j\in \N}$ fulfilling the above
	condition. Let 
	\[
	\tilde{\f}_a^{\infty}:=\sum_{s\in \S} r(a,s)\cdot\tilde{\f}_s^{\infty}
	\]
	for each $a\in A.$ We aim to prove
	Lemma \ref{theo:LimitGame_Wu2017} 
	with the subsequence $\{n_i\}_{i\in \N},$
	scaling factor sequence $\{g_i\}_{i\in \N},$
	and $\tilde{\f}^{\infty}:=(\tilde{\f}^{\infty}_s)_{s\in \S^{\infty}}.$
	
	To that end, we need some auxiliary facts.
	Fact \ref{theo:StrategyPriceConvergence} below
	indicates that prices of tight strategies
	are always well ``preserved" in the limit w.r.t.
	any sequence of feasible strategy profiles. Therefore,
	for each tight strategy $s\in \S^{\infty},$ we obtain
	that
	\[
	\lim_{i\to \infty}\frac{\tau_s(\tilde{f}^{(n_i)})}{g_i}
	=\sum_{a\in A} r(a,s)\cdot \tau_a(\tilde{\f}_a^{\infty})<\infty.
	\]
	\begin{fact}\label{theo:StrategyPriceConvergence}
		Consider a tight strategy $s\in \S^{\infty}.$
		Let $f^{(n_i)}$ be a feasible strategy profile
		w.r.t. user volume vector $d^{(n_i)},$
		for each $i\in \N.$ If
		\[
		\lim_{i\to \infty}\frac{f_{s'}^{(n_i)}}{T(d^{(n_i)})}
		=\boldsymbol{\mu}_{s'} 
		\]
		for some constant $\boldsymbol{\mu}_{s'}\in [0,1],$
		for each ${s'}\in \S,$ then
		\[
		\tau_s^{\infty}(\boldsymbol{\mu})=\sum_{a\in A} r(a,s)\cdot \tau_a^{\infty}
		(\boldsymbol{\mu}_{a})=
		\lim_{i\to\infty}
		\frac{\sum_{a\in A} r(a,s)\cdot\tau_a(f^{(n_i)}_a)}{g_i},
		\]
		where $\boldsymbol{\mu}_a=\sum_{s'\in \S} 
		r(a,s') \boldsymbol{\mu}_{s'}$
		for each $a\in A.$
	\end{fact}
	\begin{proof}[Proof of Fact \ref{theo:StrategyPriceConvergence}]
		Note that for each $a\in A,$ 
		\[
		\begin{split}
		\boldsymbol{\mu}_a&=\sum_{s'\in \S} 
		r(a,s') \boldsymbol{\mu}_{s'}=
		\sum_{s'\in \S} r(a,s')\cdot \lim_{i\to\infty} \frac{f^{(n_i)}_{s'}}{T(d^{(n_i)})}\\
		&=\lim_{i\to\infty}
		\sum_{s'\in \S} r(a,s') \frac{f^{(n_i)}_{s'}}{T(d^{(n_i)})}
		=\lim_{i\to\infty} \frac{f_a^{(n_i)}}{T(d^{(n_i)})}
		\in [0,\infty).
		\end{split}
		\]
	Therefore, for an arbitrarily fixed $\epsilon>0,$ 
	\[
	\max\big\{0,\boldsymbol{\mu}_a-\epsilon\big\}\le \frac{f_a^{(n_i)}}{T(d^{(n_i)})} \le \boldsymbol{\mu}_a+\epsilon,
	\]
	for each resource $a\in A,$ for $i$ large enough.
	As a result, for $i$ large enough, we obtain for
	each $a\in A$ that
	\begin{equation}\label{eq:Lemma1_eq1}
	\tau_a\Big(T(d^{(n_i)})
	\cdot\max\big\{0,\boldsymbol{\mu}_a-\epsilon\big\}\Big)
	\le \tau_a(f^{(n_i)}_a)
	\le \tau_a\Big(T(d^{(n_i)})\cdot\big(\boldsymbol{\mu}_a+\epsilon\big)\Big),
	\end{equation}
	since each price function $\tau_a(\cdot)$ is non-decreasing.
	
	Since $s$ is tight, by $L3.2),$ we then obtain for each
	resource $a\in A$ with $r(a,s)>0$ that
	$\tau_a^{\infty}(\cdot)$ is a continuous and non-decreasing
	function,
	and that
	\[
	\tau_a^{\infty}(\boldsymbol{\mu}_a+\epsilon)
	=\lim_{i\to\infty} 
	\frac{\tau_a\Big(T(d^{(n_i)})\cdot\big(\boldsymbol{\mu}_a+\epsilon\big)\Big)}{g_i}\ge \varlimsup_{i\to\infty}
	\frac{\tau_a(f_a^{(n_i)})}{g_i}\ge\varliminf_{i\to\infty}
	\frac{\tau_a(f_a^{(n_i)})}{g_i} \ge 
	\]
	\[
	\tau_a^{\infty}\Big(\max\big\{0,\boldsymbol{\mu}_a-\epsilon\big\}\Big)
	=\lim_{i\to\infty} 
	\frac{\tau_a\Big(T(d^{(n_i)})\cdot\max\big\{0,\boldsymbol{\mu}_a-\epsilon\big\}\Big)}{g_i}.
	\]
	By \eqref{eq:Lemma1_eq1}, the continuity of $\tau_a^{\infty}(\cdot),$
	and the \Wuu{arbitrary choice} of $\epsilon,$ we then obtain
	for each $a\in A$ with $r(a,s)>0$ that
	\[
	\tau_a^{\infty}(\boldsymbol{\mu}_a)=
	\lim_{i\to\infty}
	\frac{\tau_a(f^{(n_i)})}{g_i},
	\]
	which, in turn, implies \Wu{Fact \ref{theo:StrategyPriceConvergence}.}
	\end{proof}
	Fact \ref{theo:StrategyPriceLimit_NonTight} below considers 
	the limit prices of non-tight strategies w.r.t. 
	NE profiles. It indicates that we can completely
	ignore those non-tight strategies in the limit analysis.
	By Fact \ref{theo:StrategyPriceLimit_NonTight}, we obtain that
	\[
	\tilde{\f}_a^{\infty}=\sum_{s\in \S} r(a,s)\cdot\tilde{\f}^{\infty}_s
	=\sum_{s\in \S^{\infty}} r(a,s)\cdot\tilde{\f}^{\infty}_s,
	\]
	and that the limit volume
	\[
	\d_k=\sum_{s\in \S_k\cap\S^{\infty}} \tilde{\f}_s^{\infty}
	\]
	for each group \Wu{$k\in \K^{\infty},$} and that
	\[
	\d_k=\sum_{s\in \S_k} \tilde{\f}_s^{\infty}=0,
	\]
	for each group $k\in \K\backslash\K^{\infty}$
	(since such groups do not have tight strategies).
	Therefore, $\tilde{\f}^{\infty}=(\tilde{\f}_s^{\infty})_{s\in \S^{\infty}}$ is a strategy profile 
	of the limit game $\Gamma^{\infty}$ w.r.t.
	user volume vector $(\d_k)_{k\in \K^{\infty}}$
	and $\K^{\infty}.$
	\begin{fact}\label{theo:StrategyPriceLimit_NonTight}
		Consider a non-tight strategy $s\in \S\backslash\S^{\infty}.$
		Then, we obtain that:
		\begin{itemize}
			\item[$1)$] $\tilde{\f}_s^{\infty}=0.$
			\item[$2)$] $
			\lim_{i\to \infty}
			\frac{\tilde{f}^{(n_i)}_s\cdot\tau_s(\tilde{f}^{(n_i)})}{
				T(d^{(n_i)})\cdot g_i}=0.
			$
	    \end{itemize}
	\end{fact}
	\begin{proof}[Proof of Fact \ref{theo:StrategyPriceLimit_NonTight}]
		Before we start the proof, let us recall a basic property
		of non-tight strategies.
		By $L2)$ and $L3)$ of Definition \ref{def:limitGame}, 
		we obtain that 
		\[
		\lim_{i\to\infty}
		\frac{\tau_a(T(d^{(n_i)})x)}{g_i}=\tau_a^{\infty}(x)\equiv \infty,\quad \forall x>0,
		\]
		for some $a\in A$ with $r(a,s)>0,$
		since $s$ is non-tight.
		 
		\textbf{For $1):$} Let $k\in \K$ be the group
		such that $s\in \S_k.$ We prove $1)$ by \Wuu{contradiction.} 
		We suppose that $\tilde{\f}_s^{\infty}>0.$ Then, 
		by $L2)$ of Definition \ref{def:limitGame}, 
		with \Wu{an argument similar to that for} 
		 Lemma~\ref{theo:StrategyPriceConvergence},
		we obtain that
		\[
		\lim_{i\to\infty}
		\frac{\tau_s(\tilde{f}^{(n_i)})}{g_i}
		=\infty,\quad \text{and}\quad
		\lim_{i\to \infty}\frac{\tilde{f}^{(n_i)}\cdot\tau_s(\tilde{f}^{(n_i)})}{T(d^{(n_i)})\cdot g_i}
		=\infty,
		\]
		since $\tilde{\f}_s^{\infty}>0.$
		Thus, by $L3.1)$ of Definition
		\ref{def:limitGame}, group $k$ is non-negligible.
		Hence, group $k$ must have tight strategies.
		Let $s'\in \S_k\cap\S^{\infty}$ be a tight 
		strategy. By Fact \ref{theo:StrategyPriceConvergence},
		we obtain that
		\[
		\lim_{i\to \infty}\frac{\tau_{s'}(\tilde{f}^{(n_i)})}{g_i}
		<\infty = \lim_{i\to\infty}
		\frac{\tau_s(\tilde{f}^{(n_i)})}{g_i}.
		\]
		Hence, for $i$ large enough, 
		\[
		\tau_{s'}(\tilde{f}^{(n_i)})<\tau_s(\tilde{f}^{(n_i)}).
		\]
		This, in turn, implies that 
		$\tilde{f}^{(n_i)}_{s}\equiv 0$
		for $i$ large enough, due to the user optimality
		\eqref{def:WE} of NE profiles. This contradicts 
		\Wuu{with the assumption that $\tilde{\f}_s^{\infty}>0$.}
		
		Therefore, $\tilde{\f}_s^{\infty}=0$ must hold.
		
		\textbf{For $2):$} We prove $2)$ again by \Wuu{contradiction}.
		We assume that 
		\begin{equation}\label{eq:Fact2_1}
		\varlimsup_{i\to \infty}
		\frac{\tilde{f}^{(n_i)}_s\cdot\tau_s(\tilde{f}^{(n_i)})}{
			T(d^{(n_i)})\cdot g_i}>0.
		\end{equation}
		By $L3.1)$ of
		Definition \ref{def:limitGame}, we obtain
		immediately that $s$ must be a strategy from
		some non-negligible group.
		We assume, w.l.o.g., that
		\[
		\lim_{i\to \infty}
		\frac{\tilde{f}^{(n_i)}_s\cdot\tau_s(\tilde{f}^{(n_i)})}{
			T(d^{(n_i)})\cdot g_i}=\varlimsup_{i\to \infty}
		\frac{\tilde{f}^{(n_i)}_s\cdot\tau_s(\tilde{f}^{(n_i)})}{
			T(d^{(n_i)})\cdot g_i}>0.
		\]
		Otherwise, one can take an infinite subsequence
		of $\{n_i\}_{i\in \N}$ fulfilling the above condition.
		
		By $1),$ we can obtain that
		\[
		\lim_{i\to\infty} \frac{\tilde{f}_s^{(n_i)}}{T(d^{(n_i)})}=\tilde{\f}_s^{\infty}=0,
		\]
		and thus
		\[
		\lim_{i\to\infty}
		\frac{\tau_s(\tilde{f}^{(n_i)})}{g_i}=\infty.
		\]
		Then, with \Wuu{an argument argument similar as that in the proof of}
		$1),$ one can prove again that 
		$\tilde{f}_s^{(n_i)}\equiv 0$ for 
		$i$ large enough. This \Wuu{yields that}
		\[
		\frac{\tilde{f}^{(n_i)}_s\cdot\tau_s(\tilde{f}^{(n_i)})}{
			T(d^{(n_i)})\cdot g_i}\equiv 0,
		\]
		for $i$ large enough. \Wuu{This obviously contradicts with the
		assumption \eqref{eq:Fact2_1}.}
		Hence, $2)$ must hold.
	\end{proof}
	
	By Fact \ref{theo:StrategyPriceConvergence} and
	Fact \ref{theo:StrategyPriceLimit_NonTight},
	we obtain immediately that 
	\[
	\lim_{i\to \infty}
	\frac{C(\tilde{f}^{(n_i)})}{g_i}
	=\lim_{i\to\infty} \frac{\sum_{s\in \S^{\infty}}\tilde{f}_s^{(n_i)}
		\cdot \tau_s (\tilde{f}^{(n_i)})}{T(d^{(n_i)})\cdot g_i}=\sum_{s\in \S^{\infty}}\tilde{\f}_s^{\infty}
	\cdot \tau_s^{\infty} (\tilde{\f}^{\infty}).
	\]
	\Wuu{The fact that $\tilde{\f}^{\infty}$
	is an NE profile of $\Gamma^{\infty}$ 
	follows} immediately from that limit preserves
	numerical ordering ``$\ge$". 
	
	Combining all of the above, the proof of Lemma \ref{theo:LimitGame_Wu2017}
	is completed.
\end{proof}

\subsection*{Proof of Theorem \ref{theo:NegativeConditionforAWDG}}
\begin{proof}[Proof of Theorem \ref{theo:NegativeConditionforAWDG}]
	Consider an NCG
	\[
	\Gamma =	\Big(\K, A, \S, (r(a,s))_{a\in A, s\in \S}, (\tau_a)_{a\in A},d\Big),
	\]
	and an \Wu{arbitrary} user volume vector sequence
	$\{d^{(n)}\}_{n\in \N}$ such that
	the total volume $T(d^{(n)})\to\infty$
	as $n\to \infty.$
	
	We assume that the game
	\[
	\Gamma^{\infty} =	\Big(\K^{\infty}, A, \S^{\infty}, (r(a,s))_{a\in A, s\in \S^{\infty}}, (\tau_a^{\infty})_{a\in A},\d\Big)
	\]
	is the limit game of $\Gamma$ w.r.t. a subsequence 
	$\{n_i\}_{i\in \N}$ and a scaling factor sequence
	$\{g_i\}_{i\in \N},$ where all the components of
	$\Gamma^{\infty}$ are defined as in Definition~\ref{def:limitGame}.
	We suppose that $\Gamma^{\infty}$ is not well designed.
	We aim to show in this case that the
	PoA$(d^{(n_i)})$ does not converge to $1$ as
	$i\to\infty,$ which, in turn, implies that
	the game $\Gamma$ is not asymptotically well designed.
	
	Let $\{\tilde{f}^{(n_i)}\}_{i\in \N}$ and
	$\{f^{*(n_i)}\}_{i\in \N}$ be an NE profile sequence
	and an SO profile sequence w.r.t.
	$\{d^{(n_i)}\}_{i\in \N},$ respectively.
	We assume, w.l.o.g., \Wuu{that}
	\begin{itemize}
		\item \Wuu{the} limits
		\[
		\lim_{i\to\infty}
		\frac{\tilde{f}_s^{(n_i)}}{T(d^{(n_i)})}
		=\tilde{\f}^{\infty}_s\in [0,1],\quad\text{and}
		\quad 
		\lim_{i\to \infty}
		\frac{f^{*(n_i)}_s}{T(d^{(n_i)})}
		=\f^{*,\infty}_s\in [0,1]
		\]
	\end{itemize}
	\Wuu{exist} for some constants $\tilde{\f}^{\infty}_s,\f^{*,\infty}_s,$
	for each $s\in \S.$ Otherwise, we can take an infinite
	subsequence of $\{n_i\}_{i\in \N}$ fulfilling
	this condition.
	
	By Fact \ref{theo:StrategyPriceConvergence}
	and Fact \ref{theo:StrategyPriceLimit_NonTight}
	in the proof of Lemma \ref{theo:LimitGame_Wu2017},
	we obtain immediately that
	\[
	\lim_{i\to \infty} \frac{C(\tilde{f}^{(n_i)})}{g_i}
	=\sum_{s\in \S^{\infty}} \tilde{\f}^{\infty}_s
	\cdot \tau_s^{\infty}(\tilde{\f}^{\infty})\in (0,\infty),
	\]
	\[
	\lim_{i\to \infty} \frac{\sum_{s\in \S^{\infty}}f^{*((n_i))}_s\cdot
			\tau_s(f^{*(n_i)})}{T(d^{(n_i)})\cdot g_i}
	=\sum_{s\in \S^{\infty}} \f^{*,\infty}_s
	\cdot \tau_s^{\infty}(\f^{*,\infty})\in [0,\infty),
	\]
	and that $\tilde{\f}^{\infty}=(\tilde{\f}^{\infty}_s)_{s\in \S^{\infty}}$ is an NE profile of $\Gamma^{\infty}.$
	
	We now aim to show that Fact \ref{theo:StrategyPriceLimit_NonTight} applies also
	to SO profiles.
	
	By the fact that each $f^{*(n_i)}$ is an SO profile,
	we then obtain immediately that
	\[
	\varlimsup_{i\to \infty}
	\frac{C(f^{*(n_i)})}{g_i}
	\le \lim_{i\to \infty} \frac{C(\tilde{f}^{(n_i)})}{g_i}
	<\infty,
	\]
	which, in turn, implies that for each non-tight
	strategy $s\in \S\backslash\S^{\infty},$
		$\f^{*,\infty}_s=0.$ 
    Otherwise, if $\f^{*,\infty}_s>0,$
    then by $L3)$ of Defition
    \ref{def:limitGame}, $\varlimsup_{i\to \infty}
    \frac{C(f^{*(n_i)})}{g_i}=\infty.$ 
    
    Therefore, $\f^{*,\infty}=(\f^{*,\infty}_s)_{s\in \S^{\infty}}$ is also
    a feasible strategy profile of $\Gamma^{\infty}.$
    
    We now aim to prove for
    each non-tight strategy $s\in \S\backslash\S^\infty$
    that
  \[
    \lim_{i\to \infty}
    \frac{f^{*(n_i)}_s\cdot\tau_s(f^{*(n_i)})}{
    	T(d^{(n_i)})\cdot g_i}=0.
  \]
  Similarly, we prove this by \Wuu{contradiction.}
  We assume, w.l.o.g., that there is exactly one non-tight
  strategy, i.e., $|\S\backslash\S^{\infty}|=1.$
  Let $s$ denote this unique non-tight strategy. 
  For the case of more non-tight strategies,
  an almost identical argument will apply.
  
  We assume now, w.l.o.g., that
  \begin{equation}\label{eq:Theorem1_Contra}
    \lim_{i\to \infty}
    \frac{f^{*(n_i)}_s\cdot\tau_s(f^{*(n_i)})}{
    	T(d^{(n_i)})\cdot g_i}>0.
    \end{equation}
   Then, by $L3.1)$ of Definition
   \ref{def:limitGame}, $s$ must be a strategy from
   some non-negligible group $k\in \K$.
   By $L3.2)$ of Definition \ref{def:limitGame},
   there must exist a tight strategy $s'\in \S_k.$
   To derive a
   contradiction to \eqref{eq:Theorem1_Contra},
   we now construct some artificial
   feasible strategy profiles. For each $i\in \N,$ we put
   \begin{equation*}
   	h_{s''}^{(n_i)}=
   	\begin{cases}
   	f^{*(n_i)}_{s''}&\text{if } s''\notin \{s,s'\},\\
   	f^{*(n_i)}_{s''}+f^{*(n_i)}_s&\text{if } s''=s',\\
   	0&\text{if }s''=s.
   	\end{cases}
   \end{equation*}
   For the case of more than one non-tight strategies,
   one can similarly move the users adopting non-tight strategies to tight strategies. However, the \Wuu{explicit}
   definition of profiles $h^{(n_i)}=(h_{s''}^{(n_i)})_{s''\in \S}$
   will become very \Wuu{complicated.}

   Obviously, for each $s''\in\S,$
   \[
   \lim_{i\to \infty} \frac{h_{s''}^{(n_i)}}{T(d^{(n_i)})}
   =\f^{*,\infty}_{s''},
   \]
   since $\f_s^{*,\infty}=0.$ Moreover,
   by Fact \ref{theo:StrategyPriceConvergence} of
   Lemma \ref{theo:LimitGame_Wu2017}, we obtain that
   \[
   \begin{split}
   \lim_{i\to \infty} &\frac{\sum_{s''\in \S^{\infty}}h^{(n_i)}_{s''}\cdot
   	\tau_{s''}(h^{(n_i)})}{T(d^{(n_i)})g_i}
   =
   \lim_{i\to \infty} \frac{\sum_{s''\in \S^{\infty}}f^{*(n_i)}_{s''}\cdot
   	\tau_{s''}(f^{*(n_i)})}{T(d^{(n_i)})g_i}\\
   &=\sum_{s''\in \S^{\infty}} \f^{*,\infty}_{s''}
   \cdot \tau_{s''}^{\infty}(\f^{*,\infty}).
   \end{split}
   \]
   Note that in profiles $h^{(n_i)},$ only
   tight strategies are used. Thus, we obtain
   that
   \[
   \lim_{i\to \infty} \frac{C(h^{(n_i)})}{g_i}
   =\lim_{i\to \infty} \frac{\sum_{s''\in \S^{\infty}}f^{*(n_i)}_{s''}\cdot
   	\tau_{s''}(f^{*(n_i)})}{T(d^{(n_i)})\cdot g_i}.
   \]
   \Wuu{With} \eqref{eq:Theorem1_Contra}, we obtain
   further that
   \[
   \varliminf_{i\to \infty} \frac{C(f^{*(n_i)})}{g_i}
   > \lim_{i\to \infty} \frac{\sum_{s''\in \S^{\infty}}f^{*(n_i)}_{s''}\cdot
   	\tau_{s''}(f^{*(n_i)})}{T(d^{(n_i)})\cdot g_i}=\lim_{i\to \infty} \frac{C(h^{(n_i)})}{g_i}.
   \]
   \Wuu{This yields} that 
   \[
   C(h^{(n_i)})<C(f^{*(n_i)})
   \]
   for $i$ large enough. This contradicts \Wuu{with}
   the fact that profiles $f^{*(n_i)}$ are all
   system optimum.

	\Wuu{Therefore}, Fact \ref{theo:StrategyPriceLimit_NonTight}
	also applies to SO profiles. Thus we obtain that
	\[
	\lim_{i\to \infty} \frac{C(f^{*(n_i)})}{g_i}=\sum_{s\in \S^{\infty}} \f^{*,\infty}_s
	\cdot \tau_s^{\infty}(\f^{*,\infty}).
	\]
	\Wuu{We now aim to show that} $\f^{*,\infty}$
	is an SO profile of $\Gamma^{\infty}.$
	\Wuu{Let $\f=(\f_s)_{s\in \S^{\infty}}$
		be an arbitrary feasible strategy profile
		of $\Gamma^{\infty}.$ To show that $\f^{*,\infty}$
		is an SO profile of $\Gamma^{\infty}$, we only need to
		show that its cost is not larger than
		the cost of $\f,$ due to the arbitrary choice
		of $\f.$ Note that there must
		exist a sequence $\{f^{(n_i)}\}_{i\in \N}$
		of feasible profiles such that
		\[
		\f_s=\lim_{i\to \infty} \frac{f_s^{(n_i)}}{T(d^{(n_i)})}
		\]
		for each tight strategy $s\in \S^{\infty}.$
		Actually, we can put for each $k\in \K^{\infty}$ with $\d_k^{(n_i)}>0$
		and $s\in \S_k$ that
		\[
		f^{(n_i)}_s=\frac{d_k^{(n_i)}\cdot \f_s}{\d_k}\quad \forall i\in \N,
		\]
		and put for each other $k\in \K$
		and $s\in \S_k$ that
		\[
		f^{(n_i)}_s=\frac{d_k^{(n_i)}}{|\S_k|}\quad \forall i\in \N.
		\]
		Since each $f^{*(n_i)}$ is system optimal, 
		we obtain that
		\[
		\frac{C(f^{*(n_i)})}{g_i} \le \frac{C(f^{(n_i)})}{g_i}\quad
		\forall i\in \N.
		\]
		Letting $i\to\infty,$ this yields by Fact \ref{theo:StrategyPriceConvergence}, $L3)$ and the definition 
		of $\K^{\infty}$ that 
		\[
	\lim_{i\to\infty}\frac{C(f^{*(n_i)})}{g_i}=	\sum_{s\in \S^{\infty}} \f^{*,\infty}\cdot
		\tau_s^{\infty} (\f^{*,\infty})
		\le \sum_{s\in \S^{\infty}} \f_s\cdot \tau_s^{\infty}(\f)=\lim_{i\to \infty}\frac{C(f^{(n_i)})}{g_i},
		\]
		since $\f$ defines only
		on tight strategies, and since each group $k\in \K\backslash\K^{\infty}$ is negligible without tight strategy
		and thus must have zero limit volume, i.e., $\d_k=0.$
 		Hence, $\f^{*,\infty}$ is an SO profile of $\Gamma^{\infty}.$}
	
	Since $\Gamma^{\infty}$ is not well designed, we
	thus obtain that 
	\[
	\sum_{s\in \S^{\infty}} \f^{*,\infty}_s
	\cdot \tau_s^{\infty}(\f^{*,\infty})<\sum_{s\in \S^{\infty}} \tilde{\f}^{\infty}_s
	\cdot \tau_s^{\infty}(\tilde{\f}^{\infty}).
	\]
	This implies that
	\[
	\lim_{i\to \infty} \text{PoA}(d^{(n_i)})
	=\lim_{i\to\infty}
	\frac{C(\tilde{f}^{(n_i)})}{C(f^{*(n_i)})}
	=
	\frac{\sum_{s\in \S^{\infty}} \tilde{\f}^{\infty}_s
	\cdot \tau_s^{\infty}(\tilde{\f}^{\infty})}{\sum_{s\in \S^{\infty}} \f^{*,\infty}_s
	\cdot \tau_s^{\infty}(\f^{*,\infty})}>1.
	\]
	Therefore, $\Gamma$ is not well designed.
\end{proof}
\subsection*{Proof of Theorem \ref{theo:SelfishMainTh}}
\begin{proof}[Proof of Theorem \ref{theo:SelfishMainTh}]
	We prove Theorem \ref{theo:SelfishMainTh} by
	applying the idea of asymptotic decomposition.
	To well demonstrate the idea, we will give a
	very detailed proof. \Wu{The proof will be direct
		and elementary without using
		a heavy machinery. It \Wuu{only} uses the definition
		and connection between NE profiles
		and SO profiles, simple facts about
		the asymptotic notations $O(\cdot),\Omega(\cdot),
		\Theta(\cdot),o(\cdot)$ and $\omega(\cdot),$
		and a suitable induction along the user
		groups.}
	
	Let $\{d^{(n)}\}_{n\in \N}$ be
	an arbitrary sequence of user volume \Wu{vectors} such that:
	\begin{itemize}
		\item Each $d^{(n)}=(d^{(n)}_1,\ldots, d^{(n)}_K)$ 
		is a vector, where the $k$-th component
		$d_k^{(n)}$ represents the user volume of the $k$-th
		group for $k=1,\ldots,K,$ for each $n\in \N.$
		\item The total user volume $T(d^{(n)})=\sum_{k=1}^{K}d_k^{(n)}\to +\infty$ as
		$n\to +\infty.$
		\end{itemize}
		To prove the \Wu{Theorem \ref{theo:SelfishMainTh},} we only need to show that
		$\lim_{n\to \infty} \text{PoA}(d^{(n)})=1,$
		due to the \Wu{arbitrary choice} of the user
		volume vector \Wu{sequence $\{d^{(n)}\}_{n\in \N}.$ }
		
		Note that 
		\begin{equation}\label{eq:Global_Obj}
		\lim_{n\to \infty} \text{PoA}(d^{(n)})=1
		\end{equation}
		follows immediately \Wuu{from} the fact that
		\[
		\varlimsup_{n\to \infty}\text{PoA}(d^{(n)})=1,
		\]
		since it holds trivially that
		\[
		\varliminf_{n\to \infty}\text{PoA}(d^{(n)})\ge 1.
		\]
		
		Therefore, we assume, w.l.o.g., that the limit
		\begin{equation}\label{eq:Assum1}
		\lim_{n\to \infty} \text{PoA}(d^{(n)})\in [1, +\infty]
		\end{equation}
		exists. Otherwise, we can take an infinite subsequence $\{n_i\}_{i\in \N}$ such that the limit
		\[
		\lim_{i\to \infty} \text{PoA}(d^{(n_i)})=\varlimsup_{n\to \infty}\text{PoA}(d^{(n)})\in [1,+\infty]
		\]
		exists, and we can then restrict our discussion to this
		subsequence.
		
		With assumption \eqref{eq:Assum1}, 
		\eqref{eq:Global_Obj} follows immediately \Wu{from} the existence
		of an infinite subsequence $\{n_i\}_{i\in \N},$ s.t.,
		\[
		\lim_{i\to \infty} \text{PoA}(d^{(n_i)})=1.
		\] 
		\Wu{So, in the application of
		asymptotic decomposition,} we can take a series of nested infinite subsequences
		of the sequence $\{n\}_{n\in \N}.$  To simplify
		\Wu{notation,} we will not explicitly use the \Wu{terminology} of subsequences,
		but assume that the \Wu{user} volume vector sequence
		$\{d^{(n)}\}_{n\in \N}$ itself
		fulfills some required properties, if corresponding
		subsequences do exist.
		
		We now introduce some notations \Wu{stemming} from
		\cite{Colini2017a}. Let $\rho_{\alpha}$ be the degree
		of the polynomial $\tau_a(\cdot)$ for
		each resource $a\in A.$ Accordingly, we \Wu{can} define the {\em degree
			of a strategy} $s\in \S_k$ as
			\[
			\rho_s:=\max\{\rho_a:\ r(a,s)>0\text{ and }a\in A\},
			\]
			and the {\em degree of a group} $k\in \{1,\ldots,K\}$ as 
			\[
			\rho_k:=\min\{\rho_s:\ s\in \S_k\}.
			\]
			\Wu{Although these notations are trivial, they
				\Wuu{help us construct a suitable} ordering
				on the resource set $A$, and
				accordingly on the sets $\S$ and $\K$.
				We can compare two resources $a,b\in A$
				through their degrees $\rho_a,\rho_b.$
				This will be very helpful when we
				construct the scaling factors at each inductive
				step in the asymptotic decomposition.}
			
			\Wu{\Wuu{Moreover, we need to compare cost}
				at each inductive step during
				asymptotic decomposition. This
				requires basic knowledge
				on \Wuu{the asymptotic notation.}} Let $h(x)$ be a non-negative
			real-valued function. The big $O$ notation
			$O\big(h(x)\big)$ denotes the class \Wu{of}
			all non-negative real-valued functions $q(x)$ such that
			$\varlimsup_{x\to \infty} \frac{q(x)}{h(x)}<+\infty,$
			and the small $o$ notation $o\big(h(x)\big)$
			denotes the class of all non-negative real-valued functions $q(x)$ such that
			$\lim_{x\to \infty} \frac{q(x)}{h(x)}=0.$
			Similarly,  $\Omega\big(h(x)\big)$ 
			denotes the class of all non-negative real-valued functions
			$q(x)$ such 
			that $\varliminf_{x\to \infty}\frac{q(x)}{h(x)}>0,$
			and  $\omega\big(h(x)\big)$ denotes the class of all non-negative real-valued functions
			$q(x)$ such 
			that $\lim_{x\to \infty}\frac{q(x)}{h(x)}=+\infty.$
			We put $\Theta\big(h(x)\big)=O\big(h(x)\big)
			\cap \Omega\big(h(x)\big).$ In addition, we 
			write $h(x)\approx q(x)$ if 
			$\lim_{x\to \infty} \frac{h(x)}{q(x)}=1.$
			
			Let $f^{*(n)},\tilde{f}^{(n)}$ be
			an SO profile and an NE profile, respectively, w.r.t.
			user volume vector $d^{(n)}=(d_k^{(n)})_{k\in \K},$ for
			each $n\in \N$. Then, for each $n\in \N,$
			\begin{equation}\label{def:redefPoA}
			\text{PoA}(d^{(n)})=\frac{C(\tilde{f}^{(n)})}{C(f^{*(n)})}
			=\frac{\sum_{k\in \K}\sum_{s\in \S_k} \tilde{f}^{(n)}_s\cdot
				\tau_s(\tilde{f}^{(n)})}{\sum_{k\in \K}\sum_{s\in \S_k} f^{*(n)}_s\cdot
				\tau_s(f^{*(n)})}.
				\end{equation}
				Let $C_k(\tilde{f}^{(n)}):=\sum_{s\in \S_k} \tilde{f}^{(n)}_s\cdot
				\tau_s(\tilde{f}^{(n)})$ and $C_k(f^{*(n)}):=\sum_{s\in \S_k} f^{*(n)}_s\cdot
				\tau_s(f^{*(n)})$ denote the total \Wu{cost} of users from
				group $k\in \K$ w.r.t. NE profile $\tilde{f}^{(n)}$
				and SO profile $f^{*(n)},$ respectively, for all $k\in \K$
				and each $n\in \N.$
				By \eqref{def:redefPoA}, we obtain for each $n\in \N$ that
				\begin{equation}\label{def:redefPoA_1}
				\text{PoA}(d^{(n)})
				=\frac{\sum_{k\in \K}C_k(\tilde{f}^{(n)})}{\sum_{k\in \K}C_k(f^{*(n)})}.
				\end{equation}
				
				\Wu{We are now ready to start the asymptotic
					decomposition.}
				\Wu{We} aim to inductively partition the set
				$\K$ into mutually disjoint non-empty subsets 
				$\K_0,\K_1,\ldots,\K_t$ for some integer $t\ge 0,$ and
				prove \Wu{at} each step $m=0, \ldots,t$ that 
				\begin{equation}\label{eq:InductiveStepObj}
				\lim_{n\to \infty}\frac{\sum_{k\in \bigcup_{u=0}^{m} \K_u }C_k (\tilde{f}^{(n)})}{\sum_{k\in \bigcup_{u=0}^{m} \K_u }C_k (f^{*(n)})}
				=1.
				\end{equation}
				This procedure will not only form a proof for Theorem
				\ref{theo:SelfishMainTh}, but also \Wu{asymptotically decompose the underlying
				\Wuu{game.}}

				\textbf{Step $m=0:$ construct $\K_0,$ and
					prove \eqref{eq:InductiveStepObj} for $m=0.$}
					
					Before we formally start 
					step $m=0,$ we first introduce
					a trivial but useful fact about
					NE profiles..
					Note that for each $n\in \N$ and each 
					$k\in \K,$ by the user optimality
					\eqref{def:WE}, users from group $k$ have
					the same cost \Wu{$\tilde{L}_k^{(n)}$} w.r.t. the NE profile
					$\tilde{f}^{(n)}.$  Then, we obtain
					that
					\[
					C(\tilde{f}^{(n)})=\frac{1}{T(d^{(n)})}\sum_{k\in \N} \tilde{L}_k^{(n)}
					\cdot d_k^{(n)}.
					\]
					
					By \cite{Roughgarden2000How}, each SO profile $f^{*(n)}$ is actually an NE profile w.r.t.
					to \Wu{the auxiliary} price functions $c_a(x):=\big(x\cdot\tau_a(x)\big)'=x\cdot \tau_a'(x)+\tau_a(x),$
					for each $n\in \N.$ Obviously, each
					$c_a(\cdot)$ is again a polynomial of the same degree
					\Wu{as} $\tau_a(\cdot),$ for each $a\in A.$
					Let $\tilde{L}^{*(n)}_k$ be \Wu{the} cost of users
					from group $k\in\K$ w.r.t. price functions
					$c_a(\cdot)$ and the corresponding NE profile $f^{*(n)},$
					for each $n\in \N.$ Since $\tau_a(x)\in \Theta\big(c_a(x)\big)$ for each
					$a\in A,$ the cost of each user from group $k\in \K$
					is then in $\Theta(L_k^{*(n)})$ w.r.t. 
					SO profiles $f^{*(n)}$ and price functions
					$\tau_a(\cdot).$
					
					As a key component \Wu{at} each step of the asymptotic decomposition,
					we are now to estimate $\tilde{L}_k^{(n)},L_k^{*(n)}$ for
					each $k\in \K.$ \Wu{This will be
						\Wuu{the} base for the cost comparison.}
					\Wu{Claim \ref{theo:Tau_Magnitude} below 
					states that 
					the degree $\rho_k$ of a group $k\in \K$ reflects the magnitude
					of the cost of its users
					in both, NE and SO profiles.}
					\begin{claim}\label{theo:Tau_Magnitude}
						$\tilde{L}_k^{(n)},L_k^{*(n)}\in O\big(T(d^{(n)})^{\rho_k}\big),$ 
						for each group $k\in \K$.
						\end{claim}
						\begin{proof}[Proof of Claim \ref{theo:Tau_Magnitude}]
							Let $k\in \{1,\ldots,K\}$ be an arbitrarily
							fixed user group. 
							
							By the definition of 
							$\rho_k,$ there must exist some strategy \Wu{$s_0\in \S_k,$} such that $\rho_{s_0}=\rho_k=\min
							\{\rho_{s}:s\in \S_k\}.$ Let $f^{(n)}$ be an
							arbitrary feasible strategy profile w.r.t.
							user volume vector $d^{(n)},$ for each $n\in \N.$
							Then $\tau_{s_0}(f^{(n)})\in O\big(T(d^{(n)})^{\rho_k}\big),$
							since there are at most $T(d^{(n)})$ users adopting
							strategy $s_0$ and the degree of $s_0$ is $\rho_k.$
							Therefore, $\tau_{s_0}(f^{*(n)}),\tau_{s_0}(\tilde{f}^{(n)})\in O\big(T(d^{(n)})^{\rho_k}\big).$
							By the user optimality \eqref{def:WE} of NE profiles,
							we thus obtain that
							\[
							\tilde{L}_k^{(n)}\le \tau_{s_0}(\tilde{f}^{(n)})
							\in O\big(T(d^{(n)})^{\rho_k}\big).
							\]
							
							\Wu{Recall that} $c_{s_0}(f^{*(n)})\in \Theta\big(\tau_{s_0}(f^{*(n)})\big),$ and
							$f^{*(n)}$ is an SO profile w.r.t. price
							functions $c_a(\cdot),$ for each
							$n\in \N.$ Therefore, again by \eqref{def:WE}, \Wu{we obtain that}
							\[
							L_k^{*(n)}\le c_{s_0}(f^{*(n)})
							\in O\big(T(d^{(n)})^{\rho_k}\big).
							\]

							\end{proof}
							
							We can now formally start
							step $m=0.$ To facilitate our discussion, we assume,
							w.l.o.g., that
							\begin{itemize}
								\item \Wu{the} limit $\lim_{n\to \infty}\frac{d_k^{(n)}}{T_0(d^{(n)})}=:\d_k^{(0,\infty)}$ exists
								for some constant $\d_k^{(0,\infty)}\in [0,1],$ for all
								$k=1,\ldots,K,$ where $T_0(d^{(n)}):=T(d^{(n)}),$
								\item \Wu{the} limits 
								\[
								\lim_{n\to \infty}\frac{f^{*(n)}_s}{T_0(d^{(n)})}=:\f^{*(0,\infty)}_s
								\quad \text{and}\quad
								\lim_{n\to \infty}\frac{\tilde{f}^{(n)}_s}{T_0(d^{(n)})}=:\tilde{\f}^{(0,\infty)}_s
								\]
								exist for some constants 
								$\f^{*(0,\infty)}_s,\tilde{\f}^{(0,\infty)}_s\in [0,1],$
								for each strategy $s\in \S.$
								\end{itemize}
								Otherwise, we can take an infinite subsequence fulfilling
								these two conditions. For each $a\in A,$
								let
								\[
								\tilde{\f}_a^{(0,\infty)}:=\sum_{s\in \S}r(a,s)\cdot \tilde{\f}_s^{(0,\infty)}\quad \text{and}
								\quad \f_a^{*(0,\infty)}:=\sum_{s\in \S}r(a,s)\cdot \f_s^{*(0,\infty)}.
								\]
								
								We \Wu{now define} $\K_0.$
								Let $\alpha_0:=\max\{\rho_k:
								\d_k^{(0,\infty)}>0, k\in \K\}$
								\Wu{and} $ \K_0:=\{k\in \K: \d_k^{(0,\infty)}>0\text{ or } \rho_k\le 
								\alpha_0\}.$ Note that $\K_0\neq \emptyset,$ and
								there exist some $k\in \K_0$ such that
								\[
								\d_k^{(0,\infty)}>0\quad\text{and}\quad \rho_k=\alpha_0.
								\]
								Note also that both $(\tilde{\f}^{(0,\infty)}_s)_{s\in \S_k, k\in \K_0}$ and $(\f^{*(0,\infty)}_s)_{s\in \S_k, k\in \K_0}$
								are \Wu{feasible strategies for the} user volume vector
								$\d^{(0,\infty)}:=(\d_k^{(0,\infty)})_{k\in \K_0},$ i.e.,
								\[
								\d_k^{(0,\infty)}=\sum_{s\in \S_k} \tilde{\f}^{(0,\infty)}_s
								=\sum_{s\in \S_k}\f^{*(0,\infty)}_s\quad\forall k\in \K_0.
								\]
								Moreover,
								for each $k\in\K\backslash \K_0$ and $s\in \S_k,$ 
								\[
								\tilde{\f}_s^{(0.\infty)}=0,\quad\text{and}\quad 
								\f^{*(0,\infty)}_s=0,
								\]
								and thus for each $a\in A,$
								\[
								\tilde{\f}_a^{(0,\infty)}=\sum_{s\in \bigcup_{k\in \K_0}\S_k}r(a,s)\cdot \tilde{\f}_s^{(0,\infty)}\quad \text{and}
								\quad \f_a^{*(0,\infty)}=\sum_{s\in \bigcup_{k\in \K_0}\S_k}r(a,s)\cdot \f_s^{*(0,\infty)}.
								\]
								
								\Wu{By Claim \ref{theo:Tau_Magnitude},  
								$\tilde{L}_k^{(n)},L_k^{*(n)}\in O\big(T_0(d^{(n)})^{\alpha_0}\big)$ for each $k\in \K_0.$}
								
								\Wu{It remains at step 0} to prove \eqref{eq:InductiveStepObj}
								for $m=0.$
								Claim \ref{theo:Induction_Step_0} asserts this. It states that
								the average cost of users from groups $k\in \K_0$ 
								w.r.t.
								NE profiles $\tilde{f}^{(n)}$ \Wu{will be asymptotically equal to} that of those users w.r.t. 
								SO profiles $f^{*(n)}$, no matter which strategies the users from the other groups
								$k\in \K\backslash \K_0$ adopt.
								The proof of Claim \ref{theo:Induction_Step_0} 
								is inspired by \Wu{the proof} of the main result \Wu{Theorem 3.2} in  \cite{Wu2017Selfishness}.
								However, \Wuu{here,} we need 
								an additional \Wuu{argument} for SO profiles $f^{*(n)}$, since they
								now need not be SO profiles
								of the marginal game consisting of
								groups $k\in\K_0.$ \Wuu{The idea}
								to handle this is to
								consider profiles $f^{*(n)}$
								as NE profiles of the
								corresponding game with auxiliary
								price functions 
								$c_a(\cdot).$  
								Interestingly, 
								the two corresponding marginal games converge to limit games sharing
								NE profiles, since
								\begin{equation}\label{eq:AD_Curicial}
								\lim_{x\to\infty}
								\frac{c_a(x)}{\tau_a(x)}=1+\rho_a
								\end{equation}
								for each $a\in A.$
								A trivial fact hidden in the
								proof is that users from
								groups $k\in \K\backslash\K_0$
								\Wuu{do} not affect the limit behavior
								of users from $\K_0,$ since they
								\Wuu{only} account for a negligible
								limit proportion in the whole
								user volume. \Wuu{This makes it possible to independently consider} groups
								$k\in \K_0$ in the limit analysis.
								\begin{claim}\label{theo:Induction_Step_0}
									\[
									\lim_{n\to \infty}
									\frac{\sum_{k\in\K_0}
										C_k (\tilde{f}^{(n)})}{\sum_{k\in\K_0}C_{k}(f^{*(n)})}=1.
										\]
										Moreover,
										\[
										\sum_{k\in\K_0}
										C_k (\tilde{f}^{(n)}),\sum_{k\in\K_0}C_{k}(f^{*(n)})
										\in \Theta\big(T(d^{(n)})^{\alpha_0+1}\big).
										\]
										\end{claim}
										\begin{proof}[Proof of Claim \ref{theo:Induction_Step_0}]
											The main step of \Wuu{this} proof is to show that
											the marginal game consisting of \Wu{all} groups \Wu{in} $\K_0$
											will ``converge" to a  limit game
											with user volume vector $\d^{(0,\infty)}\!=\!(\d_k^{(0,\infty)})_{k\!\in\! \K_0},$ and $(\tilde{\f}^{(0,\infty)}_s)_{s\in \S_k, k\!\in\! \K_0}$ and $(\f^{*(0,\infty)}_s)_{s\in \S_k, k\!\in\! \K_0}$ are both NE profiles of the limit game.
											To show this, we employ a similar argument \Wu{as in the proof of Theorem 3.2} in
											\cite{Wu2017Selfishness}.
											
											Let $g_n:=T_0(d^{(n)})^{\alpha_0}$ be a scaling factor for each $n\in \N.$ 
											Note that for each resource $a\in A$ and each \Wuu{$x> 0,$} 
											\begin{equation}\label{eq:Scaled_Tau_Step_0}
											\begin{split}
											&\lim_{n\to \infty}\frac{\tau_a\big(T_0(d^{(n)})x\big)}{g_n}
											=:\tau_a^{(0,\infty)}(x)=
											\begin{cases}
											0,&\text{if } \rho_a<\alpha_0,\\
											b_\alpha \cdot x^{\alpha_0}, &\text{if }\rho_a=\alpha_0,\\
											\infty, &\text{otherwise,}
											\end{cases}
											\quad \text{and}\\
											&\lim_{n\to \infty}\frac{c_a\big(T_0(d^{(n)})x\big)}{g_n}
											=:c_a^{(0,\infty)}(x)=
											\begin{cases}
											0,&\text{if } \rho_a<\alpha_0,\\
											b_\alpha \cdot (\alpha_0+1)\cdot x^{\alpha_0}, &\text{if }\rho_a=\alpha_0,\\
											\infty, &\text{otherwise,}
											\end{cases}
											\end{split}
											\end{equation}
											where $b_{\alpha}> 0$ is the coefficient of 
											term $x^{\alpha_0}$ in the polynomial $\tau_a(\cdot)$
											if $\tau_a(\cdot)$ has degree $\alpha_0,$ for all
											$a\in A.$ 
											
											Let us define 
											\Wu{the} two ``limit" \Wu{games}
											\[
											\Gamma_{\tau}^{(0,\infty)}:=\big(\K_0, A, \bigcup_{k\in \K_0}\S_k, (r(a,s))_{a\in A, s\in \S_k,k\in \K_0}, (\tau_a^{(0,\infty)})_{a\in A},\d^{(0,\infty)}\big)
											\]
											and
											\[
											\Gamma_{c}^{(0,\infty)}:=\big(\K_0, A, \bigcup_{k\in \K_0}\S_k, (r(a,s))_{a\in A, s\in \S_k, k\in \K_0}, (c_a^{(0,\infty)})_{a\in A},\d^{(0,\infty)}\big).
											\]
											Obviously, $\Gamma_{\tau}^{(0,\infty)}$ and 
											$\Gamma_{c}^{(0,\infty)}$ have the same NE profiles
											and SO profiles, since $c_a(x)^{(0,\infty)}=(\alpha_0+1)
											\cdot \tau_a(x)^{(0,\infty)}$ for all 
											$a\in A.$ 
											
											We now aim to show that $(\tilde{\f}^{(0,\infty)}_s)_{s\in \S_k, k\in \K_0}$ is an NE profile
											of the game $\Gamma_{\tau}^{(0,\infty)}.$
											Let us arbitrarily fix \Wu{some} $k\in \K_0$ and
											two strategies $s,s'\in \S_k$ with 
											$\tilde{\f}_s^{(0,\infty)}>0.$
		\Wuu{This implies by Fact \ref{theo:StrategyPriceLimit_NonTight}
			that $s$ is tight.}
	By \Wuu{$L3)$ of Definition \ref{def:limitGame},} we obtain that $\tau_s^{(0,\infty)}(\tilde{\f}^{(0,\infty)})<\infty.$
											Thus, if $\tau_{s'}^{(0,\infty)}\big(\tilde{\f}^{(0,\infty)}\big)=\infty,$
											then 
											\begin{equation}\label{eq:Step0_Limit_NE}
											\tau_s^{(0,\infty)}\big(\tilde{\f}^{(0,\infty)}\big)\le \tau_{s'}^{(0,\infty)}\big(\tilde{\f}^{(0,\infty)}\big).
											\end{equation}
	We now assume that $\tau_{s'}^{(0,\infty)}(\tilde{\f}^{(0,\infty)})<\infty.$ \Wuu{Then, by $L3)$ of Definition \ref{def:limitGame},
		$s'$ is also tight.}
											We will prove that \eqref{eq:Step0_Limit_NE} still holds
											in this case,
											which, in turn, implies that $(\tilde{\f}^{(0,\infty)}_s)_{s\in \S_k, k\in \K_0}$ is an NE profile
											of the game $\Gamma_{\tau}^{(0,\infty)},$
											due to \Wu{the arbitrary choice} of $s,s'.$
											
											We recall that
											\[
											\lim_{n\to \infty}\frac{\tilde{f}^{(n)}_s}{T_0(d^{(n)})}=
											\tilde{\f}^{(0,\infty)}_s>0.
											\] 
											Thus, we obtain
											for large enough $n$  that
											\[
											\frac{\tilde{f}^{(n)}_s}{T_0(d^{(n)})}>0,
											\]
											which implies that 
											$\tilde{f}^{(n)}_s>0$
											for large enough $n.$ Since each
											$\tilde{f}^{(n)}$ is an NE profile for 
											each $n\in \N,$  we further obtain
											\Wu{by the user optimality 
												\eqref{def:WE}} that
											\[
											\tau_s(\tilde{f}^{(n)}) \le \tau_{s'}(\tilde{f}^{(n)})
											\] 
											for large enough $n$. Hence, by
											Fact \ref{theo:StrategyPriceConvergence}
											in the proof
											of Lemma
											\ref{theo:LimitGame_Wu2017}, we obtain that
											\[
											\tau_s^{(0,\infty)}(\tilde{\f}^{(0,\infty)})
											=\lim_{n\to \infty}
											\frac{\tau_s(\tilde{f}^{(n)})}{g_n}
											\le \lim_{n\to \infty}
											\frac{\tau_{s'}(\tilde{f}^{(n)})}{g_n}
											=\tau_{s'}^{(0,\infty)}(\tilde{\f}^{(0,\infty)}),
											\]
					\Wuu{since both $s$ and $s'$
						are tight.}
											
											\Wuu{Hence,} $(\tilde{\f}^{(0,\infty)}_s)_{s\in \S_k, k\in \K_0}$ is an NE profile
											of the game $\Gamma_{\tau}^{(0,\infty)}.$
											Similarly, we can prove that $(\f^{*(0,\infty)}_s)_{s\in \S_k, k\in \K_0}$
											is an NE profile of the game 
											$\Gamma_{c}^{(0,\infty)},$ which, in turn, implies
											that $(\f^{*(0,\infty)}_s)_{s\in \S_k, k\in \K_0}$ is
											also an NE profile of the game $\Gamma_{\tau}^{(0,\infty)}.$
											Therefore, 
											$(\tilde{\f}^{(0,\infty)}_s)_{s\in \S_k, k\in \K_0}$
											and $(\f^{*(0,\infty)}_s)_{s\in \S_k, k\in \K_0}$
											have equal cost w.r.t. the game $\Gamma_{\tau}^{(0,\infty)}.$
											Hence,  
											\begin{equation}\label{eq:Step0_limit_equiv}
											\frac{\sum_{k\in \K_0} \sum_{s\in \S_k:\rho_s\le \alpha_0}\tilde{\f}^{(0,\infty)}_s
												\cdot \sum_{a\in A}r(a,s)\cdot \tau_a^{(0,\infty)}(\tilde{\f}^{(0,\infty)}_a)}{\sum_{k\in \K_0} \sum_{s\in \S_k:\rho_s\le \alpha_0}\f^{*(0,\infty)}_s
												\cdot \sum_{a\in A}r(a,s) \cdot \tau_a^{(0,\infty)}(\f^{*(0,\infty)}_a)}=1,
												\end{equation}
												where we observe the fact that for $k\in \K_0,$
												each strategy $q\in \S_k$ with degree $\rho_q>\alpha_0$
												will be non-tight, since there exists a strategy
												$p\in \S_k$ with $\rho_p=\rho_k\le \alpha_0.$
												Moreover,
												both the \Wuu{numerator} and \Wuu{denominator} in 
												\eqref{eq:Step0_limit_equiv} are  positive and
												finite, i.e., in $\Theta(1),$
												since there exists a $k\in \K_0$ such that
												$\d_k^{(0,\infty)}>0$ and $\rho_k=\alpha_0.$
												
												By Fact
												\ref{theo:StrategyPriceConvergence} and
												Fact \ref{theo:StrategyPriceLimit_NonTight}
												in the proof of Lemma \ref{theo:LimitGame_Wu2017}, we obtain immediately that
												\begin{equation}\label{eq:Crucial1_in_Step0}
												\begin{split}
												\lim_{n\to \infty} \frac{\sum_{k\in \K_0}C_k(\tilde{f}^{(n)})}{T(d^{(n)})\cdot g_n}&=\sum_{k\in \K_0} \sum_{s\in \S_k:\rho_s\le \alpha_0}\tilde{\f}^{(0,\infty)}_s
												\cdot \sum_{a\in A}r(a,s)\cdot \tau_a^{(0,\infty)}(\tilde{\f}^{(0,\infty)}_a)\\
												&\in \Theta(1),
												\end{split}
												\end{equation}
												and 
												\begin{equation}\label{eq:Crucial2_in_Step0}
												\begin{split}
												\lim_{n\to \infty} \frac{\sum_{k\in \K_0}C_k(f^{*(n)})}{T(d^{(n)})\cdot g_n}&=\sum_{k\in \K_0} \sum_{s\in \S_k:\rho_s\le \alpha_0}\f^{*(0,\infty)}_s
												\cdot \sum_{a\in A}r(a,s)\cdot \tau_a^{(0,\infty)}(\f^{*(0,\infty)}_a)\\
												&\in \Theta(1).
												\end{split}
												\end{equation}
												\Wuu{Here} \eqref{eq:Crucial2_in_Step0} is \Wuu{obtained from}  
												the fact that
												profiles 
												$f^{*(n)}$ are
												NE profiles
												of the corresponding game
												with price functions $c_a(\cdot),$
												and that
												\begin{equation}
												\label{eq:ADCrucial2}
												\lim_{n\to\infty}
												\frac{c_a\big(T(d^{(n)})z\big)}{g_n}=c_a^{(0,\infty)}(z)=(1+\rho_a)\tau_a^{(0,\infty)}(z)=(1+\rho_a)\cdot\lim_{n\to\infty}
												\frac{\tau_a\big(T(d^{(n)})z\big)}{g_n}.
												\end{equation}

Claim \ref{theo:Induction_Step_0} \Wu{then} follows
\Wu{from \eqref{eq:Step0_limit_equiv},															\eqref{eq:Crucial1_in_Step0} and \eqref{eq:Crucial2_in_Step0}.}
																				\end{proof}

In the proof of Claim \ref{theo:Induction_Step_0}, \eqref{eq:AD_Curicial} and the scaling
factors $g_n$ play pivotal roles \Wuu{as \Wu{the equation} \eqref{eq:ADCrucial2} does not hold
	without them.}

																				\Wu{If} $\K_0=\K,$ then we have already finished the decomposition and completed the proof of Theorem~\ref{theo:SelfishMainTh} \Wu{by Claim \ref{theo:Induction_Step_0}}. We assume now that $\K\backslash\K_0\neq\emptyset.$
																						In this case, we need to further partition
																						$\K\backslash\K_0$ and proceed to
																						step $m=1.$
																						
																						\textbf{Step $m=1$: contruct $\K_1,$ and prove
																							\eqref{eq:InductiveStepObj} for $m=1.$}
																							
																							\Wu{To facilitate our discussion}, we \Wu{will first define} some 
																						notations and propose further assumptions
																							on the \Wu{fixed} user volume sequence 
																							$\{d^{(n)}\}_{n\in \N}.$										 
																							
																							For each $n\in \N$ and
																							each $a\in A,$ we denote by
																							\[
																							\tau_a^{(1,n)}(x):=\tau_a\big(x+\tilde{f}_a^{(n)}(\K_0)\big)\quad 
																							\text{and}\quad c_a^{(1,n)}(x):=c_a\big(x+f_a^{*(n)}(\K_0)\big)
																							\]
																							the price function of resource $a$
																							\Wu{under the condition} that the users from 
																							groups $k\in \K_0$ stick to strategies
																							they used in NE profiles $\tilde{f}^{(n)}$ and SO profiles
																							$f^{*(n)},$ respectively.
																							\Wu{Here,} 
																							\[
																							\tilde{f}_a^{(n)}(\K_0):=\sum_{k\in \K_0}\sum_{s\in \S_k}\sum_{a\in A}r(a,s)\tilde{f}_s^{(n)}\in O\big(T_0(d^{(n)})\big)
																							\]
																							and 
																							\[
																							f^{*(n)}_a(\K_0):=\sum_{k\in \K_0}\sum_{s\in \S_k}\sum_{a\in A}r(a,s)f_s^{*(n)}
																							\in O\big(T_0(d^{(n)})\big)
																							\]
																							denote the volumes of resource
																							$a$ consumed by users from 
																							groups $k\in\K_0$ w.r.t.
																							NE profile $\tilde{f}^{(n)}$ 
																							and SO profile $f^{*(n)},$
																							repectively,
																							for each $a\in A$
																							and each $n\in \N.$  
																							
																							Let $f^{*(n)}(\K\backslash\K_0):=(f_s^{*(n)})_{k\in \K\backslash\K_0:s\in \S_k}$
																							and $\tilde{f}^{(n)}(\K\backslash\K_0):=(\tilde{f}_s^{(n)})_{k\in \K\backslash\K_0:s\in \S_k}$
																							be the ``marginal" profiles consisting of users from \Wu{the}
																							remaining groups $\K\backslash\K_0$
																							w.r.t. the SO profile$f^{*(n)}$
																							and NE profile $\tilde{f}^{(n)},$
																							respectively, for each $n\in \N.$
																							Obviously,
																							\[
																							\tau_a^{(1,n)}\big(\tilde{f}^{(n)}(\K\backslash\K_0)\big)
																							=\tau_a\big(\tilde{f}^{(n)}\big)
																							\quad\text{and}\quad c_a^{(1,n)}\big(f^{*(n)}(\K\backslash\K_0)\big)
																							=c_a\big(f^{*(n)}\big)
																							\]
																							for each $a\in A$ and each $n\in \N.$
																							
																							Let
																							\[
																							\Gamma_{\tau}^{(1,n)}\!:=\!\big(\K\backslash\K_0,\! A,\! \bigcup_{k\in \K\backslash\K_0}\!\S_k,\! (r(a,s))_{a\in A, s\in \S_k, k\in \K\backslash\K_0},\! (\tau_a^{(1,n)})_{a\in A},\!(d_k^{(n)})_{k\in \K\backslash\K_0}\big)
																							\]
																							and
																							\[
																							\Gamma_{c}^{(1,n)}\!:=\!\big(\K\backslash\K_0,\! A,\! \bigcup_{k\in \K\backslash\K_0}\!\S_k,\! (r(a,s))_{a\in A, s\in \S_k, k\in \K\backslash\K_0},\! (c_a^{(1,n)})_{a\in A},\!(d_k^{(n)})_{k\in \K\backslash\K_0}\big)
																							\]
																							be the corresponding marginal games
																							\Wu{under the condition} that users from 
																							groups in $\K_0$ stick to the strategies
																							they used in NE profile $\tilde{f}^{(n)}$ and
																							SO profile $f^{*(n)},$
																							respectively, for each $n\in \N.$ 
																							Obviously, profiles $\tilde{f}^{(n)}(
																							\K\backslash\K_0)$ and $f^{*(n)}(
																							\K\backslash\K_0)$
																							are NE profiles of the two
																							\Wuu{marginal} games $\Gamma_{\tau}^{(1,n)}$
																							and $\Gamma_c^{(1,n)}$, respectively, w.r.t. the user volume vector
																							$(d^{(n)}_k)_{k\in \K\backslash\K_0},$
																							for each $n\in \N.$
																							
																							\Wu{At} step $m=1,$ we shall employ \Wuu{an  argument
																							similar to that for} step $m=0.$ However, we shall now
																							consider the marginal profiles and the marginal games consisting of groups from $\K\backslash\K_0$.
																							
																							Let $T_1(d^{(n)})=\sum_{k\in \K\backslash\K_0}d_k^{(n)}$ be \Wu{the} total volume
																							of users from groups in $\K\backslash\K_0,$ for each $n\in \N.$ Then,
																							$T_1(d^{(n)})\in o\big(T_0(d^{(n)})\big).$
																							\Wu{Then, it follows trivially
																								for the} price functions $\tau_a^{(1,n)}(\cdot)$ and
																							$c_a^{(1,n)}(\cdot)$ \Wu{that,} for
																							each $x\ge 0,$
																							\begin{equation}\label{eq:Crucial1_ind_ste_0.5}
																							\tau_a^{(1,n)}(T_1(d^{(n)}) x)\in \Theta\Big(\max\big\{\tau_a\big(T_1(d^{(n)}) x\big),
																							\tau_a\big(\tilde{f}_a^{(n)}(\K_0)\big)\big\}\Big)
																							\end{equation}
																							and
																							\begin{equation}\label{eq:Crucial2_ind_ste_0.5}
																							c_a^{(1,n)}(T_1(d^{(n)}) x)\in \Theta\Big(\max\big\{c_a\big(T_1(d^{(n)}) x\big),
																							c_a\big(f_a^{*(n)}(\K_0)\big)\big\}\Big),
																							\end{equation}
																							for each $a\in A,$ since both 
																							$\tau_a(\cdot)$ and $c_a(\cdot)$ are \Wu{asymptotically non-decreasing} polynomials 
																							for all $a\in A.$
																							

																							Similarly, we assume, w.l.o.g., that
																							\begin{itemize}
																								\item \Wu{the} limit $\lim_{n\to \infty}\frac{d_k^{(n)}}{T_1(d^{(n)})}=:\d_k^{(1,\infty)}$
																								exists for some constant $\d_k^{(1,\infty)}\in [0,1],$ for each $k\in \K\backslash\K_0,$ 
																								\item \Wu{the} limits 
																								\[
																								\lim_{n\to \infty}\frac{\tilde{f}^{(n)}_s}{T_1(d^{(n)})}=:\tilde{\f}_s^{(1,\infty)}
																								\quad \text{and}\quad 
																								\lim_{n\to \infty} \frac{f^{*(n)}_s}{T_1(d^{(n)})}=:\f_s^{*(1,\infty)}
																								\]
																								exist for some constants $\tilde{\f}_s^{(1,\infty)},\f_s^{*(1,\infty)}\in [0,1]$
																								for each $k\in \K\backslash\K_0$ and 
																								each $s\in \S_k.$
																								\end{itemize}
																								Otherwise, we can again take an infinite subsequence $\{n_i\}_{i\in \N}$ fulfilling these two conditions. Again,
																								let 
																								\[
																								\tilde{\f}_a^{(1,\infty)}:=
																								\sum_{k\in \K\backslash\K_0}
																								\sum_{s\in \S_k} r(a,s)\cdot \tilde{\f}^{(1,\infty)}_s
																								=\sum_{k\in \K_1}
																								\sum_{s\in \S_k} r(a,s)\cdot \tilde{\f}^{(1,\infty)}_s
																								\]
																								and 
																								\[
																								\f_a^{*(1,\infty)}:=
																								\sum_{k\in \K\backslash\K_0}
																								\sum_{s\in \S_k} r(a,s)\cdot \f^{*(1,\infty)}_s=
																								\sum_{k\in \K_1}
																								\sum_{s\in \S_k} r(a,s)\cdot \f^{*(1,\infty)}_s
																								\]
																								for each resource $a\in A,$
																								where $\K_1$ is defined below.
																								
																								Now we are ready to partition $\K\backslash\K_0.$
																								We define $\alpha_1:=\max\{\rho_k: k\in \K\backslash\K_0, \d_k^{(1,\infty)}>0\}>\alpha_0,$
																								and $\K_1:=\{k\in \K\backslash \K_0: \rho_k\le \alpha_1\}.$
																								Obviously, 
																								\[
																								\sum_{k\in \K_1}\frac{d_k^{(n)}}{T_1(d^{(n)})}\to \sum_{k\in \K_1}\d_k^{(1,\infty)}=1 \quad \text{as}\quad 
																								n\to \infty,
																								\]
																								and there exists $k\in \K_1$ such that
																								$\rho_k=\alpha_1$ and $d_k^{(n)}\in \Theta\big(T_1(d^{(n)})\big).$ 
																								
																								
																								To show \eqref{eq:InductiveStepObj}, we need a
																								\Wu{tighter} bound of $\tilde{L}_k^{(n)}$ and
																								$L_k^{*(n)}$ for each $k\in \K_1.$
																								Note that for 
																								each $k\in \K_1,$ $\tilde{L}_k^{(n)}$ is still the cost
																								of users from group $k$ w.r.t. \Wu{the} (marginal)
																								NE profile $\tilde{f}^{(n)}(\K\backslash\K_0)$
																								and \Wu{the} (marginal) game $\Gamma_\tau^{(1,n)}.$ 
																								\Wu{Similarly,} $L_k^{*(n)}$ is still the cost
																								of users from group $k$ w.r.t. \Wu{the} (marginal)
																								NE profile $f^{*(n)}(\K\backslash\K_0)$
																								and \Wu{the} (marginal) game $\Gamma_c^{(1,n)},$ 
																								for each $n\in \N.$
																								
																								Similar \Wu{to the} Claim \ref{theo:SelfishMainTh} at step $m=0$, Claim
																								\ref{theo:Crucial_Induc_Step_1} estimates
																								the \Wu{cost} of users w.r.t. the two marginal games
																								and corresponding NE profiles.
																								\begin{claim}\label{theo:Crucial_Induc_Step_1}
																									For each group $k\in \K_1,$ 
																									\[
																									\tilde{L}_k^{(n)}, L_k^{*(n)}\in 	O\Big(\max\{g_n^{(0)},g_n^{(1)}\}\Big),	\]
where $g_n^{(0)}=T_0(d^{(n)})^{\alpha_0}$ is the scaling factor at step $m=0,$ and $g_n^{(1)}=T_1(d^{(n)})^{\alpha_1}$ will be the scaling factor used at step $m=1,$ for each $n\in \N.$											
\end{claim}
																									\begin{proof}[Proof of Claim \ref{theo:Crucial_Induc_Step_1}]
																										We first prove that for each resource $a\in A,$ 
																										\begin{equation}\label{eq:Step1_IMPORT}
																										\tau_a\big(f_a^{*(n)}(\K_0)\big),\tau_a(\tilde{f}_a^{(n)}\big(\K_0)\big)\in O\big(g_n^{(0)}\big). 
																										\end{equation}
																										We only prove \Wu{
																											this for} NE profiles $\tilde{f}^{(n)}.$
																										An almost identical argument applies to the SO profiles $f^{*(n)}$.
																										
																										Consider an arbitrarily fixed $a\in A.$
																										If $\tilde{f}_a^{(n)}(\K_0)=\sum_{s\in \S_{k'},k'\in \K_0} r(a,s)\cdot \tilde{f}_s^{(n)}>0,$ then $\tilde{f}_s^{(n)}>0$ for 
some $s\in \S_{k'}$ with $r(a,s)>0$
\Wu{and} for some $k'\in \K_0.$  By Claim \ref{theo:Tau_Magnitude}, if $\tilde{f}_s^{(n)}>0,$ i.e., $s$ is used by
																										some users from group $k'\in \K_0$, then 
																										\[
																										\tau_s(\tilde{f}^{(n)})=\tilde{L}_{k'}^{(n)}\in O\big(T(d^{(n)})^{\alpha_0}\big)=O(g_n^{(0)}),
																										\]
																										which in turn implies that
																										\[
																										\tau_a\big(\tilde{f}_a^{(n)}(\K_0)\big)\in  O\big(g_n^{(0)}\big)
																										\]
																										since 
																										\[
																										\tau_a\big(\tilde{f}_a^{(n)}(\K_0)\big)\le \tau_s(\tilde{f}^{(n)}).
																										\]
																										If $\tilde{f}_a^{(n)}(\K_0)=0,$ then 
																										\[
																										\tau_a\big(\tilde{f}_a^{(n)}(\K_0)\big)=\eta_a\in \Theta(1),
																										\]
																										for some constant \Wuu{$\eta_a\ge  0.$}
																										Thus we obtain that
																										\begin{equation*}
																										\tau_a\big(\tilde{f}_a^{(n)}(\K_0)\big)\in  O\big(g_n^{(n)}\big).
																										\end{equation*}
																										\Wu{A} similar result holds for SO profiles $f^{*(n)}.$

																										We are now ready to finish the proof of Claim \ref{theo:Crucial_Induc_Step_1}.
																										We first \Wu{prove
																											this for}
																										NE profiles $\tilde{f}^{(n)}(\K\backslash\K_0).$
																										By \eqref{eq:Step1_IMPORT},   \eqref{eq:Crucial1_ind_ste_0.5} and \eqref{eq:Crucial2_ind_ste_0.5}, we
																										obtain for each $a\in A$ that
																										\[
																										\begin{split}
																										\tau_a^{(1,n)}&\big(\tilde{f}^{(n)}_a(\K\backslash\K_0)\big)
																										=\tau_a^{(1,n)}\Big(T_1(d^{(n)})\frac{\tilde{f}^{(n)}_a(\K\backslash\K_0)}{T_1(d^{(n)})}\Big)\\
																										&\in O\big(\max\{g_n^{(0)}, T_1(d^{(n)})^{\rho_a}\}\big).
																										\end{split}
																										\]
																										where we observe that
																										$\frac{\tilde{f}^{(n)}_a(\K\backslash\K_0)}{T_1(d^{(n)})}\in O(1),$
																										since
																										\[
																										\lim_{n\to \infty}\frac{\tilde{f}^{(n)}_a(\K\backslash\K_0)}{T_1(d^{(n)})}=\tilde{\f}_a^{(1,\infty)}\in O(1).
																										\]
																										
																										Let $k\in \K_1$ be an arbitrarily fixed group. For each strategy 
																										$s\in \S_k,$  we obtain
																										\Wu{by the above discussion} that
																										the price of $s$ w.r.t. \Wu{the} NE profiles $\tilde{f}^{(n)}(\K\backslash\K_0)$ and
																										\Wu{the} marginal games $\Gamma_\tau^{(1,n)}$
																										are in
																										\[
																										O\Big(\max\{g_n^{(0)},T_1(d^{(n)})^{\rho_s}\}\Big).
																										\]
																										By the user optimality \eqref{def:WE}
																										of NE profiles, we then obtain that
																										the cost $\tilde{L}_k^{(n)}$ of users in each group 
																										$k\in \K_1$ \Wuu{is} in
																										\[
																										O\Big(\max\{g_n^{(0)},T_1(d^{(n)})^{\rho_k}\}\Big)\subseteq O\Big(\max\{g_n^{(0)}, g_n^{(1)}\}\Big),
																										\]
																										since $\rho_s\ge \rho_k$
																										and $\rho_k\le \alpha_1$
																										for each $s\in \S_k$ and $k\in \K_1.$

																										An almost identical argument
																										carries over to  
																										the NE profiles $f^{*(n)}(\K\backslash\K_0)$
																										w.r.t. the marginal games $\Gamma_c^{(1,n)}.$
																										This completes the proof.
																										\end{proof}
																										
																										\Wu{With} Claim \ref{theo:Crucial_Induc_Step_1},
																										we can now prove \eqref{eq:InductiveStepObj}
																										for  $m=1.$ \Wu{To this end}, we make \Wuu{the} further assumption that
	\[
	\lim_{n\to\infty} \frac{g_n^{(1)}}{g_n^{(0)}}
	=\lim_{n\to \infty} \frac{T_1(d^{(n)})^{\alpha_1}}{T_0(d^{(n)})^{\alpha_0}}
	=\beta_0,
	\]
	for some constant $\beta_0\in [0,\infty].$
	\Wu{Otherwise, one can take a subsequence fulfilling this condition.}
																										\begin{claim}\label{theo:Curial_Ind_Step1_2}
																											\begin{equation*}
																											\lim_{n\to \infty}\frac{\sum_{k\in \bigcup_{u=0}^{1} \K_u }C_k (\tilde{f}^{(n)})}{\sum_{k\in \bigcup_{u=0}^{1} \K_u }C_k (f^{*(n)})}
																											=1.
																											\end{equation*}
																											Moreover,
																											\[		
	\begin{split}			
		\sum_{k\in \bigcup_{u=0}^{1} \K_u }C_k (\tilde{f}^{(n)}),
																										&	\sum_{k\in \bigcup_{u=0}^{1} \K_u }C_k (f^{*(n)})
																											\in \Theta\big(\max\{T_0(d^{(n)})^{\alpha_0+1},T_1(d^{(n)})^{\alpha_1+1}\}\big)\\
			&=\Theta\Big(\max\{T_0(d^{(n)})\cdot g_n^{(0)}, T_1(d^{(n)})\cdot g_n^{(1)}\}\Big).
										\end{split}																	\]
																											\end{claim}
																											\begin{proof}[Proof of Claim \ref{theo:Curial_Ind_Step1_2}]
																												We separate the discussion into two cases.
																												
																												\textbf{(Case 1: $\beta_0<\infty$)}
																												In this case, the total cost of \Wuu{the} groups from $\K_1$
																												are negligible w.r.t. the total cost of \Wu{the} groups from
																												$\K_0.$ By Claim~\ref{theo:Induction_Step_0},
																												see
																												\eqref{eq:Crucial1_in_Step0} and
																												\eqref{eq:Crucial2_in_Step0}, we obtain 
																												for large enough $n$  that
																												\[
																												\sum_{k\in \K_0}C_k(\tilde{f}^{(n)})\approx \sum_{k\in \K_0}C_k(f^{*(n)})\in \Theta\big(T_0(d^{(n)})\cdot g_n^{(0)}\big).
																												\]
																												\Wu{With} Claim \ref{theo:Crucial_Induc_Step_1},  we obtain that
																												\[
																												\sum_{k\in \K_1}C_k(\tilde{f}^{(n)})=\sum_{k\in \K_1}d_k^{(n)} \cdot \tilde{L}_k^{(n)}\in O\big(T_1(d^{(n)})\cdot g_n^{(0)}\big),
																												\]
\Wuu{since $g_n^{(1)}\in O(g_n^{(0)})$ in this case.}
																												
																												Note that 
																												\[
																												\tau_a(f^{*(n)}_a)\in \Theta\big(c_a(f^{*(n)}_a)\big)
																												\quad \text{and}\quad c_a(f^{*(n)}_a)=c_a^{(1,n)}\big(f^{*(n)}_a(\K\backslash\K_0)\big)
																												\]
																												for each $a\in A.$ Thus, again by Claim \ref{theo:Crucial_Induc_Step_1},
																												we obtain
																												\[
																												\sum_{k\in \K_1}C_k(f^{*(n)})\in\Theta\Big(\sum_{k\in \K_1}d_k^{(n)}\cdot L^{*(n)}_k\Big)\subseteq O\big(T_1(d^{(n)})\cdot g_n^{(0)}\big).
																												\]
																												
																												Since  $T_1(d^{(n)})\in o\big(T_0(d^{(n)})\big),$ we obtain that 
																												\[
																												\sum_{k\in\K_1}C_k(\tilde{f}^{(n)})
																												\in o\Big(\sum_{k\in \K_0}C_k(\tilde{f}^{(n)})\Big)
																												\quad\text{and}\quad\sum_{k\in \K_1}C_k(f^{*(n)})\in o\Big(\sum_{k\in \K_0}C_k(f^{*(n)})\Big).
																												\]
																												\Wu{This gives}
																												\[
																												\lim_{n\to \infty}\frac{\sum_{k\in \bigcup_{u=0}^{1}\K_u}C_k(\tilde{f}^{(n)})}{
																													\sum_{k\in \bigcup_{u=0}^{1}\K_u}C_k(f^{*(n)})
																													}
																													=\lim_{n\to \infty}\frac{\sum_{k\in \K_0}C_k(\tilde{f}^{(n)})}{
																														\sum_{k\in \K_0}C_k(f^{*(n)})
																														}=1.
																														\]
																														
																														\Wu{With} Claim \ref{theo:Induction_Step_0}, we obtain that
																														\[
																														\sum_{k\in \bigcup_{u=0}^{1} \K_u }C_k (\tilde{f}^{(n)}),
																														\sum_{k\in \bigcup_{u=0}^{1} \K_u }C_k (f^{*(n)})
																														\in \Theta\big(\max\{T_0(d^{(n)})\cdot g_n^{(0)},T_1(d^{(n)})\cdot g_n^{(1)}\}\big),\]
																														since 
\Wuu{$g_n^{(1)}\in O(g_n^{(0)})$ and $T_1(d^{(n)})\in o(T_0(d^{(n)})).$}
																														
																														\textbf{(Case 2: $\beta_0=\infty$)}
																														In this case, we can use a similar argument as in the proof of Claim 
																														\ref{theo:Induction_Step_0} to show that 
																														\[
																														\lim_{n\to \infty}\frac{\sum_{k\in \K_1}C_k(\tilde{f}^{(n)})}{\sum_{k\in \K_1}
																															C_k(f^{*(n)})}=1.
																															\]
																															To this end, we need to look more closely into
																															the price functions $\tau_a^{(1,n)}(\cdot)$ and
																															$c_a^{(1,n)}(\cdot)$. 
																															
																															Let $a\in A$ be an arbitrarily fixed resource.
																															We aim to show that the price of $a$ is asymptotically
																															determined only by users from groups $k\in \K_1$ \Wu{under}
																															the assumption that $\beta_0=\infty$ and \Wu{the} scaling
																															factor \Wu{is} $g_n^{(1)}=T_1(d^{(n)})^{\alpha_1}$ for each $n\in \N.$
																															
																															If $\rho_a<\alpha_1,$ then we obtain by \eqref{eq:Crucial1_ind_ste_0.5} and \eqref{eq:Crucial2_ind_ste_0.5} that  the limit
																															\begin{equation}\label{eq:Step1_LimitTau1}
																															\lim_{n\to\infty}\frac{\tau_a^{(1,n)}(T_1(d^{(n)})x)}{g_n^{(1)}}=0,
																															\end{equation}
									\Wu{for	each $x> 0,$	}														since $g_n^{(0)}\in o(g_n^{(1)})$ and
																															\[
																															\tau_a^{(1,n)}(T_1(d^{(n)})x)\!\in\! O\Big(
																															\max\big\{g_n^{(0)},T_1(d^{(n)})^{\rho_a}\big\}\Big) \!\subseteq \! o\big(g_n^{(1)}\big).
																															\]

																															If $\rho_a\ge \alpha_1,$
																															then we obtain by the definition
																															of $\tau_a^{(1,n)}(\cdot)$ and
																															\eqref{eq:Step1_IMPORT} that 
																															the limit
																															\begin{equation}\label{eq:Step1_LimitTau2}
																															\lim_{n\to\infty}\frac{\tau_a^{(1,n)}(T_1(d^{(n)})x)}{g_n^{(1)}}=
																															\lim_{n\to\infty}\frac{\tau_a(T_1(d^{(n)})x)}{T_1(d^{(n)})^{\alpha_1}},
																															\end{equation}
																															exists \Wu{for all $x\ge 0.$} \Wu{Here} we observe that
																															\[
																															\tau_a\big(\tilde{f}_a^{(n)}(\K_0)\big)\in O(g_n^{(0)})
																															\subseteq o(g_n^{(1)}).\]  
																															Similarly, we obtain  that for each $x>0$ and each $a\in A$ 		with $\rho_a\ge \alpha_1,$ the limit
\[
\lim_{n\to\infty} \frac{c_a^{(1,n)}(T_1(d^{(n)})x)}{g_n^{(1)}}
=
\lim_{n\to\infty}\frac{c_a(T_1(d^{(n)})x)}{g_n^{(1)}}
\]
exists.																														
																							Then,
																															 \eqref{eq:Step1_LimitTau1}
																															and \eqref{eq:Step1_LimitTau2} \Wu{yield} for each \Wuu{$x> 0$} and $a\in A$ that
																															\begin{equation}\label{eq:Limit_CPF1}
																															\begin{split}
																															&\lim_{n\to\infty}\frac{\tau_a^{(1,n)}(T_1(d^{(n)})x)}{g_n^{(1)}}=:
																															\tau_a^{(1,\infty)}(x)=
																															\begin{cases}
																															0,&\text{if } \rho_a<\alpha_1,\\
																															b_{a}\cdot (\alpha_1+1) x^{\alpha_1},&\text{if } \rho_a=\alpha_1,\\
																															\infty,&\text{if }\rho_a>\alpha_1,
																															\end{cases}\\
																															&\lim_{n\to\infty}\frac{c_a^{(1,n)}(T_1(d^{(n)})x)}{g_n^{(1)}}=:
																															c_a^{(1,\infty)}(x)=
																															\begin{cases}
																															0,&\text{if } \rho_a<\alpha_1,\\
																															b_{a}\cdot (\alpha_1 +1) x^{\alpha_1},&\text{if } \rho_a=\alpha_1,\\
																															\infty,&\text{if }\rho_a>\alpha_1,
																															\end{cases}
																															\end{split}
																															\end{equation}
																															where, again, $b_a>0$ is the coefficient of
																															the term $x^{\alpha_1}$ of polynomial
																															$\tau_a(\cdot),$ if $\tau_a(\cdot)$ does have
																															degree $\alpha_1,$ for each
																															$a\in A.$ Notice that \eqref{eq:Limit_CPF1}
																															actually indicates that users from groups
																															$k\in \K\backslash\bigcup_{u=0}^{1}\K_u$ can be ignored
																															when we discuss the limit behavior of users
																															from groups $k\in \K_1,$ since their total
																															volume is negligible compared to that of users
																															from groups $k\in \K_1.$
																															
																															By \eqref{eq:Limit_CPF1}, with an
																															\Wuu{ argument 
																																similar to} that \Wuu{for} Claim \ref{theo:Induction_Step_0},
																															\Wu{we} can show that 
																															\[
																															\f^{*(1,\infty)}:=\big(\f_s^{*(1,\infty)}\big)_{k\in \K_1:s\in \S_k}
																															\quad \text{and} \quad \tilde{\f}^{(1,\infty)}:=\big(\tilde{\f}_s^{(1,\infty)}\big)_{k\in \K_1:s\in \S_k}
																															\]
																															are NE profiles w.r.t.
																															\Wu{the limit} marginal games 
																															\[
																															\Gamma_{\tau}^{(1,\infty)}\!:=\!\big(\K_1,\! A,\! \bigcup_{k\in \K_1}\!\S_k,\! (r(a,s))_{a\in A, s\in \S_k, k\in \K_1},\! (\tau_a^{(1,\infty)})_{a\in A},\!(\d_k^{(1,\infty)})_{k\in \K_1}\big)
																															\]
																															and
																															\[
																															\Gamma_{c}^{(1,\infty)}\!:=\!\big(\K_1,\! A,\! \bigcup_{k\in \K_1}\!\S_k,\! (r(a,s))_{a\in A, s\in \S_k, k\in \K_1},\! (c_a^{(1,\infty)})_{a\in A},\!(\d_k^{(1,\infty)})_{k\in \K_1}\big),
																															\]
																															respectively. 
																															
																															Moreover, we can show that 
																															\[
																															\lim_{n\to\infty} \frac{\sum_{k\in \K_1}C_k(\tilde{f}^{(n)})}{T_1(d)\cdot g_n^{(1)}}=
																															\sum_{k\in \K_1}\sum_{s\in \S_k: \rho_s\le \alpha_1}\tilde{\f}_s^{(1,\infty)}
																															\cdot \sum_{a\in A}r(a,s)\tau_a^{(1,\infty)}(\tilde{\f}_a^{(1,\infty)})\in \Theta(1)
																															\]
																															and 
																															\[
																															\lim_{n\to\infty} \frac{\sum_{k\in \K_1}C_k(f^{*(n)})}{T_1(d)\cdot g_n^{(1)}}=
																															\sum_{k\in \K_1}\sum_{s\in \S_k: \rho_s\le \alpha_1}{\f}_s^{*(1,\infty)}
																															\cdot \sum_{a\in A}r(a,s)\tau_a^{(1,\infty)}({\f}_a^{*(1,\infty)})\in \Theta(1).
																															\]
																															Note that there exists at least one
																															group $k\in \K_1,$ such 
																															that $\d^{(1,\infty)}_k>0$ and
																															$\rho_k=\alpha_1.$ Therefore, 
																															\[
																															\frac{
																																\sum_{k\in \K_1}\sum_{s\in \S_k: \rho_s\le \alpha_1}\tilde{\f}_s^{(1,\infty)}
																																\cdot \sum_{a\in A}r(a,s)\tau_a^{(1,\infty)}(\tilde{\f}_a^{(1,\infty)})
																																}{
																																\sum_{k\in \K_1}\sum_{s\in \S_k: \rho_s\le \alpha_1}{\f}_s^{*(1,\infty)}
																																\cdot \sum_{a\in A}r(a,s)\tau_a^{(1,\infty)}({\f}_a^{*(1,\infty)})
																																}=1,
																																\]
																																since the two games
																																$\Gamma_\tau^{(1,\infty)}$
																																and $\Gamma_c^{(1,\infty)}$
																																have \Wu{the} same NE profiles, the \Wu{numerator}
																																is the cost of the NE profile $\tilde{\f}^{(1,\infty)},$
																																and the \Wu{denominator} is the cost of the
																																cost of the NE profile
																																$\f^{*(1,\infty)}.$
																																
																																As a result, we obtain that
																																\[
																																\lim_{n\to \infty}\frac{\sum_{k\in \K_1}C_k(\tilde{f}^{(n)})}{\sum_{k\in \K_1}
																																	C_k(f^{*(n)})}=1,
																																	\]
																																	which, in turn, \Wu{implies by Claim \ref{theo:Induction_Step_0}} that
																																	\[
																																	\lim_{n\to \infty}\frac{\sum_{k\in \bigcup_{u=0}^{1} \K_u }C_k (\tilde{f}^{(n)})}{\sum_{k\in \bigcup_{u=0}^{1} \K_u }C_k (f^{*(n)})}
																																	=1.
																																	\]
																																	This completes the proof of
																																	Claim \ref{theo:Curial_Ind_Step1_2}.
																																	
																																	Consequently,
																																	\[
																																	\sum_{k\in \K_1}C_k(\tilde{f}^{(n)})\approx\sum_{k\in \K_1}
																																	C_k(f^{*(n)})\in \Theta\Big( T_1(d^{(n)})\cdot g_n^{(1)}\Big).
																																	\]
																																	\Wu{With} Claim \ref{theo:Induction_Step_0},
																																	we \Wu{then} further obtain that
																																	\[
																																	\sum_{k\in \bigcup_{u=0}^{1} \K_u }C_k (\tilde{f}^{(n)}),
																																	\sum_{k\in \bigcup_{u=0}^{1} \K_u }C_k (f^{*(n)})
																																	\in \Theta\big(\max\{T_0(d^{(n)})g_n^{(0)},T_1(d^{(n)})g_n^{(1)}\}\big).\]
																																	\end{proof}
																																	
																																	 \Wu{If}
																																	$\K=\K_0\cup\K_1,$ then we have already finished
																																	the decomposition and completed
																																	the whole proof \Wu{with} Claim \ref{theo:Curial_Ind_Step1_2}. Otherwise, we can continue with \Wu{an} 
																																	argument \Wu{similar} to those \Wu{at} steps $m=0,1.$ 
																																	Note that this procedure will eventually
																																	terminate, since the number $|\K|=K$ of groups is finite.
																																	\Wu{We now outline} the
																																	general inductive step $m=l$ for some integer
																																	$l=0,\ldots,K$.
																																	
																																	\textbf{Step $m=l:$ construct $\K_{l},$ and
																																		prove \eqref{eq:InductiveStepObj} for $m=l.$}
																																		
																																		We assume that we have
																																		partitioned $\K$ into 
																																		$\K_0, \ldots, \K_{l-1}, \K\backslash
																																		\bigcup_{u=0}^{l-1}\K_u$ for some
																																		integer $l=0,\ldots,K,$ where
																																		we use a convention that
																																		$\bigcup_{u=0}^{-1}\K_u=\emptyset,$ and
																																		$\K_{-1}=\emptyset.$ Moreover, we \Wu{make
																																			the following inductive assumptions.}
																																		\begin{itemize}
		\item[$IA1.$] The limit																															\[																																\lim_{n\to \infty} \frac{T_{u+1}(d^{(n)})}{T_u(d^{(n)})}=0,
	\quad\text{and}\quad 
	\lim_{n\to\infty}\frac{\sum_{k\in \K_{l-1}}d_k^{(n)}}{T_{l-1}}=1,																																	\]
																																		for each $u=0,1,\ldots,l\!-\!2,$ where  each
		$T_u(d^{(n)})=\sum_{k\in \K\backslash\bigcup_{i=0}^{u-1}\K_i}d_k^{(n)},$ for each $n\in \N$ for $u=0,\ldots,l-1.$ 
																																			\item[$IA2.$] 											Define 
																																			\[
																																			\alpha_u:=\max\{\rho_k: k\in\K_u\}
																																			\]
																																			for each $u=0,1,\ldots,l-1.$ 
																																			\item[$IA3.$] For each $k\in \bigcup_{u=0}^{l-1}\K_u,$ 
																																			\[ 
																			\tilde{L}_k^{(n)},																L_k^{*(n)}\in 	O\big(\max\{g_n^{(0)},\ldots,g_n^{(l-1)}\}\big),
																																			\]
																																			where each $g_n^{(u)}=T_u(d^{(n)})^{\alpha_u},$
																																			for $u=0,\ldots,l-1.$
																																			
																																			\item [$IA4.$] For
																																			each $u=0,\ldots,l-1,$
																																			\[
																																			\lim_{n\to\infty} \frac{\sum_{k\in \bigcup_{i=0}^{u}\K_i}C_k\big(\tilde{f}^{(n)}\big)}{
																																				\sum_{k\in \bigcup_{i=0}^{u}\K_i}C_k\big(f^{*(n)}\big)
																																				}=1,
																																				\]
																																				and
																																				\[
		\begin{split}																															\sum_{k\in \bigcup_{i=0}^{u}\K_i}C_k\big(f^{*(n)}\big),
																																				&\sum_{k\in\bigcup_{i=0}^{u}\K_i} C_k\big(\tilde{f}^{(n)}\big)\\
																																				&\in \Theta\Big(
																																		\max\{T_0(d^{(n)})g_n^{(0)},\ldots,T_u(d^{(n)})g_{n}^{(u)}\}																																\Big).
\end{split}																				\]
																																				\end{itemize}
																																				
																																				Note that 
																																				we only need to check the validity of 																											$IA3$-$IA4$ at each step. $IA1$
follows immediately from the definition
of total user volumes $T_u(d^{(n)})$ of
the marginals, and $IA2$ defines the ``degree"
of each $\K_u$ for $u=0,\ldots,l-1$.

																																				Again, if $\bigcup_{u=0}^{l-1}\K_u=\K,$ then we have already
finished the decomposition and completed the
proof. Otherwise, we can \Wuu{apply an argument similar to} \Wu{those above.} Due to the heavy similarity,
we now only
list the key components at this general step
$m=l$, but 
omit the detailed proof.

Similarly we need the following assumptions: 
\begin{itemize}
	
																																					\item The limit
																																					\[
																																					\lim_{n\to \infty}\frac{d_k^{(n)}}{T_{l}(d^{(n)})}=:\d_k^{(l,\infty)}
																					\in [0,1]																\]
																																					\Wu{exists} for each $k\in \K\backslash\bigcup_{u=0}^{l-1}\K_{u},$
where $T_l(d^{(n)})=\sum_{k\in \K\backslash\bigcup_{u=0}^{l-1}\K_u}d_k^{(n)}.$
																																					\item The limit
																																					\[
																																					\lim_{n\to \infty} \frac{\tilde{f}_s^{(n)}}{T_{l}(d^{(n)})}=:\tilde{\f}^{(l,\infty)}_s
																																					\quad \text{and}\quad
																																					\lim_{n\to \infty} \frac{{f}_s^{*(n)}}{T_{l}(d^{(n)})}=:\f^{*(l,\infty)}_s
																																					\]
																																					\Wu{exist} for some constants $\tilde{\f}^{(l,\infty)}_s,\f^{*(l,\infty)}_s\in [0,1]$ for each $s\in \S_k,$ and each $k\in \K\backslash\bigcup_{u=0}^{l-1}\K_u.$
																																					\end{itemize}
																																					Otherwise, we can \Wu{again} take \Wu{an} infinite subsequence to fulfill
																																					these \Wu{assumptions}.

																											Similarly,	we define
																																					\[
																																					\alpha_{l}:=\max\Big\{\rho_k: k\in \K\backslash\bigcup_{u=0}^{l-1}\K_u,
																																					\d_k^{(l,\infty)}>0\Big\},
																																					\]
																																					and put
																																					\[
																																					\K_{l}:=\Big\{k\in \K\backslash\bigcup_{u=0}^{l-1}\K_u: \rho_k\le
																																					\alpha_{l}\Big\}.
																																					\]
																																					Obviously, $\K_l,\alpha_l$ and 																$T_l(d^{(n)})$ together validate $IA1$-$IA2$ for step $m=l.$ 
Moreover, there exists $k\in \K_l$ such that
																																					$\d_k^{(1,\infty)}>0$ and $\rho_k=\alpha_l.$

																																					Similarly, \Wuu{we} put
																																					\[
																																					\tau_a^{(l,n)}(x):=\tau_a\Big(x+\tilde{f}_a^{(n)}(\bigcup_{u=0}^{l-1}\K_u)\Big)
																																					\]
																																					and
																																					\[
																																					c_a^{(l,n)}(x):=c_a\Big(x+{f}_a^{*(n)}(\bigcup_{u=0}^{l-1}\K_u)\Big),
																																					\]
																																					be the price functions \Wu{under
																																						the condition that} 
																																					users from groups in $\bigcup_{u=0}^{l-1}\K_u$
																																					stick to the strategies they used in 
																																					NE profiles $\tilde{f}^{(n)}$ and SO profiles
																																					$f^{*(n)},$ respectively, for each $n\in \N$
																																					and each $a\in A.$ 
																																					\Wu{Here}
																																					\[
																																					\tilde{f}_a^{(n)}(\bigcup_{u=0}^{l-1}\K_u):=\sum_{k\in \bigcup_{u=0}^{l-1}\K_u}
																																					\sum_{s\in \S_k} r(a,s)\cdot \tilde{f}_s^{(n)}
																																					\]
																																					and 
																																					\[
																																					{f}_a^{*(n)}(\bigcup_{u=0}^{l-1}\K_u):=\sum_{k\in \bigcup_{u=0}^{l-1}\K_u}
																																					\sum_{s\in \S_k} r(a,s)\cdot {f}_s^{*(n)}
																																					\]
																																					are \Wu{the} volumes of resource $a$ consumed by users from groups $\in
																																					\bigcup_{u=0}^{l-1}\K_u$ w.r.t. \Wu{the} profiles $\tilde{f}^{(n)}$ and $f^{*(n)},$  
																																					respectively, for each $n\in \N$ and each $a\in A.$
																																					
																																					Obviously, 
																																					\[
																																					\tau_a(\tilde{f}_a^{(n)})=\tau_a^{(1,n)}\Big(\tilde{f}_a^{(n)}\big(\K\backslash\bigcup_{u=0}^{l-1}\K_u\big)\Big)\text{ and }
																																					c_a(f_a^{*(n)})=c_a^{(1,n)}\Big(f_a^{*(n)}\big(\K\backslash\bigcup_{u=0}^{l-1}\K_u\big)\Big)
																																					\]
				\Wu{for each $n\in \N,$ and each resource $a\in A.$	}																																Here 
																																					\[
																																					\tilde{f}_a^{(n)}\big(\K\backslash\bigcup_{u=0}^{l-1}\K_u\big):=
																																					\tilde{f}_a^{(n)}-\tilde{f}_a^{(n)}\big(\bigcup_{u=0}^{l-1}\K_u\big)
																																					=\sum_{k\in \K\backslash\bigcup_{u=0}^{l-1}\K_u}\sum_{s\in \S_k}
																																					r(a,s)\cdot \tilde{f}_s
																																					\]
																																					denotes the volume of resource $a$ consumed by users from 
																																					groups $k\in \K\backslash\bigcup_{u=0}^{l-1}\K_u$
																																					w.r.t. NE profile $\tilde{f}^{(n)}.$  \Wu{This yields}
																																					similarly for $f_a^{*(n)}\big(\K\backslash\bigcup_{u=0}^{l-1}\K_u\big).$
																																					
																																					\Wu{By the inductive assumption $IA3$ and
																																					a simlar argument to} the proof of
																																					Claim \ref{theo:Crucial_Induc_Step_1},  we obtain for each $a\in A$ that
																																					\[
																																					\tau_a\Big(\tilde{f}_a^{(n)}\big(\bigcup_{u=0}^{l-1}\K_u\big)\Big)
																																					 \in O\big(\max\{g_n^{(0)},\ldots,g_n^{(l-1)}\}\big)
																																					\]
																																					and																											\[
																																					c_a\Big(f_a^{*(n)}\big(\bigcup_{u=0}^{l-1}\K_u\big)\Big)
									\in
																																					O\big(\max\{g_n^{(0)},\ldots,g_n^{(l-1)}\}\big).
																																					\]
Since both $\tau_a(\cdot)$ and
$c_a(\cdot)$ are asymptotically 
non-decreasing, we thus obtain  that
\[
																																					c_a^{(1,n)}\Big(T_l(d^{(n)})x\Big),\tau_a^{(1,n)}\Big(T_l(d^{(n)})x\Big)\!\in\! O\Big(\max\big\{T_l(d^{(n)})^{\rho_a},g_n^{(0)},\ldots,g_n^{(l-1)}\big\}\Big)
																																					\]
\Wu{for each $a\in A$
and each $x\ge 0.$	}										
																																					\Wu{An  argument similar to that for Claim \ref{theo:Crucial_Induc_Step_1} the gives:}
																																					\begin{claim}\label{theo:Crucial_Induc_Step_l}
																																						For each $k\in \K_l,$ 
																																						\[
																																						\tilde{L}_k^{(n)},L_k^{*(n)}\in O\Big(\max\{g_n^{(0)},\ldots,g_n^{(l)}\}\Big),
																																						\]
where $g_n^{(l)}=T_l(d^{(n)})^{\alpha_l}.$
																																						\end{claim}
																																						
																																					 Claim \ref{theo:Crucial_Induc_Step_l} validates
																																						the inductive assumption $IA3$ for step $m=l.$ 
																																			
																																							To validate the inductive assumption $IA4$ for step $m=l,$ we assume that the limit
\[
\lim_{n\to\infty} \frac{g_l^{(n)}}{\max\{g_n^{(0)},\ldots,g_n^{(l-1)}\}}=\beta_{l-1}\in [0,\infty]
\]
exists for some constant $\beta_{l-1}.$	
In the case that $\beta_{l-1}<\infty,$ we obtain
that $g_n^{(l)}\in O\big(\max\{g_n^{(0)},\ldots,g_n^{(l-1)}\}\big),$
which implies that groups $\in\K_l$ are negligible
w.r.t. groups $\in \bigcup_{u=0}^{l-1}\K_u,$
since $T_l(d^{(n)})\in o\big(T_{l-1}(d^{(n)})\big).$
Hence, if $\beta_{l-1}<\infty,$ then $IA4$ is valid for step $m=l.$

In the case that $\beta_{l-1}=\infty,$ we obtain that
$g_n^{(l)}\in \omega\big(\max\{g_n^{(0)},\ldots,g_n^{(l-1)}\}\big),$
which implies that the behavior of users from groups
$\K_l$ is asymptotically independent of users from 
groups $\bigcup_{u=0}^{l-1}\K_u.$ Then, \Wu{an}
argument \Wu{similar to} that \Wu{for} Claim \ref{theo:Curial_Ind_Step1_2} applies.
\Wu{This yields} Claim \ref{theo:Crucial_Induc_Step_l_2} below, \Wu{and}
validates $IA4$ for step $m=l.$
\begin{claim}\label{theo:Crucial_Induc_Step_l_2}	\[
																																								\lim_{n\to\infty} \frac{\sum_{k\in \bigcup_{i=0}^{l}\K_i}C_k\big(\tilde{f}^{(n)}\big)}{
																																									\sum_{k\in \bigcup_{i=0}^{l}\K_i}C_k\big(f^{*(n)}\big)
																																									}=1,
																																									\]
																																									and
																																									\[
																																									\sum_{k\in \bigcup_{i=0}^{l}\K_i}C_k\big(f^{*(n)}\big),
																																									\sum_{k\in\bigcup_{i=0}^{l}\K_i} C_k\big(\tilde{f}^{(n)}\big)\in \Theta\Big(
																																									\max\{T_0(d^{(n)})g_n^{(0)},\ldots,T_l(d^{(n)})g_n^{(l)}\}												\Big).
																																									\]
																																									\end{claim}

																																																																			All \Wu{above} together validate the inductive assumptions
																																														$IA1$-$IA4$ for step $m=l.$ 																	
																																														\Wu{So the induction completes}, \Wu{and}
																																																Theorem \ref{theo:SelfishMainTh} \Wu{is proved}.
\end{proof}																																
\subsection*{Proof of Lemma \ref{theo:RV_Properties}}
\begin{proof}[Proof of Lemma \ref{theo:RV_Properties}]
	\Wu{A proof \Wuu{of} Lemma \ref{theo:RV_Properties}
		may already exist in \cite{Bingham1987Regular}.
		However, we cannot directly access \cite{Bingham1987Regular}.
		Our knowledge on regular variation
		is actually indirectly obtained from Wikipedia on the page
		\[
		\text{https://en.wikipedia.org/wiki/Slowly\_varying\_function.}
		\]
		Therefore, we supply a detailed proof to ensure \Wuu{completeness of this paper}.}
	
	\textbf{Proof of $a):$} Let $\epsilon>0$ be an \Wu{arbitrarily} fixed constant.
	\Wuu{Using} Karamata's Characterization Theorem and
	Representation Theorem, see, e.g., \cite{Bingham1987Regular},
	we can write 
	\[
	\tau(x)=x^{\rho}\cdot e^{\eta(x)+\int_{b}^{x}\frac{\xi(t)}{t} dt},
	\] 
	where 
	\begin{itemize}
		\item $\eta(x)$ is a real-valued measurable function such that 
		$\lim_{x\to \infty}\eta(x)=p$ for some constant $p\ge 0,$
		\item $b\ge 0$ is a constant, and $\xi(x)$ is a real-valued measurable function such that $\lim_{x\to \infty}\xi(x)=0.$
	\end{itemize}
	Thus, we obtain that
	\[
	\begin{split}
	\lim_{x\to \infty} \frac{\tau(x)}{x^{\rho+\epsilon}}
	&=\lim_{x\to \infty}e^{-\epsilon\cdot\ln x+\eta(x)+\int_{b}^{x}\frac{\xi(t)}{t}dt}
	=\lim_{x\to \infty}e^{-\int_{1}^{x}\frac{\epsilon}{t}dt+p+\int_{b}^{x}\frac{\xi(t)}{t}dt}\\
	&=\lim_{x\to \infty}e^{-\int_{1}^{x}\frac{\epsilon-\xi(t)}{t}dt+p+\int_{b}^{1}\frac{\xi(t)}{t}dt}=0,
	\end{split}
	\]
	where we observe that
	\[
	\lim_{x\to \infty}\int_{1}^{x}\frac{\epsilon-\xi(t)}{t}dt=\infty,
	\]
	since $\xi(t)\to 0$ as $t\to \infty,$ and $\epsilon>0.$
	
	Similarly, one can prove that
	\[
	\lim_{x\to\infty} \frac{\tau(x)}{x^{\rho-\epsilon}}=\infty.
	\]
	
	\textbf{Proof of $b):$} For each $x>0,$
	\[
	\lim_{t\to\infty}\frac{g(tx)}{g(t)}
	=\lim_{t\to\infty}\frac{\frac{g(tx)}{\tau(tx)}}{\frac{g(t)}{\tau(t)}}
	\cdot\lim_{t\to\infty} \frac{\tau(tx)}{\tau(t)}=
	x^{\rho}.
	\]
	
	\textbf{Proof of $c):$} For each $x>0,$
	\[
	\lim_{t\to\infty} \frac{\frac{g(tx)}{\tau(tx)}}{\frac{g(t)}{\tau(t)}}
	=\lim_{t\to \infty} \frac{g(tx)}{g(t)}\cdot
	\lim_{t\to \infty} \frac{1}{\frac{\tau(tx)}{\tau(t)}}
	=x^{\rho'}\cdot x^{-\rho}=x^{\rho'-\rho}\in (0,\infty).
	\]
\end{proof}	

\subsection*{Proof of 
	Lemma \ref{theo:RegularlyVaryingConvex}}
\begin{proof}[Proof of Lemma \ref{theo:RegularlyVaryingConvex}]
	We assume that $\tau(x)$ is nondecreasing,
	non-negative, convex, differentiable and regularly 
	varying
	with \Wuu{index} $\rho\in \mathbb{R}$.
	
	We first show that $\rho\ge 0.$
	\Wuu{Since $\tau(\cdot)$ is convex, 
	\[
	\tau\big((1+\eta)x\big)=\tau\Big((1-\eta)\cdot x+\eta\cdot 2x\Big)
	\le (1-\eta)\tau(x)+\eta\tau(2x)
	\]
	for each $x\ge 0$ and
	each $\eta\in [0,1].$}
	Therefore, 
	\[
	\frac{\tau\big((1+\eta)x\big)}{\tau(x)}
	\le 1-\eta + \eta\frac{\tau(2x)}{\tau(x)}
	\]
	\Wuu{for each $x$ with $\tau(x)>0.$}
	Letting $x\to \infty,$ we obtain from the 
	regular variation of $\tau(\cdot)$ that
	\begin{equation}\label{eq:ConvexityOfTau}
	\big(1+\eta\big)^{\rho}\le 1-\eta+\eta\cdot 2^{\rho}
	\end{equation}
	\Wuu{for each $\eta\in [0,1],$} which, in turn, implies that $\rho\ge 0.$
	\Wuu{Otherwise,} if $\rho<0,$ then
	\[
	\frac{\partial}{\partial \eta}\Big(1-\eta+\eta\cdot 2^{\rho}-\big(1+\eta\big)^{\rho}\Big)
	=2^{\rho}-1-\rho\big(1+\eta\big)^{\rho-1}<0,
	\]
	when 
	\[
	0\le \eta<\Big(\frac{\rho}{2^{\rho}-1}\Big)^{\frac{1}{1-\rho}}-1
	\in (0,1).
	\]
	Therefore, if $\rho<0,$ then 
	we \Wuu{obtain for these $\eta$}
	that
	\[
	1-\eta+\eta\cdot 2^{\rho}-\big(1+\eta\big)^{\rho}
	<1-0+0\cdot 2^{\rho}-\big(1+0\big)^{\rho}=0,
	\]
	which \Wu{contradicts} \eqref{eq:ConvexityOfTau}.
	
	\Wu{So}, convexity of $\tau(\cdot)$ implies
	that $\rho\ge 0.$
	
	\Wuu{Convexity} and differentiability of $\tau(\cdot)$
	further imply for each $t>0$ and $x>0$ that
	\[
	\frac{1}{t}x\cdot \tau'(x)\le 
	\int_{x}^{\big(1+\frac{1}{t}\big)x}
	\tau'(u)du=\tau\Big(\big(1+\frac{1}{t}\big)x\Big)
	-\tau(x).
	\]
	Therefore, \Wuu{
	\[
	\frac{x\cdot \tau'(x)}{t\tau(x)}
	\le \Big(\frac{\tau\big((1+\frac{1}{t})x\big)}{\tau(x)}-1\Big)
	\]
	for each $t>0$ and $x>0.$}
	Letting $x\to \infty,$ \Wuu{the regular variation of $\tau(x)$ yields
	\[
	\varlimsup_{x\to\infty}
	\frac{x\tau'(x)}{\tau(x)}
	\le t\Big(\big(1+\frac{1}{t}\big)^{\rho}-1\Big)
	\]
	for each $t>0.$}
	Note that $\rho\ge 0$ implies that
	\[
	\lim_{t\to\infty}t\Big(\big(1+\frac{1}{t}\big)^{\rho}-1\Big)
	=\lim_{z\to 0}\frac{(1+z)^{\rho}-1}{z}= \rho.
	\]
	\Wuu{So, altogether, if} $\tau(\cdot)$ is convex and differentiable,
	\Wuu{then}
	\[
	\varlimsup_{x\to\infty}
	\frac{x\tau'(x)}{\tau(x)}\le \rho.
	\]
	
	Similarly, \Wuu{using again} the convexity and differentiability
	of $\tau,$ we obtain for each $t>1$ and
	each $x>0$ that
	\[
	\frac{1}{t}x\tau'(x)\ge \int_{\big(1-\frac{1}{t}\big)x}^{x}
	\tau'(u)du=\tau(x)-\tau\Big(\big(1-\frac{1}{t}\big)x\Big).
	\]
	\Wuu{An} almost identical argument to the above \Wuu{then yields}
	\[
	\varliminf_{x\to \infty}
	\frac{x\tau'(x)}{\tau(x)}\ge 
	\lim_{z\to 0} \frac{1-(1-z)^{\rho}}{z}=\rho.
	\]

	\Wuu{Altogether,} we obtain that 
	\[
	\lim_{x\to\infty} \frac{x\tau'(x)}{\tau(x)}=\rho\ge 0,
	\]
	which completes the proof.
\end{proof}

\subsection*{Proof of Theorem \ref{theo:MainTheoremRegularVariation}}										
\begin{proof}[Proof of Theorem \ref{theo:MainTheoremRegularVariation}]
	The proof is similar \Wu{to} that \Wu{for} Theorem
	\ref{theo:SelfishMainTh}. 
	To save space, we only sketch
	the main idea and give details for some
	crucial points.
	
	Consider an \Wu{arbitrary} sequence 
	$\{d^{(n)}\}_{n\in \N}$ \Wu{of user volume
		vectors} such that 
	$\lim_{n\to \infty} T(d^{(n)})=\infty,$
	and \Wu{consider an} NE profile $\tilde{f}^{(n)}$ and 
	\Wu{an} SO profile
	$f^{*(n)}$ \Wu{for each $n\in \N.$}
	We \Wuu{want} to apply the asymptotic decomposition
	to this sequence. The Theorem \Wu{will follow directly} from the
	\Wu{arbitrary choice} of the sequence $\{d^{(n)}\}_{n\in \N}$.
	
	\Wu{To construct suitable
		scaling factors at each
		inductive step, we now aim to define a \Wuu{suitable} ordering $\preceq$ on
		the set $A$ of} 
	resources. 
	\Wu{Note that the degrees of
		the polynomial price functions
		serve as \Wuu{such an} ordering
		in the proof of Theorem \ref{theo:SelfishMainTh}.
		Here, \Wuu{the} $\tau_a(\cdot)$ are generally no longer \Wuu{polynomials,} thus \Wuu{such an} ordering \Wuu{is not}
		so obvious.}
	
	Let $\rho_a$ \Wuu{denote} the regular variation 
	index of $\tau_a(\cdot)$ for each $a\in A.$
	Note that the ``degrees" $\rho_a$ now
	\Wuu{cannot be used directly as such an} ordering
	on the set $A.$ 
	\Wuu{We instead use} the mutual comparability of
	the price functions \Wuu{to} define a \Wu{suitable}  ordering
	$\preceq$ on the set
	$A.$ We put
	\[
	a \preceq b
	\iff \lim_{x\to \infty}\frac{\tau_a(x)}{\tau_b(x)}=q_{a,b}
	<\infty. 
	\]
	Obviously, by Lemma \ref{theo:RV_Properties}, $\rho_a<\rho_b$ implies that
	$a\preceq b.$ \Wuu{But the inverse \Wu{need} not be true.}

	\Wu{This} ordering on
	$A$ also carries over to the case that the price functions are \Wuu{$c_a(x)=x\tau_a'(x)+\tau_a(x).$}
	By Lemma \ref{theo:RegularlyVaryingConvex} and
	Lemma \ref{theo:RV_Properties}, \Wuu{and
		the convexity of
		each $\tau_a(\cdot),$} we obtain 
	immediately for
	each $a\in A$ that $c_a(\cdot)$ is also regular
	varying with \Wu{the same} index $\rho_a,$ and 
	\[
	\lim_{x\to\infty}\frac{c_a(x)}{\tau_a(x)}=\rho_a+1.
	\]
	
	For each $s\in \S,$ let $\bar{a}_s$ be \Wuu{a} maximum 
	element of $\{a\in A: r(a,s)>0\}$ w.r.t.
	the ordering $\preceq.$ Then, we obtain
	\Wu{by the mutual comparability} that
	\[
	q_{\bar{a}_s, \bar{a}'_s}=\lim_{x\to\infty} \frac{\tau_{\bar{a}_s}(x)}{\tau_{\bar{a}'_s}(x)}\in (0,\infty),
	\]
	where $\bar{a}'_s$ is another maximum element 
	of $\{a\in A:r(a,s)>0\}$ w.r.t. \Wuu{the} ordering
	$\preceq.$  \Wu{This means that these maximum elements
		are mutually equivalent w.r.t. the ordering $\preceq.$} Hence, by
	Lemma \ref{theo:RV_Properties}, two maximum elements
	of $\{a\in A:r(a,s)>0\}$ have the same regular
	variation index. Again, the inverse \Wu{need} not
	be true. 
	
	With the ordering $\preceq,$ we can now 
	easily \Wuu{identify}  the {\em cheapest} strategy 
	$s^*_k$ for each group 
	$k\in \K.$ Obviously, $s^*_k$ is \Wuu{a} strategy
	$s\in \S_k$  such that $\bar{a}_s\preceq \bar{a}_{s'}$
	for each $s'\in \S_k.$
	
	The  ordering $\preceq$ on resource
	set $A$ \Wuu{are crucial} \Wu{for the construction of} the scaling
	factors \Wuu{of the marginal games} \Wu{at each inductive step} in the asymptotic
	decomposition. 
	
	\textbf{step $m=0:$ Construct $\K_0$ and the scaling
		factors $g_n^{(0)}$}

	Let us put $T_0(d^{(n)})=T(d^{(n)})$ for each
	$n\in \N.$
	
	\Wuu{Similar to the proof of Theorem \ref{theo:SelfishMainTh}}, we assume w.l.o.g. that  the limits 
	\[
	\lim_{n\to\infty}
	\frac{d_k^{(n)}}{T_0(d^{(n)})}
	\!=\!\d_k^{(0,\infty)},\quad\tilde{\f}_s^{(0,\infty)}\!=\!\lim_{n\to\infty}
	\frac{\tilde{f}_s^{(n)}}{T_0(d^{(n)})},
	\quad\text{and}\quad 
	\f_s^{*(0,\infty)}\!=\!\lim_{n\to\infty}
	\frac{f_s^{*(n)}}{T_0(d^{(n)})}
	\]
	exist for each $k\in \K$ and $s\in S$. We put $\alpha_0=\max_{\preceq}\{\bar{a}_{s_k^*}: k\in \K,
	\d_k^{(0,\infty)}>0\}\in A,$ where the \Wu{maximization}
	is w.r.t. the ordering $\preceq.$
	Note that if there are multiple \Wuu{maxima}, then
	we \Wu{can} \Wuu{pick} \Wu{arbitrary} one \Wuu{of} them.
	Moreover, we put $\K_0=\{k\in \K: \bar{a}_{s_k^*}
	\preceq \alpha_0\},$ and $g_n^{(0)}=\tau_{\alpha_0}\big(T_0(d^{(n)})\big)$
	for each $n\in \N.$
	
	\Wu{With} the user optimality \eqref{def:WE}, \Wu{we} can easily obtain for each $k\in \K_0$ that 
	\[
	\tilde{L}_k^{(n)},L_k^{*(n)}
	\in O\big(g_n^{(0)}\big)
	=O\big(\tau_{\alpha_0}(T_0(d^{(n)}))\big),
	\]
	since there exists an strategy $s\in \S_k$
	such that $\bar{a}_s\preceq \alpha_0$ for each
	$k\in \K_0,$ and the maximum volume
	of users adopting a strategy is in $O(T_0(d^{(n)})).$
	
	\Wu{With} the regular variation of the price functions
	$\tau_a(\cdot)$ and $c_a(\cdot)$, we obtain
	for each $x>0$ that
	\[
	\lim_{n\to \infty} \frac{\tau_a(T(d^{(n)})x)}{
		g_n^{(0)}
		}=\lim_{n\to \infty} \frac{\tau_a(T(d^{(n)})x)}{
		\tau_{\alpha_0}(T(d^{(n)}))
	}=q_{a,\alpha_0}\cdot x^{\alpha_0},
	\]
	and 
	\[
	\lim_{n\to \infty} \frac{c_a(T(d^{(n)})x)}{
		g_n^{(0)}
	}=\lim_{n\to \infty} \frac{c_a(T(d^{(n)})x)}{
	\tau_{\alpha_0}(T(d^{(n)}))
}=(1+\rho_{\alpha_0})q_{a,\alpha_0}\cdot x^{\alpha_0}.
	\]
\Wu{Therefore}, we can \Wu{then} obtain from \Wu{an} 
argument \Wu{similar to} that \Wu{for} Claim \ref{theo:Induction_Step_0} that
\[
\sum_{k\in \K_0} C_k(\tilde{f}^{(n)})
\approx \sum_{k\in \K_0} C_k(f^{*(n)})
\in \Theta\big(T_0(d^{(n)})g_n^{(0)}\big).
\]
\Wu{Here,} we observe that there exists at least
one group $k\in \K_0$ such that $\bar{a}_{s_k^*}=\alpha_0$
and $\d_k^{(0,\infty)}>0.$

\textbf{step $m=l:$ Construct $\K_l$ and scaling
	factors $g_n^{(l)}$}

\Wuu{Similar to the proof of Theorem \ref{theo:SelfishMainTh}}, we \Wu{now make inductive assumptions} that we have constructed 
$\K_0,\ldots,\K_{l-1}$ and 
$g_n^{(0)}, \ldots, g_n^{(l-1)}$ such that
for each $k\in \bigcup_{u=0}^{l-1}\K_u$
\[
\tilde{L}_k^{(n)}, L_k^{*(n)}
\in O\big(\max\{g_n^{(0)},\ldots,g_n^{(l-1)}\}\big),
\]
and 
\[
\begin{split}
\sum_{k\in \bigcup_{u=0}^{l-1}\K_u} C_k(\tilde{f}^{(n)})
\approx &\sum_{k\in \bigcup_{u=0}^{l-1}\K_u} C_k(f^{*(n)})\\
&\in \Theta\big(\max\{T_0(d^{(n)})g_n^{(0)}, \ldots,T_{l-1}(d^{(n)})g_n^{(l-1)}\}\big).
\end{split}
\]
Moreover, we assume for each 
$u=0,\ldots,l-2$ \Wu{that }
\[
\lim_{n\to\infty} \frac{T_{u+1}(d^{(n)})}{T_{u}(d^{(n)})}=0,\quad\text{and}\quad
\lim_{n\to\infty}
\frac{\sum_{k\in \K\backslash\bigcup_{i=0}^{l-1}\K_i}d_k^{(n)}}{T_{l-1}(d^{(n)})}=0.
\]

To construct $\K_l$ and $g_n^{(l)},$ we \Wu{make a }further 
\Wu{assumption} that the limits 
\[
\lim_{n\to\infty}\frac{d_k^{(n)}}{T_l(d^{(n)})}
=\d_k^{(l,\infty)}, \ 
\lim_{n\to\infty} \frac{\tilde{f}_s^{(n)}}{T_l(d^{(n)})}
=\tilde{\f}_s^{(l,\infty)},\
\text{and}\ 
\lim_{n\to \infty}
\frac{f^{*(n)}_s}{T_l(d^{(n)})}
=\f^{*(l,\infty)}_s
\]
exist for each $s\in \S_k$ and 
$k\in \K\backslash\bigcup_{u=0}^{l-1}\K_u,$
where $T_l(d^{(n)})=\sum_{k\in \K\backslash\bigcup_{u=0}^{l-1}\K_u}d_k^{(n)}$
for each $n\in \N.$

Let $\alpha_{l}=\max_{\preceq}\{\bar{a}_{s_k^*}: k
\in\K\backslash\bigcup_{u=0}^{l-1}\K_u,\d_k^{(l,\infty)}>0\}$ and $\K_l=\{k\in \K\backslash\bigcup_{u=0}^{l-1}\K_u: 
\bar{a}_{s^*_k}\preceq \alpha_l\}.$
Then, we put $g_{n}^{(l)}=\tau_{\alpha_l}(T_l(d^{(n)})).$

Similarly, we can obtain \Wu{for each $k\in \K_l$} that
\[
\tilde{L}_k^{(n)},L_k^{*(n)}
\in O\big(\max\{g_n^{(0)},\ldots,g_n^{(l)}\}\big),
\]
since both $\tau_a(\cdot)$ and
$c_a(\cdot)$ are non-decreasing, and \Wu{thus} for each $a\in A$ and each $x>0$
\[
\begin{split}
\tau_a^{(l,n)}\big(T_l(d^{(n)})x\big)
&=\tau_a\big(T_l(d^{(n)})x+\tilde{f}^{(n)}_a(\bigcup_{u=0}^{l-1}\K_u)\big)\\
&\in O\big(\max\{\tau_a(T_l(d^{(n)})),
\tau_a(\tilde{f}^{(n)}_a(\bigcup_{u=0}^{l-1}\K_u))
\}\big)
\end{split}
\]
and
\[
\begin{split}
c_a^{(l,n)}\big(T_l(d^{(n)})x\big)
&=c_a\big(T_l(d^{(n)})x+f^{*(n)}_a(\bigcup_{u=0}^{l-1}\K_u)\big)\\
&\in O\big(\max\{c_a(T_l(d^{(n)})),
c_a(f^{*(n)}_a(\bigcup_{u=0}^{l-1}\K_u))
\}\big).
\end{split}
\]

\Wuu{So,} \Wu{if}
$g_n^{(l)}\in O\big(\max\{g_n^{(0)},\ldots,g_n^{(l-1)}\}\big),$
\Wu{then all} groups \Wu{in} $\K_l$ \Wuu{are} negligible w.r.t. 
groups $\bigcup_{u=0}^{l-1}\K_u.$
Otherwise, $g_n^{(l)}\in \omega
\big(\max\{g_n^{(0)},\ldots,g_n^{(l-1)}\}\big).$
\Wu{If this is the case,} we \Wu{can then } independently consider $\K_l$
\Wu{with an argument similar to} that \Wu{for} Theorem
\ref{theo:SelfishMainTh}.
\Wu{This will yield}
\[
\begin{split}
\sum_{k\in \bigcup_{u=0}^{l}\K_u} C_k(\tilde{f}^{(n)})
\approx &\sum_{k\in \bigcup_{u=0}^{l}\K_u} C_k(f^{*(n)})\\
&\in \Theta\big(\max\{T_0(d^{(n)})g_n^{(0)}, \ldots,T_{l-1}(d^{(n)})g_n^{(l-1)},T_l(d^{(n)})g_n^{(l)}\}\big),
\end{split}
\]
which completes \Wu{the induction and finishes} the proof.
\end{proof}

\subsection*{Proof of Theorem \ref{theo:MainTheoremGaugeableExtension}}
\begin{proof}[Proof of Theorem \ref{theo:MainTheoremGaugeableExtension}]
	This proof is very similar \Wu{to} those for
	Theorem \ref{theo:SelfishMainTh} and Theorem \ref{theo:MainTheoremRegularVariation}.
	We thus only sketch the \Wu{main idea}.
	
	\Wu{Note that we will not need a \Wuu{suitable} ordering on
		resources in this proof. Conditions
		$G1')$-$G3')$ have already guaranteed
		the existence of a suitable scaling factor sequence
		$\{g_n^{(l)}\}_{n\in \N}$ at each
		inductive step.}
	
	We assume that we are now at an inductive step
	$m=l,$ and we have already shown that
	for each $k\in \bigcup_{u=0}^{l-1}\K_u$
	\[
	\tilde{L}_k^{(n)}, L_k^{*(n)}
	\in O\big(\max\{g_n^{(0)},\ldots,g_n^{(l-1)}\}\big),
	\]
	and 
	\[
	\begin{split}
	\sum_{k\in \bigcup_{u=0}^{l-1}\K_u} C_k(\tilde{f}^{(n)})
	\approx &\sum_{k\in \bigcup_{u=0}^{l-1}\K_u} C_k(f^{*(n)})\\
	&\in \Theta\big(\max\{T_0(d^{(n)})g_n^{(0)}, \ldots,T_{l-1}(d^{(n)})g_n^{(l-1)}\}\big),
	\end{split}
	\]
	where all the notations are the same as those \Wu{for} the proofs
	of 
	Theorem \ref{theo:SelfishMainTh} and Theorem \ref{theo:MainTheoremRegularVariation}.
	Moreover, \Wu{we make \Wuu{the} further inductive \Wuu{assumptions} that}
	\[
	\lim_{n\to\infty}\frac{T_{u+1}(d^{(n)})}{T_u(d^{(n)})}=0,
	\quad\text{and}\quad 
	\lim_{n\to\infty} \frac{\sum_{k\in\K\backslash\bigcup_{u=0}^{l-1}\K_u}d_k^{(n)}}{T_{l-1}(d^{(n)})}=0
	\]
	for each $u=0,\ldots,l-2.$
	If $\K=\bigcup_{u=0}^{l-1}\K_u,$ then we terminate.
	Otherwise, we continue \Wuu{as} follows.
	
	We define $\K_l=\{k\in \K\backslash\bigcup_{u=0}^{l-1}\K_u:\d_k^{(l,\infty)}>0\},$
	where all $\d_k^{(l,\infty)}$ are again defined \Wu{similarly}
	as in the proofs of Theorem~\ref{theo:SelfishMainTh}
	and Theorem~\ref{theo:MainTheoremRegularVariation}.
	\Wu{The} condition of Theorem \ref{theo:MainTheoremGaugeableExtension}
	\Wu{implies that}
	there \Wu{is} a regularly varying function $g(\cdot)$
	for $\K_l$ such that $G1')$-$G3')$ hold.
	We put $g_n^{(l)}=g(T_l(d^{(n)})),$ where \Wu{we put again}
	$T_l(d^{(n)})=\sum_{k\in \K_l}d_k^{(n)}$ for each $n\in \N.$
	
	Then, \Wuu{by the user optimality \eqref{def:WE} and
		conditions $G1')$-$G2'),$} we can similarly obtain \Wu{for each $k\in \K_l$} that 
	\[
	\tilde{L}_k^{(n)}, L_k^{*(n)}
	\in O\big(\max\{g_n^{(0)},\ldots,g_n^{(l)}\}\big),
	\]
	since $\tau_a(x)$ is non-decreasing, \eqref{eq:ImportantCondiAD}
	and thus
	$c_a(x)=x\tau_a'(x)+\tau_a(x)$ are asymptotically
	non-decreasing for each $a\in A.$ \Wu{By} comparing
	$g_n^{(l)}$ with 
	$\max\{g_n^{(0)},\ldots,g_n^{(l-1)}\},$ we can \Wu{again}
	obtain \Wu{with conditions
	\eqref{eq:ImportantCondiAD} and $G3')$} that
	\[
	\begin{split}
	\sum_{k\in \bigcup_{u=0}^{l}\K_u} C_k(\tilde{f}^{(n)})
	\approx &\sum_{k\in \bigcup_{u=0}^{l}\K_u} C_k(f^{*(n)})\\
	&\in \Theta\big(\max\{T_0(d^{(n)})g_n^{(0)}, \ldots,T_{l-1}(d^{(n)})g_n^{(l-1)},T_l(d^{(n)})g_n^{(l)}\}\big),
	\end{split}
	\]
	which validates the inductive assumption at step $m=l.$
	
	\Wu{Therefore,} the game is asymptotic decomposable and 
	asymptotically well designed.
\end{proof}


\begin{thebibliography}{10}
	
	\bibitem{amap2017}
	AMap.
	\newblock 2017 traffic analysis report for major cities in china, 2017, {\small http://report.amap.com/share.do?id=8a38bb86614afa0801614b0a029a2f79}.
	
	\bibitem{baidu2017}
	BaiduMap.
	\newblock 2017 q4 \& annual china urban research report, Jan. 2018,
	http://huiyan.baidu.com/reports/2017Q4\_niandu.html.
	
	\bibitem{Bingham1987Regular}
	N.~H. Bingham, C.~M. Goldie, and J.~L. Teugels.
	\newblock {\em Regular variation}.
	\newblock Cambridge University Press, 1987.
	
	\bibitem{Cole2003Pricing}
	Richard Cole, Yevgeniy Dodis, and Tim Roughgarden.
	\newblock Pricing network edges for heterogeneous selfish users.
	\newblock {\em Proceedings of the Annual {ACM} Symposium on the Theory of
		Computing}, pages 521--530, 2003.
	
	\bibitem{Cole2006How}
	Richard Cole, Yevgeniy Dodis, and Tim Roughgarden.
	\newblock How much can taxes help selfish routing?
	\newblock {\em Journal of Computer \& System Sciences}, 72(3):444--467, 2006.
	
	\bibitem{Colini2017a}
	R.~Colini-Baldeschi, R.~Cominetti, P.~Mertikopoulos, and M.~Scarsini.
	\newblock The asymptotic behavior of the price of anarchy.
	\newblock {\em WINE 2017, N.R. Devanur and P. Lu (Eds), Lecture Notes in
		Computer Science 10674}, pages 133--145, 2017.
	
	\bibitem{Colini2017b}
	R.~Colini-Baldeschi, R.~Cominetti, P.~Mertikopoulos, and M.~Scarsini.
	\newblock On the asymptotic behavior of the price of anarchy: Is selfich
	routing bad in highly congested network?
	\newblock {\em arXiv:1703.00927v1 [cs.GT]}, 2017.
	
	\bibitem{Colini2016}
	R.~Colini-Baldeschi, R.~Cominetti, and M.~Scarsini.
	\newblock On the price of anarchy of highly congested nonatomic network games.
	\newblock {\em SAGT 2016, M. Gairing and R. Savani (Eds), Lecture Notes in
		Computer Science 9928}, pages 117--128, 2016.
	
	\bibitem{Correa2005On}
	Jos{\'e}~R. Correa, Andreas~S. Schulz, and Nicolas~E. Stier-Moses.
	\newblock On the inefficiency of equilibria in congestion games. extended
	abstract.
	\newblock In {\em Integer Programming and Combinatorial Optimization,
		International IPCO Conference, Berlin, Germany, June 8-10, 2005,
		Proceedings}, pages 167--181. Springer, Lecture Notes in Computer Science
	3509, 2005.
	
	\bibitem{Dafermos1969}
	S.~C. Dafermos and F.~T. Sparrow.
	\newblock The traffic assignment problem for a general network.
	\newblock {\em Journal of Research of the U.S. National Bureau of Standards
		73B}, pages 91--118, 1969.
	
	\bibitem{Fleischer2004Tolls}
	L~Fleischer, K~Jain, and M~Mahdian.
	\newblock Tolls for heterogeneous selfish users in multicommodity networks and
	generalized congestion games.
	\newblock In {\em IEEE Symposium on Foundations of Computer Science, 2004.
		Proceedings}, pages 277--285, 2004.
	
	\bibitem{Harks2015Computing}
	Tobias Harks, Ingo Kleinert, Max Klimm, and Rolf~H M{\"o}hring.
	\newblock Computing network tolls with support constraints.
	\newblock {\em Networks}, 65(3):262--285, 2015.
	
	\bibitem{Nisan2007}
	Noam Nisan, Tim Roughgarden, Eva Tardos, and Vijay~V. Vaz.
	\newblock {\em Algorithmic game theory}.
	\newblock Cambridge University Press, 2007.
	
	\bibitem{BPR}
	Bureau of~Public~Roads.
	\newblock {\em Traffic assignment manual}.
	\newblock U.S. Department of Commerce, Urban Planning Division, Washington,
	D.C., 1964.
	
	\bibitem{O2016Mechanisms}
	Steven~J. O'Hare, Richard~D. Connors, and David~P. Watling.
	\newblock Mechanisms that govern how the price of anarchy varies with travel
	demand.
	\newblock {\em Transportation Research Part B Methodological}, 84:55--80, 2016.
	
	\bibitem{Papadimitriou2001Algorithms}
	Christos Papadimitriou.
	\newblock Algorithms, games, and the internet.
	\newblock In {\em International Colloquium on Automata, Languages, and
		Programming}, pages 1--3. Springer, Lecture Notes in Computer Science 2076,
	2001.
	
	\bibitem{Phang2004Road}
	Sock~Yong Phang and Rex~S. Toh.
	\newblock Road congestion pricing in singapore: 1975 to 2003.
	\newblock {\em Transportation Journal}, 43(2):16--25, 2004.
	
	\bibitem{Rosenthal1973A}
	Robert~W. Rosenthal.
	\newblock A class of games possessing pure-strategy nash equilibria.
	\newblock {\em International Journal of Game Theory}, 2(1):65--67, 1973.
	
	\bibitem{Roughgarden2000How}
	T~Roughgarden and E~Tardos.
	\newblock How bad is selfish routing?
	\newblock {\em Journal of the ACM}, 49:236--259, 2002.
	
	\bibitem{Roughgarden2001Designing}
	Tim Roughgarden.
	\newblock Designing networks for selfish users is hard.
	\newblock {\em Proceedings of Annual Symposium on Foundations of Computer
		Science}, 72(72):472--481, 2001.
	
	\bibitem{Roughgarden2002The}
	Tim Roughgarden.
	\newblock The price of anarchy is independent of the network topology.
	\newblock {\em Journal of Computer \& System Sciences}, 67(2):341--364, 2003.
	
	\bibitem{Roughgarden2005Selfish}
	Tim Roughgarden.
	\newblock {\em Selfish Routing and the Price of Anarchy}.
	\newblock The MIT Press, 2005.
	
	\bibitem{Roughgarden2015Intrinsic}
	Tim Roughgarden.
	\newblock Intrinsic robustness of the price of anarchy.
	\newblock {\em Journal of the ACM (JACM)}, 62(5):32, 2015.
	
	\bibitem{Roughgarden2004Bounding}
	Tim Roughgarden and Eva Tardos.
	\newblock Bounding the inefficiency of equilibria in nonatomic congestion
	games.
	\newblock {\em Games \& Economic Behavior}, 47(2):389--403, 2004.
	
	\bibitem{Roughgarden2007Introduction}
	Tim Roughgarden and Eva Tardos.
	\newblock Introduction to the inefficiency of equilibria.
	\newblock In {\em Algorithmic game theory}, pages 443--459. Cambridge
	University Press, Cambridge, 2007.
	
	\bibitem{Smith1979The}
	M.~J. Smith.
	\newblock The existence, uniqueness and stability of traffic equilibria.
	\newblock {\em Transportation Research Part B Methodological}, 13(4):295--304,
	1979.
	
	\bibitem{Wardrop1952ROAD}
	J.~G. Wardrop.
	\newblock Some theoretical aspects of road traffic research.
	\newblock {\em Proceedings of the Institution of Civil Engineers, Part II},
	1:325--378, 1952.
	
	\bibitem{Wu2017Selfishness}
	Zijun Wu, Rolf~H. M{\"o}hring, and Yanyan Chen.
	\newblock Selfishness need not be bad.
	\newblock {\em arXiv:1712.07464, submitted}, pages 1--51, 2017.
	
\end{thebibliography}
\end{document}